\documentclass[letterpaper,12pt]{article}

\usepackage[english]{babel}
\usepackage[utf8]{inputenc}
\usepackage[T1]{fontenc}
\usepackage[letterpaper,margin=1.25in]{geometry}
\usepackage{amsmath,amssymb,amsthm,amsfonts,dsfont}
\usepackage{array}
\usepackage{booktabs}
\usepackage{enumitem}
\usepackage{graphicx}
\usepackage[round]{natbib}
\usepackage[colorlinks=true,allcolors=blue]{hyperref}
\usepackage{setspace}
\usepackage{etoolbox}
\usepackage{needspace}
\usepackage{placeins}
\usepackage[format=plain,font={small,stretch=1},labelfont=bf,textfont=up]{caption}

\newtheorem{thm}{Theorem}[section]
\newtheorem{lem}[thm]{Lemma}
\newtheorem{assumption}[thm]{Assumption}
\newtheorem{prop}[thm]{Proposition}
\newtheorem{cor}[thm]{Corollary}
\theoremstyle{definition}
\newtheorem{defn}[thm]{Definition}
\newtheorem{ex}[thm]{Example}
\newtheorem{remark}[thm]{Remark}
\numberwithin{equation}{section}

\newcommand{\mr}{\mathbb{R}}
\newcommand{\mc}{\mathcal}

\newcommand{\ylevel}{\bar Y}
\newcommand{\yleveli}{\bar Y_i}
\newcommand{\SATE}{\mathrm{SATE}}

\newcommand{\one}{\mathds{1}}

\newcommand{\norm}[1]{\lVert#1\rVert}
\newcommand{\indep}{\mathrel{\raisebox{0.05em}{\rotatebox[origin=c]{90}{$\models$}}}}
\newcommand{\simiid}{\stackrel{\mathrm{iid}}{\sim}}
\newcommand{\convp}{\overset{p}{\to}}
\newcommand{\E}{E}

\newcommand{\Unif}{\mathrm{Unif}}

\DeclareMathOperator{\var}{Var}

\DeclareMathOperator{\cov}{Cov}

\DeclareMathOperator*{\argmin}{argmin}

\makeatletter
\renewcommand{\paragraph}{\@startsection{paragraph}{4}{\z@}%
  {\medskipamount}%
  {-\fontdimen2\font}%
  {\normalfont\normalsize\bfseries}}
\makeatother
\AtBeginEnvironment{abstract}{\singlespacing}
\AtBeginEnvironment{thebibliography}{\singlespacing}

\hypersetup{
pdftitle={Experimental Designs for Nonparametrically Balancing Many Covariates},
pdfauthor={Max Cytrynbaum},
pdfsubject={Finite-sample efficiency, covariate-adaptive randomization, matched pairs, and Gram-Schmidt walk},
pdfkeywords={experimental design, matched pairs, semiparametric efficiency, Gram-Schmidt walk}
}

\title{The Limits of Experimental Design: \\ Covariate Balance Beyond Low Dimension}
\author{Max Cytrynbaum\thanks{Department of Economics, Yale University.}}
\date{July 2026}

\begin{document}

\begingroup
\renewcommand{\thefootnote}{\fnsymbol{footnote}}
\maketitle
\endgroup
\setcounter{footnote}{0}

\begin{abstract}
We study how fast experimental designs can approach the semiparametric efficiency bound in finite samples, as measured by the excess variance of unadjusted treatment effect estimation.
We prove an impossibility theorem: under weak conditions, no design can approach the variance bound uniformly over smooth outcome models unless covariate dimension $d \ll \log n$.
Even in experiments with thousands of units, this permits only a handful of covariates.
Motivated by this, we propose new designs based on discrepancy minimization that instead attempt to control imbalances over restricted-complexity nonparametric function classes.
Such designs achieve fast rates to their corresponding restricted efficiency targets, permitting $d \ll n$ covariates in an additive nonparametric specification.
They can also be combined with matching to protect against unmodeled outcome variation.
In simulations calibrated to 12 published experiments, our designs reduce variance relative to matched pairs randomization in every empirical setting.
\end{abstract}



\clearpage

\section{Introduction}\label{section:introduction}

Balancing covariates is a central task of experimental design.
Available methods range from stratified randomization, in use at least since \citet{fisher1926}, to rerandomization and newer discrepancy minimization methods \citep{li2018asymptotic,harshaw2024gsw}.
Effective covariate balance can improve precision, reducing the size and cost of the experiment required to study a given causal effect.
This is well known to practitioners. 
For example, a survey of 200 randomized experiments in recent NBER working papers reported in \citet{cytrynbaum2022local} finds that 43\% used some form of stratification.

Recent theoretical work provides a strong justification for finely stratified designs like matched-pairs randomization.
Under appropriate conditions, such designs make simple unadjusted estimators semiparametrically efficient, automatically attaining the \citet{hahn1998} variance bound for estimating treatment effect parameters \citep{bai2021inference,bai2023efficiency}.
However, these guarantees rest on strong assumptions, requiring covariate differences within matched groups to vanish asymptotically. 
This may not be possible with the number of covariates encountered in practice.
Note that nearest-neighbor distances for $n$ samples in $d$ dimensions typically decrease at rate $n^{-1/d}$, so halving the average distance between matches requires $2^d$ times as many observations.

To see the practical implications of this curse of dimensionality, consider the OpenResearch Unconditional Income Study, a prominent experiment that randomized 3,000 participants to a universal basic income program, distributing \$40 million in unconditional cash transfers \citep{broockman2024income}.
To randomize treatments, researchers first formed matched triples using several dozen covariates, then assigned one participant in each triple to receive \$1,000 per month for three years.
For the efficiency results above to be relevant here, we would need to obtain near-perfect matches across several dozen covariates, using only $n=3000$ samples.

Motivated by this gap between theory and practice, we develop a new finite-sample efficiency theory for experimental design.
For an unadjusted estimator $\widehat{\theta}$ of the ATE and efficiency bound $V^*$, we study the \emph{excess variance} $\mc V_n\equiv n\var(\widehat\theta)-V^*$.
Let $\psi_i$ denote the covariates of unit $i$ and define the outcome model $m(\psi)\equiv E[(Y(1)+Y(0))/2|\psi]$.
Also let $Z_i\in\{-1,1\}$ be the centered treatment assignment with $E[Z] = 0$.
A calculation shows that excess variance is determined by in-sample imbalances in the outcome model:
\begin{equation}\label{equation:introduction-excess-variance}
\mc V_n = 4n \cdot E\bigg[\bigg(\frac{1}{n}\sum_{i=1}^n Z_i m(\psi_i)\bigg)^2\bigg].
\end{equation}

By studying such imbalances, we derive finite sample upper bounds on the excess variance $\mc V_n$, showing how rapidly different designs can approach $V^*$ for rough, smooth, and restricted-complexity nonparametric outcome models $m(\psi)$.
In several cases, we also obtain lower bounds establishing that these rates are minimax optimal.

More fundamentally, we prove an impossibility theorem showing that the curse of dimensionality illustrated above affects not only matching but any experimental design seeking to attain the classic variance bound $V^*$.
Under mild conditions, our results imply that no design can guarantee uniform asymptotic efficiency, even over smooth outcome models, beyond the regime $d=o(\log n)$, i.e.\ $d / \log n \to 0$ as $n \to \infty$.
For experiments with thousands of units, this permits only a few covariates.

This suggests an alternative design principle.
Rather than pursuing the classical efficiency bound $V^*$, we propose to balance lower-complexity nonparametric classes $\mc G$ that may still approximate the true outcome model well.
We develop new experimental designs based on this principle using nonparametric versions of the Gram-Schmidt walk (GSW) of \citet{bansal2019gram} and \citet{harshaw2024gsw}.

To measure their performance, we define the efficiency target $V^*+\Delta_{\mc G}$, where $\Delta_{\mc G}$ records outcome variation not captured by a restricted function class $\mc G$.
The lower complexity of such classes allows our designs to approach this target at much faster uniform rates and balance many more covariates than stratification can accommodate.
Designs targeting an additive nonparametric specification can attain the restricted efficiency target with $d=o(n)$ covariates, up to log factors.

We show these designs can also be combined with matching to preserve the fast rates above, while providing some protection against unmodeled outcome variation.
In an empirical application based on 12 randomized experiments recently published in economics, our preferred designs reduce estimator variance relative to classical matched pairs in every simulated empirical setting.

We make the following main contributions: 
\begin{enumerate}
	\item In Section~\ref{section:matching}, we study the finite-sample properties of matched pairs.
	We show that no matching scheme can make excess variance $\mc V_n$ vanish faster than $n^{-2/d}$, even over linear, hence infinitely smooth, outcome models.
	This extremely slow rate shows that matched pairs cannot realistically approach the efficiency bound in finite samples beyond the very low dimensional regime. 
	Despite this, we show matched pairs is minimax optimal in fixed dimension for rough outcome models, attaining $\mc V_n\asymp n^{-2\beta/d}$ over $\beta$-H\"older classes for $0<\beta\leq1$.
	Its slow convergence is therefore the price of robustness under very weak assumptions.

	\item Section~\ref{section:structure} introduces new randomization methods designed to exploit smoothness. 
	To do so, we develop a nonparametric Gram-Schmidt Walk that controls the worst-case imbalances over a chosen reproducing kernel Hilbert space.
	We exhibit a design that adapts without prior knowledge to every Sobolev smoothness order $s>0$ and attains rate $\mc V_n \asymp n^{-(2s/d)\wedge1}$, up to logarithmic factors. 
	This rate is minimax optimal when $2s\leq d$ and nearly parametric when $2s>d$.

	\item Section~\ref{section:structure} also establishes our main impossibility result: under weak conditions, no design can guarantee excess variance uniformly smaller than order $n^{-2s/d}$ over full-dimensional outcome models for any fixed smoothness $s>0$.
	This precludes uniform asymptotic efficiency for any experimental design outside the very-low dimensional regime $d=o(\log n)$, even for smooth outcome models.

	\item In view of this impossibility result, Section~\ref{section:restricted-efficiency} develops GSW designs using kernels tailored to lower-complexity nonparametric classes, including additive models of the form $g(\psi)=a+\sum_{j=1}^d g_j(\psi_j)$ and richer specifications with nonparametric bivariate interactions.
	The additive and bivariate designs approach their restricted efficiency targets at fast rates, allowing $d=o(n)$ and $d=o(\sqrt n)$ covariates, respectively, up to logarithmic factors.

	\item Finally, Section~\ref{section:restricted-efficiency} shows how to combine structured GSW with matching, retaining the fast rates above for low-complexity outcome components while also providing protection against unmodeled variation.
	We develop a conservative design-based variance estimator for such methods that performs well in our simulations and empirical application.
\end{enumerate}

\subsection{Related Literature}

This paper is motivated by a recent literature showing that finely stratified randomization makes simple estimators attain the \citet{hahn1998} variance bound for the ATE.
Early contributions include \citet{bai2021inference}, \citet{bai2020pairs}, and \citet{cytrynbaum2022local}.
Building on \citet{armstrong2022}, \citet{bai2023efficiency} further show that this classical efficiency bound remains valid for general experimental designs and is attained by fine stratification in a GMM setting.
Related work includes \citet{imai2008}, \citet{fogarty2018}, \citet{pashley2021}, and \citet{kapelner2022}, among others.

Our asymptotic inefficiency result for matched pairs is an experimental analogue of the classical result of \citet{abadie2006large}, who show that slow nearest-neighbor convergence rates prevent matching from effectively debiasing observational estimates.
A related result is \citet{wang2022rerandomization}, who show a $d=o(\log n)$ threshold for balancing linear functions using rerandomization.
For comparison, our designs can balance additive nonparametric function spaces with $d = o(n)$ up to logarithms.  

Our finite-sample perspective is related to that in \citet{kallus2018}, who also notes an analogue of Equation~\eqref{equation:introduction-excess-variance}. 
His analysis leads to semi-deterministic designs, some of which are NP-hard to compute. 
By contrast, we develop a full efficiency theory with explicit convergence rates of finite sample variances to both classical and newly introduced variance bounds.
We show that fast convergence rates can be achieved over restricted-complexity nonparametric spaces by computationally tractable randomized designs. 
See the remarks in the text for a detailed comparison. 

Our new designs in Section \ref{section:structure} fundamentally rely on the Gram-Schmidt walk vector balancing algorithm of \citet{bansal2019gram}.
\citet{harshaw2024gsw} first studied this algorithm in an experimental setting, proving a design-based oracle inequality for the mean squared error of an unadjusted estimator and using the algorithm to balance linear functions.
\citet{harshaw2021dissertation} first observed that the design can be kernelized, extending the original oracle inequality to this setting.
\citet{chen2026nonlinear} also suggest using GSW to balance nonlinear functions of the covariates.

To the best of our knowledge, we are the first to propose a finite sample efficiency theory for experimental design, using minimax analysis of the excess variance $\mc V_n$ as the benchmark.
The behavior of minimax rates over smoothness and restricted-complexity classes is a central organizing theme in nonparametric statistics \citep{stone1982optimal,yang1999information,gine2016mathematical}.
Our use of Sobolev norms to quantify smoothness also follows a long tradition in this literature.
For example, \citet{vdv2011information} show that convergence rates for Gaussian process regression with Mat\'ern priors are driven by a match between prior regularity and the Sobolev smoothness of the truth. 
\citet{fischer2020sobolev} derive learning rates for kernel ridge regression in Sobolev norms.
See also \citet{kennedy2024minimax} for minimax rates for heterogeneous treatment effect estimation in observational studies.

The new designs in Section~\ref{section:restricted-efficiency} target balance over restricted function spaces. 
Our main examples involve additive and bivariate nonparametric models motivated by functional ANOVA \citep{stone1985additive,hastie1986gam,sobol1993}.

Other prominent approaches to covariate balance include rerandomization \citep{morgan2012,li2018asymptotic} and optimization-based designs \citep{kasy2016,krieger2019}.
A complementary literature studies variance reduction by regression adjustment with moderately high-dimensional covariates \citep{bloniarz2016lasso,wager2016highdimensional,lei2021regression,lu2025debiased}.
Our work can be viewed as a randomization-based analogue of such results.

Finally, our variance estimator combines the design-based quadratic bound approach of \citet{harshaw2026optimized} with the collapsed-strata and pairs-of-pairs methods of \citet{hansen1953}, \citet{abadie2008}, and \cite{bai2021inference}.

\section{Setup}\label{section:setup}

Consider an experiment with $n$ units.
Denote potential outcomes $Y_i(a)$ for $a\in\{0,1\}$ and $i\in[n]$, letting $[r]=\{1,\ldots,r\}$ for any positive integer $r$.

\begin{assumption}\label{assumption:sampling}
Observations $\bigl(\psi_i,Y_i(0),Y_i(1)\bigr)\simiid P$ for $i\in[n]$.
The covariate $\psi$ is supported on $[0,1]^d$ and has a density $p$ satisfying $0<\underline p\leq p(\psi)\leq\overline p<\infty$ for every $\psi\in[0,1]^d$.
For $a\in\{0,1\}$, $E_P[Y(a)^2]<\infty$.
\end{assumption}

Unless stated otherwise, every data-generating law $P$ in every model class $\mc P$ satisfies Assumption~\ref{assumption:sampling}.
We are interested in estimating the average treatment effect $\theta(P) = E_P[Y(1)-Y(0)]$.
Write $\psi_{1:n}=(\psi_1,\ldots,\psi_n)$ for the realized covariates.
For a realized treatment assignment $A_i\in\{0,1\}$, the observed outcome is $Y_i=Y_i(A_i)$.
For convenience, in what follows, we work with the coding $Z_i=2A_i-1\in\{-1,1\}$ and let $Z=(Z_1,\ldots,Z_n)$ be the full assignment vector.
\begin{defn}[Admissible Designs]\label{definition:admissible-designs}
A conditional assignment law $\sigma = Z | \psi_{1:n}$ is admissible $\sigma \in \mathcal D_n$ if $Z\indep(Y_i(0),Y_i(1))_{i=1}^n|\psi_{1:n}$ and $E[Z_i|\psi_{1:n}]=0$ for every $i\in[n]$.
\end{defn}

For example, admissible designs include iid assignment, matched-pairs and other forms of stratified randomization, rerandomization with symmetric acceptance rules, as well as the Gram-Schmidt walk designs introduced in Sections~\ref{section:structure} and~\ref{section:restricted-efficiency}.
We use Horvitz-Thompson estimation 
\begin{equation}\label{equation:ht-estimator}
\widehat\theta = \frac{2}{n}\sum_{i=1}^nZ_iY_i.
\end{equation}
When exactly $n/2$ units are treated, as in matched pairs with even $n$, Equation~\eqref{equation:ht-estimator} is the usual difference in means.
Denote $\tau=Y(1)-Y(0)$ for the individual treatment effect and let $v_a^2(\psi)=\var(Y(a)|\psi)$ for $a\in\{0,1\}$.
Let $\tau(\psi)=E[\tau|\psi]$. 
At treatment probability one half, the semiparametric efficiency bound of \citet{hahn1998} for $\theta(P)$ is 
\begin{equation}\label{equation:efficiency-bound}
V^*(P) = \var\bigl(\tau(\psi)\bigr) +2E\bigl[v_1^2(\psi)\bigr] +2E\bigl[v_0^2(\psi)\bigr].
\end{equation}
This bound was originally derived for superpopulation settings with iid data.
Building on \citet{armstrong2022}, \citet{bai2023efficiency} show it holds over a broad class of experimental designs with dependent treatment assignments.

In what follows, we study the gap between the finite sample variance $n\var_{P,\sigma}(\widehat\theta)$ for $\sigma \in \mc D_n$ and the variance bound $V^*(P)$.
To that end, define the \emph{excess variance} 
\begin{equation}\label{equation:excess-variance}
\mc V_n(\sigma,P)\equiv n\var_{P,\sigma}(\widehat\theta)-V^*(P).
\end{equation}
A sequence of designs $\sigma_n \in \mc D_n$ is pointwise asymptotically efficient at $P$ if the excess variance $\mc V_n(\sigma_n,P)\to0$ as $n \to \infty$.
Previous results have established this property for matched pairs designs and more general forms of fine stratification \citep{bai2021inference,bai2023efficiency,cytrynbaum2022local}.
Our results in Section~\ref{section:structure} also add to this list, providing a new family of pointwise efficient designs based on nonparametric GSW.

This classical formulation, however, holds $P$, and hence $d$, fixed as $n$ grows and does not characterize how quickly the variance bound is approached uniformly over a model class.
We show in what follows that this can lead to a very poor approximation to finite sample performance. 
Instead, we propose to study finite-sample upper and lower bounds on excess variance $\mc V_n$ that make the dependence on sample size, covariate dimension, and outcome-model complexity explicit.
The following simple result identifies excess variance exactly with imbalances in $m(\psi) = E[(Y(1) + Y(0))/2 | \psi]$.

% PROOF-ID: variance-decomposition
\begin{prop}[Excess Variance]\label{proposition:variance-decomposition}
For every $\sigma\in\mathcal D_n$, the excess variance is 
\begin{equation}\label{equation:variance-decomposition}
\mc V_n(\sigma,P) = 4n \cdot E \bigg[  \Big( n^{-1} \sum_{i=1}^nZ_im(\psi_i) \Big)^2  \bigg].
\end{equation}
\end{prop}
This result recasts finite-sample efficiency as a function-balancing problem, with performance determined by the mean squared imbalance $n^{-1}\sum_{i=1}^nZ_im(\psi_i)$.
This equivalence was also previously noted in a slightly different form by \citet{kallus2018}.

Since the distribution $P$ is unknown at design time, a design cannot target $m(\psi)$ directly and must instead seek to control imbalance over a class of plausible models $P \in \mc P$.
Protecting against a very broad class may cause excess variance $\mc V_n(\sigma,P)$ to vanish too slowly for meaningful efficiency gains in finite samples.
Next, we show that matched pairs exhibits just such an extreme case of this tradeoff.

\section{The Limits of Matched Pairs Designs}\label{section:matching}

Matched pairs protects against broad classes of weakly regular outcome models, but this robustness comes at a cost.
In particular, no matched-pairs design can guarantee uniform asymptotic efficiency, even over linear outcome models, beyond the very-low dimensional regime $d=o(\log n)$.

\subsection{A Lower Bound for Matching}\label{subsection:matching-lower-bound}

Recall that a matched-pairs design first partitions the $n$ units into $n/2$ unordered pairs $(i_p,j_p)$, then randomly assigns one unit in each pair to treatment and the other to control with equal probability, independently across pairs.
The pairing itself may be any measurable function of the realized covariates $\psi_{1:n}$ and exogenous randomness $U_M$.
We write $\mathcal M_n$ for the resulting class of designs, with a given instance $M \in \mc M_n$. 
Matched pairs is clearly admissible, $\mc M_n \subseteq \mc D_n$.
The definition implies $Z_{i_p}=-Z_{j_p}$ within each pair. 
Independence between pairs and Proposition~\ref{proposition:variance-decomposition} then imply
\begin{equation}
\mc V_n(M,P)=\frac{4}{n}E \bigg[\sum_{p=1}^{n/2}\bigl(m(\psi_{i_p})-m(\psi_{j_p})\bigr)^2\bigg].
\end{equation}
Under Lipschitz continuity of $m(\psi)$, the within-pair differences $|m(\psi_{i_p})-m(\psi_{j_p})|$ are bounded by a constant multiple of the corresponding covariate distances $\|\psi_{i_p}-\psi_{j_p}\|_2$.
Pointwise asymptotic efficiency results therefore typically show $\mc V_n(M_n,P)\to0$ for each fixed law $P$ under a tight-matching condition of the form \citep{bai2021inference}:
\begin{equation}\label{equation:tight-matching}
\frac{1}{n}\sum_{p=1}^{n/2} \|\psi_{i_p}-\psi_{j_p}\|_2^2 \convp0.
\end{equation}
See also \citet{cytrynbaum2022local} for algorithms guaranteeing such tight-matching conditions for general forms of fine stratification in fixed dimension.
By holding $d$ fixed, these asymptotic guarantees can obscure how quickly matching deteriorates with increasing covariate dimension in finite samples: excess variance may vanish asymptotically while remaining large at empirically relevant sample sizes.
The next result makes this limitation precise, even for linear outcome models.
To state the result, let $\mathcal P_{\mathrm{lin}}$ contain the laws satisfying Assumption~\ref{assumption:sampling} whose outcome model has the form $m(\psi)=a+\gamma'\psi$, with $\|\gamma\|_2\leq1$.

% PROOF-ID: matching-linear-lower-bound
\begin{thm}[Matching Slow Rate]\label{theorem:matching-lower-bound}
For every even $n$ and every $d\geq1$,
\begin{equation}\label{equation:matching-lower-bound}
\inf_{M\in\mathcal M_n} \sup_{P\in\mathcal P_{\mathrm{lin}}}\mc V_n(M,P) \geq \frac{1}{2\pi e}n^{-2/d}.
\end{equation}
\end{thm}

Theorem~\ref{theorem:matching-lower-bound} applies to every pairing rule, not only optimal squared Euclidean distance matching as studied in \cite{bai2021inference}.
Changing the matching criterion therefore cannot improve this very slow excess variance rate.

Next, we use this finite-sample result to obtain alternative asymptotics for matched-pairs randomization in which $d=d_n$ may vary with $n$.
This allows us to formalize the sense in which matching is asymptotically inefficient beyond the very low-dimensional regime $d_n\ll\log n$.

% PROOF-ID: matching-asymptotic-inefficiency
\begin{cor}[Asymptotic Inefficiency]\label{corollary:matching-inefficiency}
If $d_n\geq c\log n$ for some fixed $c>0$, then 
\begin{equation}\label{equation:matching-inefficiency}
\liminf_{n\to\infty} \inf_{M\in\mathcal M_n} \sup_{P\in\mathcal P_{\mathrm{lin}}}\mc V_n(M,P) \geq \frac{1}{2\pi e}e^{-2/c} >0.
\end{equation}
\end{cor}

Equivalently, sample size $n$ may be exponentially larger than covariate dimension $d_n$, yet the worst-case excess variance still remains bounded away from zero.
The same curse of dimensionality drives the non-vanishing bias of matching estimators in observational studies \citep{abadie2006large} and slow convergence rates for nearest-neighbor methods in nonparametric regression \citep{gyorfi2002}.

Figure~\ref{figure:matching-reversal-simulation} illustrates the finite-sample inefficiency suggested by both previous results.
Although recording more covariates steadily lowers the corresponding efficiency bound, the finite sample variance of matched-pairs eventually rises in both panels.
Linear GSW as in \cite{harshaw2024gsw} avoids this reversal under a linear outcome model and rerandomization delays it, but both offer almost no gain under an adversarial nonlinear additive model.
The additive nonparametric GSW design developed in Section~\ref{section:restricted-efficiency} below performs well in both settings.

\begin{figure}[!htbp]
\centering
\includegraphics[width=\textwidth]{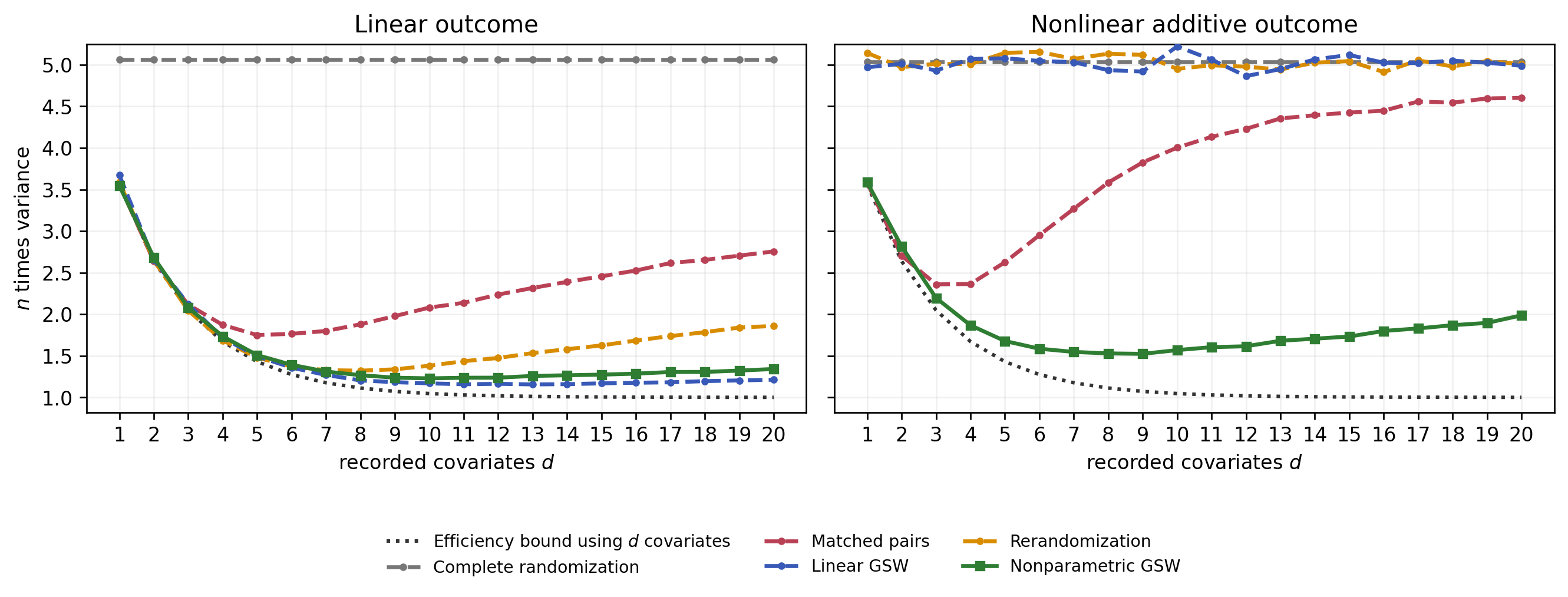}
\caption{Variance for $n=240$ units with independent uniform covariates. The left panel uses a normalized linear outcome model. The right panel replaces each linear coordinate with a nonlinear main effect while preserving the same coordinate-level variance profile. In both panels, the coordinate coefficients are proportional to $0.8^{j-1}$ and the dotted line is the efficiency bound using the first $d$ covariates. The secondary online appendix gives the exact models and simulation protocol.}
\label{figure:matching-reversal-simulation}
\end{figure}
\FloatBarrier

\subsection{Matching Targets Rough Outcome Classes}\label{subsection:rough-classes}

For $0<\beta\leq1$, define the H\"older coefficient $[m]_\beta=\sup_{x\neq y}|m(x)-m(y)|/\norm{x-y}_2^\beta$.
In this section, we show that slow excess variance rates of form $n^{-2\beta/d}$ are minimax optimal for fixed $d < \infty$ when the outcome model is $\beta$-H\"older continuous.
We prove that stable matchings generically achieve these rates, even adaptively over $\beta\in(0,1]$, but cannot exploit additional smoothness beyond the Lipschitz endpoint $\beta=1$.

\medskip

\emph{Optimal Matching.} We describe a broad class of pairing rules that attain the excess variance rate $n^{-2\beta/d}$ for fixed $d$.
Optimal Euclidean matching, as in \citet{bai2021inference}, minimizes $\sum_{p=1}^{n/2}\|\psi_{i_p}-\psi_{j_p}\|_2^2$ over all pairings.
More generally, an experimenter may specify any symmetric matching cost $c(x,y)$ and use \citet{derigs1988} algorithm to compute optimal pairs 
\begin{equation}
\min_{(i_p,j_p)_{p=1}^{n/2}}\sum_{p=1}^{n/2}c(\psi_{i_p},\psi_{j_p}).
\end{equation}

\emph{Stable Matching.}
In fact, global optimality is stronger than necessary.
A weaker notion is 2-swap stability, which requires that no two pairs can be rematched to reduce their combined cost.
For any two matched pairs $\{i,j\}$ and $\{k,l\}$, require that
\begin{equation}\label{equation:swap-stability}
c(\psi_i,\psi_j)+c(\psi_k,\psi_l) \leq \min\left\{ c(\psi_i,\psi_k)+c(\psi_j,\psi_l), c(\psi_i,\psi_l)+c(\psi_j,\psi_k) \right\}.
\end{equation}

We show that every 2-swap-stable pairing attains the excess variance rate $n^{-2\beta/d}$ for fixed $d$ whenever the matching cost satisfies a weak geometric condition, requiring $c(x,y)$ to be comparable above and below to Euclidean distance.

\begin{assumption}[Matching Cost]\label{assumption:matching-cost}
For some $q>0$ and constants $0<\underline\lambda\leq\overline\lambda<\infty$, the symmetric cost $c$ satisfies $\underline\lambda\|x-y\|_2^q\leq c(x,y)\leq\overline\lambda\|x-y\|_2^q$ for every $x,y\in[0,1]^d$.
\end{assumption}

We define $\Lambda=(\overline\lambda/\underline\lambda)^{1/q}$ as the condition number of the matching cost, which may depend on $d$.
It measures the loss from translating control of the matching cost $c(x,y)$ into control of Euclidean distance.
Thus $\Lambda=1$ for any power of Euclidean distance, while $\Lambda=\sqrt d$ for $c(x,y)=\|x-y\|_1^q$ or $c(x,y)=\|x-y\|_\infty^q$ for any $q>0$, for example.

Let $\mathcal P_\beta$ contain the laws satisfying Assumption~\ref{assumption:sampling} whose outcome model has H\"older coefficient $[m]_\beta\leq1$.
At $\beta=1$, this is the ordinary Lipschitz class, while smaller $\beta$ allow progressively rougher outcome functions.

% PROOF-ID: stable-matching-upper-bound
\begin{thm}[Stable Matching]\label{theorem:stable-matching}
Fix an even $n\geq2$ and $0<\beta\leq1$, and suppose $d\geq3$.
Impose Assumption~\ref{assumption:matching-cost} and let $M_c$ be the matched-pairs design for any pairing that is 2-swap-stable for $c(x,y)$.
There is $C > 0$ depending only on $\beta$ such that
\begin{equation}\label{equation:stable-matching-high-d}
\sup_{P\in\mathcal P_\beta}\mc V_n(M_c,P) \leq C(d\Lambda^2)^{2\beta} \cdot n^{-2\beta/d}.
\end{equation}
\end{thm}

This is a finite-sample bound, but we can also interpret it asymptotically. 
For example, in the low-dimensional regime with $d<\infty$ fixed while $n\to\infty$, any stable matching as above attains the excess variance rate $n^{-2\beta/d}$.
This holds simultaneously for every $\beta\in(0,1]$, showing that matching adapts to the unknown H\"older smoothness of $m(\psi)$.
Next, we establish a corresponding finite sample lower bound over all admissible designs, showing that no design can improve this rate uniformly over the H\"older class in the regime with low fixed dimension $d < \infty$.

% PROOF-ID: design-agnostic-lower-bound
\begin{thm}[H\"older Lower Bound]\label{theorem:design-agnostic-lower-bound}
For every $0<\beta\leq1$, there is a constant $C>0$ depending only on $(\beta, d)$ such that, for every $n\geq2$ and every $d\geq1$,
\begin{equation}\label{equation:design-agnostic-lower-bound}
\inf_{\sigma\in\mathcal D_n} \sup_{P\in\mathcal P_\beta}\mc V_n(\sigma,P) \geq C n^{-2\beta/d}.
\end{equation}
\end{thm}

The lower bound in Theorem~\ref{theorem:design-agnostic-lower-bound} is completely design agnostic.
It applies not only to existing procedures, such as stratified randomization, symmetric rerandomization \citep{li2018asymptotic}, and optimization-based balancing \citep{kallus2018}, but to every possible experimental design under the mild admissibility requirements above.
Combining Theorems~\ref{theorem:stable-matching} and~\ref{theorem:design-agnostic-lower-bound} immediately shows stable matched pairs is minimax rate optimal over rough H\"older classes in fixed low dimension $d < \infty$.

% PROOF-ID: rough-class-minimax
\begin{cor}[Minimaxity Over Rough Spaces]\label{corollary:minimax-matching}
Let $0<\beta\leq1$ and fix dimension $3\leq d<\infty$. 
Then as $n\to\infty$, stable matching achieves minimax rate
\begin{equation}\label{equation:minimax-matching}
\inf_{\sigma\in\mathcal D_n} \sup_{P\in\mathcal P_\beta}\mc V_n(\sigma,P) \asymp n^{-2\beta/d}.
\end{equation}
\end{cor}

Matching therefore provides minimax protection adaptively over the rough H\"older scale $\beta \in (0, 1]$.
Despite this, Theorem~\ref{theorem:matching-lower-bound} also shows that no matched-pairs design can improve upon the rate $n^{-2/d}$ attained at $\beta=1$ even over the infinitely smooth linear class $\mathcal P_{\mathrm{lin}}$.
Thus matching cannot exploit smoothness beyond Lipschitz order.
Motivated by this, the next section develops new GSW designs that can adaptively exploit such higher order smoothness conditions, while also retaining good performance for rough outcome models.

\begin{remark}[Design Agnostic Bounds]
To the best of our knowledge, Corollary~\ref{corollary:minimax-matching} is the first design-agnostic minimax rate for convergence of finite sample variance to the \cite{hahn1998} variance bound under covariate-balancing randomization. 
To prove the lower bound, we place a prior on an appropriate hard subfamily of $\mathcal P_\beta$.
We show that, for every treatment allocation $Z$, the squared imbalance $(n^{-1}\sum_{i=1}^nZ_im(\psi_i))^2$, averaged over this family, is bounded below by unavoidable local contributions for a constant fraction of sampled units.
Since this bound holds pointwise in $Z$, it remains valid after averaging over any admissible assignment distribution $\sigma\in\mathcal D_n$, providing a design-agnostic lower bound.
\end{remark}

\begin{remark}
\citet{kallus2018} proves that optimal matched pairs randomization minimizes a related conditional imbalance criterion uniformly over Lipschitz outcome models.
In his framework, the least-favorable outcome model $m$ may depend on both realized covariates $\psi_{1:n}$ and the designer's candidate treatment allocation $Z$.
By contrast, here we study the usual statistical minimax risk over a fixed superpopulation law.
He does not study upper or lower bounds on the convergence rate of the resulting excess variance, as we do here.
\end{remark}

% !TeX root = ../main.tex
\section{Balancing Smooth Functions}\label{section:structure}

Motivated by the inefficiency results for matched pairs above, we develop nonparametric discrepancy minimization designs that can exploit higher-order smoothness, building on the Gram-Schmidt walk of \cite{bansal2019gram} and \cite{harshaw2024gsw}.
We show that a certain kernelized GSW specification can achieve excess variance rate $n^{-2s/d}$ adaptively over smoothness parameters $0 < s \leq d / 2$ and $n^{-1}$ otherwise, up to logarithmic factors.
Unlike matched pairs, such global balancing schemes continue to improve under additional orders of Sobolev smoothness.

Despite the attractive properties of nonparametric GSW, we also establish our main impossibility result, showing that no admissible design can guarantee uniform convergence to the efficiency bound $V^*(P)$, even over smooth outcome models, outside the very-low dimensional regime $d=o(\log n)$.

%However, these designs remain subject to the curse of dimensionality. 
%When $d_n \ge c \log n$, the $n^{-2s/d}$ rate above no longer guarantees that excess variance vanishes asymptotically.
%Indeed, we prove a design-agnostic lower bound showing that this barrier is unavoidable over full-dimensional Sobolev classes.
%Motivated by this, Section~\ref{section:restricted-efficiency} extends the designs in this section to balance nonparametric function classes with additive and other low-order interaction structure beyond low dimensions.

\subsection{Sobolev Spaces and RKHS}\label{subsection:sobolev-models}

For nonparametric GSW, it will be convenient to parameterize regularity using Sobolev spaces.
For a multi-index $\alpha=(\alpha_1,\ldots,\alpha_d)$, write $|\alpha|=\sum_{j=1}^d\alpha_j$ and let $D^\alpha F$ denote the weak derivative of a function $F$ on $\mr^d$.
Weak derivatives extend ordinary differentiation to functions whose rates of change exist only in an integrated sense, agreeing with ordinary derivatives whenever they exist.
Recall that $L^2(\mr^d)=\{F:\mr^d\to\mr:\int_{\mr^d}|F(\psi)|^2d\psi<\infty\}$ is the space of square-integrable functions.
For integers $k\geq1$, define the Sobolev space
\begin{equation}\label{equation:sobolev-whole-space}
H^k(\mr^d) \equiv \left\{ F\in L^2(\mr^d): D^\alpha F\in L^2(\mr^d) \text{ for every }|\alpha|\leq k \right\}.
\end{equation}
The norm on $H^k(\mr^d)$ can be expressed in terms of integrated derivatives 
\begin{equation}\label{equation:sobolev-whole-space-norm}
\norm{F}_{H^k(\mr^d)}^2 \asymp \sum_{|\alpha|\leq k} \int_{\mr^d}|D^\alpha F(\psi)|^2d\psi.
\end{equation}
In contrast to the H\"older classes in Section~\ref{section:matching}, which bound local changes uniformly across the covariate space, the Sobolev norm controls the aggregate size of these derivatives.
It therefore imposes an integrated global budget on local fluctuations.
Increasing the Sobolev order $k$ requires more weak derivatives to exist and to be square integrable, producing a progressively more regular class.
The definition extends to smoothness of every real order $s>0$.
Recall the Fourier transform $\widehat F(\omega)=(2\pi)^{-d/2}\int_{\mr^d}e^{-i\omega'x}F(x)dx$.
Then $H^s(\mr^d)$ consists of functions $F\in L^2(\mr^d)$ for which
\begin{equation}\label{equation:sobolev-spectral-norm}
\norm{F}_{H^s(\mr^d)}^2 \equiv \int_{\mr^d}|\widehat F(\omega)|^2(1+\norm{\omega}_2^2)^sd\omega<\infty.
\end{equation}
At integer orders $s=k$, this yields the same space as the derivative-based definition above.
Increasing the smoothness parameter $s$ requires the high-frequency components of $F$ to decay faster.
The use of Sobolev spaces is convenient rather than essential: the supplementary online appendix shows how the corresponding rates for nonparametric GSW over Sobolev spaces extend to H\"older classes $0<\beta\leq1$ studied above.

\medskip

\emph{RKHSs.} We use reproducing kernel Hilbert space (RKHS) methods to construct a feasible balancing criterion that automatically adapts to the unknown Sobolev order of the outcome model $m(\psi)$.
Let $W$ be a continuous positive semidefinite kernel, so that for every $r\geq1$ and any $(x_i)_{i=1}^r \subseteq \mathbb{R}^d$, the Gram matrix $[W(x_i,x_j)]_{i,j=1}^r$ is positive semidefinite.
Its RKHS $\mc H_W$ is a Hilbert space of functions with the reproducing property $h(\psi)=\langle h,W(\psi,\cdot)\rangle_{\mc H_W}$ for any $h\in\mc H_W$, where $\langle\cdot,\cdot\rangle_{\mc H_W}$ denotes its inner product.
By Cauchy-Schwarz, the reproducing property implies $|h(\psi)|\leq W(\psi,\psi)^{1/2}\norm{h}_{\mc H_W}$, so point evaluation is a continuous linear functional on $\mc H_W$.

The Mat\'ern family $W_d^{\mathrm{Mat},\nu}$ of kernels is indexed by the smoothness parameter $\nu>0$.
Its simplest member, obtained at $\nu=1/2$, is the exponential kernel $W_d^{\mathrm{Mat},1/2}(\psi,\psi')=\exp(-\norm{\psi-\psi'}_2)$.
A formula for general $\nu$ is given in the secondary online appendix.
Write $\mc H_d^{\mathrm{Mat},\nu}\equiv\mc H_{W_d^{\mathrm{Mat},\nu}}$ for the associated RKHS.
The following standard correspondence links these RKHSs to the Sobolev spaces defined above.

\begin{prop}[Mat\'ern-Sobolev Correspondence]\label{proposition:matern-sobolev-equivalence}
For every $d\geq1$ and $\nu>0$,
\begin{equation}\label{equation:matern-sobolev-equivalence}
\mc H_d^{\mathrm{Mat},\nu} = H^{d/2+\nu}(\mr^d), \qquad \norm{f}_{\mc H_d^{\mathrm{Mat},\nu}} \asymp_{d,\nu} \norm{f}_{H^{d/2+\nu}(\mr^d)}.
\end{equation}
\end{prop}
In particular, every Sobolev space $H^s(\mr^d)$ with $s>d/2$ is a Mat\'ern RKHS, obtained by taking $\nu=s-d/2$.
This correspondence is standard \citep{wendland2005scattered,kanagawa2018gaussian}.
We rederive it in the online appendix and record a uniform bound on the norm-equivalence as $\nu$ approaches zero, which is used in our later results.

%When $s\leq d/2$, $H^{s}(\mr^d)$ is not itself an RKHS under the Sobolev norm \citep{adams2003sobolev}.
%We refer to the cases $s>d/2$, $s=d/2$, and $s<d/2$ as the supercritical, critical, and subcritical regimes, respectively.
%Below, we choose Mat\'ern RKHSs whose Sobolev orders approach $d/2$ from above.
%Although each design RKHS is supercritical, the resulting designs can exploit outcome regularity in all three regimes by balancing suitable RKHS approximations to the outcome function $m(\psi)$.

%By classical Sobolev embedding, $\mc H_d^{\mathrm{Mat},\nu}\subset C^q([0,1]^d)$ for every nonnegative integer $q<\nu$ \citep{adams2003sobolev}.
%The RKHS also imposes a square-summability budget on the derivatives, imposing additional regularity through a global budget on local fluctuations.

\subsection{Controlling Imbalances for Smooth Functions}\label{subsection:unit-gsw}

Proposition~\ref{proposition:variance-decomposition} shows that excess variance is determined by the expected square of the outcome-model imbalance $n^{-1}\sum_{i=1}^nZ_im(\psi_i)$.
Since $m(\psi)$ is unknown at design time, the design cannot balance it directly.
We can form a feasible design criterion by instead controlling the worst-case imbalance over the unit ball of a chosen RKHS $\mc H_W$.
Consider an allocation $Z\in\{-1,1\}^n$ and let $W_n=[W(\psi_i,\psi_j)]_{i,j=1}^n$ be the corresponding kernel matrix.
By the reproducing property, one calculates imbalance
\begin{equation}\label{equation:rkhs-uniform-imbalance}
\mc I_W(Z) \equiv \sup_{\norm{g}_{\mc H_W}\leq1} \left(\frac1n\sum_{i=1}^nZ_ig(\psi_i)\right)^2 = \frac1{n^2}Z'W_nZ.
\end{equation}
For $W=W_d^{\mathrm{Mat},\nu}$, Proposition~\ref{proposition:matern-sobolev-equivalence} identifies this criterion with the worst-case squared imbalance over a norm ball in $H^{d/2+\nu}(\mr^d)$.
We therefore seek an admissible assignment law $\sigma \in \mc D_n$ under which the quadratic form $\mc I_W(Z)$ is small.
%For intuition, note that under exact treatment balance $\sum_i Z_i = 0$, $\mc I_W(Z)$ is proportional to the empirical squared maximum mean discrepancy (MMD$^2$) between the treatment and control covariate distributions \citep{gretton2012kernel}.
%Thus, minimizing RKHS imbalance makes these empirical distributions difficult to distinguish using functions in $\mc H_W$.

One can obtain such a design using the Gram-Schmidt walk (GSW) of \citet{bansal2019gram}, as adapted to experimental design by \citet{harshaw2024gsw}. This provides a computationally tractable assignment law that keeps $\mc I_W(Z)$ small while retaining the randomization needed for robustness and inference.

\emph{GSW Algorithm.} The Gram-Schmidt walk of \citet{bansal2019gram} takes as input an $n\times n$ positive semidefinite matrix $\Gamma_n$ with diagonal entries at most one.
It begins with the infeasible fractional treatment allocation $z^{(0)}=0\in[-1,1]^n$.
At each step $t\geq1$, it selects a pivot index $p_t$, which it retains until $z_{p_t}$ reaches $-1$ or $1$, and chooses an update direction $u$ minimizing $u'\Gamma_nu$, subject to $u_{p_t}=1$ and $u_j=0$ for coordinates that are already integral.
It randomizes between the largest feasible positive and negative steps in direction $u$, with probabilities chosen so the update has conditional mean zero.
At least one fractional coordinate reaches the boundary after each update, so the procedure terminates in finitely many steps in a feasible treatment allocation $Z \in \{\pm 1\}^n$.
See the secondary online appendix for a more detailed discussion.

\emph{Kernelization.}
\citet{harshaw2021dissertation} originally observed that GSW can be kernelized by constructing its input matrix $\Gamma_n$ from a kernel matrix.
We use the \emph{design kernel} $K=1+W$, where $W$ is a continuous positive semidefinite kernel satisfying $\sup_\psi W(\psi,\psi)\leq1$.
The added constant ensures that the associated RKHS can represent arbitrary constant levels of the outcome model $m$.
By the sums-of-kernels theorem \citep{paulsen2016rkhs}, $\mc H_K$ consists of functions $g=a+w$, where $a\in\mr$ and $w\in\mc H_W$, with norm $\norm{g}_{\mc H_K}^2=\min_{g=a+w}(a^2+\norm{w}_{\mc H_W}^2)$.
Given covariates $\psi_{1:n}$ and a robustness parameter $\varphi\in(0,1)$, define the positive definite matrix
\begin{equation}\label{equation:kernel-gsw-gram}
\Gamma_n = \varphi I_n + \frac{1-\varphi}{2} \bigl[K(\psi_i,\psi_j)\bigr]_{i,j=1}^n.
\end{equation}
Because $K(\psi,\psi)\leq2$, the matrix $\Gamma_n$ is positive definite with diagonal entries at most one and therefore satisfies the GSW input conditions above.
The second term in Equation~\eqref{equation:kernel-gsw-gram} rewards balance over $\mc H_K$, while the identity term supplies robustness in outcome directions that are not represented by this RKHS.

Write $\mathrm{GSW}(K,\varphi)$ for the resulting assignment law, suppressing its dependence on $n$.
The mean-zero updates imply $E[Z_i|\psi_{1:n}]=0$ for every unit \citep{harshaw2024gsw}.
The update directions and step sizes depend only on the covariates and auxiliary randomness independent of the potential outcomes, so $Z \indep (Y_i(0),Y_i(1))_{i=1}^n | \psi_{1:n}$.
Thus, $\mathrm{GSW}(K,\varphi)$ is admissible.

\begin{remark}[Kernel Allocation]
The RKHS imbalance objective in Equation~\eqref{equation:rkhs-uniform-imbalance} also underlies the kernel allocation design of \citet{kallus2018}, which randomizes between $z^*$ and $-z^*$ for the optimal vector $z^*\in\arg\min z'W_nz$ with $z\in\{-1,1\}^n$.
Computing $z^*$ requires solving an NP-hard mixed integer program, and the resulting design has support of size two.
This limited randomization can have poor robustness properties and preclude standard inference \citep{kallus2021}.
Kernelized GSW avoids these issues through computationally tractable updates and broad randomization support.
\end{remark}

\subsection{Oracle Inequality and Smoothness Gains}\label{subsection:kernel-oracle-rates}

\citet{harshaw2024gsw} show that under GSW, the MSE of the Horvitz-Thompson estimator is bounded above by the loss of an implicit ridge regression of the outcome levels on the covariates.
\citet{harshaw2021dissertation} extends this inequality to kernelized GSW, yielding an implicit kernel-ridge objective.
Extending these results to our current superpopulation setting and combining with Proposition~\ref{proposition:variance-decomposition} yields the following oracle inequality for excess variance.

% PROOF-ID: kernel-gsw-oracle
\begin{prop}[Oracle Inequality]\label{proposition:kernel-oracle}
Let $\sigma=\mathrm{GSW}(K,\varphi)$ as above.
Then
\begin{equation}\label{equation:kernel-oracle}
\mc V_n(\sigma,P) \leq \inf_{g\in\mc H_{K}} \left\{ \frac4\varphi E\left[(m(\psi)-g(\psi))^2\right] + \frac8{(1-\varphi)n} \norm{g}_{\mc H_{K}}^2 \right\}.
\end{equation}
\end{prop}

The first term controls the residual $m-g$ not captured by the RKHS, while the second reflects the difficulty of balancing $g$, as measured by its RKHS norm.
The guarantee is strongest when $m$ can be approximated accurately by functions of moderate norm, which we show below occurs under sufficient smoothness restrictions.

Recall that a sequence of designs $\sigma_n$ is pointwise asymptotically efficient at $P$ if $\mc V_n(\sigma_n,P)\to0$.
Previously, this has only been established for matched pairs and other finely stratified designs \citep{bai2021inference,bai2023efficiency}.
Equation~\eqref{equation:kernel-oracle} shows asymptotic efficiency also holds for any kernelized GSW design whose RKHS is dense in $L^2(P_\psi)$.

% PROOF-ID: kernel-gsw-pointwise
\begin{cor}[Pointwise Efficiency]\label{corollary:kernel-pointwise}
Fix a design kernel $K$ and $\varphi\in(0,1)$, and, for each $n$, let $\sigma_n=\mathrm{GSW}(K,\varphi)$.
Suppose $\mc H_K$ is dense in $L^2(P_\psi)$.
As $n \to \infty$, excess variance 
\[
\mc V_n(\sigma_n,P) \longrightarrow 0.
\]
\end{cor}

For every $\nu>0$, the Mat\'ern design RKHS is dense in $L^2(P_\psi)$ under every Borel law on $[0,1]^d$.
The Gaussian radial basis function kernel with length scale $\ell>0$ and $W_{\mathrm{RBF},\ell}(\psi,\psi')=\exp\{-\norm{\psi-\psi'}_2^2/(2\ell^2)\}$ has the same density property \citep{scholkopf2002learning}.
Thus, Corollary~\ref{corollary:kernel-pointwise} applies to either kernel family.

However, as we argued in Section~\ref{section:matching}, such pointwise results can fail to yield meaningful guarantees on finite-sample efficiency.
Instead, we use the oracle inequality to derive uniform rates on  the excess variance over suitable regularity classes.

The simplest example of such a bound follows from the inequality above when the outcome model belongs to a bounded ball in the design RKHS.
For a design kernel $K$ and $B>0$, define this class by $\mc P_K(B)=\left\{P:m\in\mc H_K,\ \norm{m}_{\mc H_K}\leq B\right\}$.

% PROOF-ID: kernel-gsw-ball
\begin{cor}[RKHS Fast Rate]\label{corollary:kernel-ball}
Let $\sigma=\mathrm{GSW}(K,\varphi)$ as above.
Then
\begin{equation}\label{equation:kernel-ball}
\sup_{P\in\mc P_K(B)}\mc V_n(\sigma,P) \leq \frac{8B^2}{(1-\varphi)n}.
\end{equation}
\end{cor}

This shows kernelized GSW attains the fast parametric rate $n^{-1}$ uniformly over any fixed RKHS ball.
The next result is a direct Sobolev specialization: when smoothness $s>d/2$, the Mat\'ern-Sobolev correspondence places a bounded $H^s$ ball inside a bounded Mat\'ern RKHS ball, so the rate in Corollary~\ref{corollary:kernel-ball} transfers immediately.
To state this result on a scale that remains comparable across covariate dimensions, let $H_{\mathrm{av}}^s(\mr^d)$ denote $H^s(\mr^d)$ equipped with the dimension-normalized norm
\begin{equation}\label{equation:average-sobolev-norm}
\norm{F}_{H_{\mathrm{av}}^s(\mr^d)}^2\equiv\int_{\mr^d}|\widehat F(\omega)|^2\left(1+\frac{\norm{\omega}_2^2}{d}\right)^sd\omega.
\end{equation}
For fixed $d$, this is equivalent to the ordinary Sobolev norm and hence defines the same function space.
Define the probability model 
\begin{equation}\label{equation:sobolev-outcome-class}
\mc P_d^s \equiv \left\{ P: m\in H_{\mathrm{av}}^s(\mr^d),\ \norm{m}_{H_{\mathrm{av}}^s(\mr^d)} \leq 1 \right\}.
\end{equation}

% PROOF-ID: matern-supercritical
\begin{cor}[Non-adaptive Fast Rate]\label{corollary:matern-supercritical}
Let $s>d/2$ and $\varphi\in(0,1)$.
Choose $\nu>0$ such that $d/2+\nu\leq s$, and let $K=1+W_d^{\mathrm{Mat},\nu}$.
For design $\sigma=\mathrm{GSW}(K,\varphi)$, there is a constant $C$ depending only on $(d,\nu,\varphi)$ such that
\begin{equation}\label{equation:matern-supercritical}
\sup_{P\in\mc P_d^s}\mc V_n(\sigma,P) \leq C n^{-1}.
\end{equation}
\end{cor}

A fixed Mat\'ern kernelized GSW design thus attains the parametric excess-variance rate $n^{-1}$ uniformly over Sobolev outcome models with smoothness $s>d/2$.
However, this simple result requires both enough smoothness $s>d/2$ and enough prior knowledge of $s$ to correctly choose the Mat\'ern tuning parameter $\nu$.
Also, for fixed $s$ the condition $s>d/2$ eventually fails along any sequence $d=d_n\to\infty$, so this analysis is unsuitable for studying asymptotics beyond low, fixed dimension.

Next, we remove both restrictions by constructing a single sequence of kernelized GSW designs that adapts to every $s>0$ without prior knowledge.

\subsection{Adaptation to Unknown Smoothness}\label{subsection:sobolev-adaptation}

For any fixed $c>0$, set $\nu_n=c/\log n$, so the Sobolev order $d/2+\nu_n$ of the Mat\'ern design RKHS approaches the critical boundary $d/2$ from above, independent of the unknown smoothness $s$.
This yields a design sequence that adapts to every $s>0$.

% PROOF-ID: sobolev-gsw-upper
\begin{thm}[Adaptive Kernel GSW]\label{theorem:sobolev-gsw}
Fix $c>0$ and $\varphi\in(0,1)$, and set $\nu_n=c/\log n$, $K_n=1+W_d^{\mathrm{Mat},\nu_n}$, and $\sigma_n=\mathrm{GSW}(K_n,\varphi)$.
For every $s>0$, there is a constant $C$ depending only on $(s,c,\varphi,\overline p)$ such that, for every $d\geq1$ and $n\geq2$,
\begin{equation}\label{equation:sobolev-gsw}
\sup_{P\in\mc P_d^s}\mc V_n(\sigma_n,P) \leq C\left(\frac{\log n}{n}\right)^{(2s/d)\wedge 1}.
\end{equation}
\end{thm}

For fixed $d$, the exponent $2s/d$ increases continuously with smoothness until the rate reaches $\log n/n$ at $s=d/2$.
When $s>d/2$, this is only a logarithmic factor slower than the $n^{-1}$ rate in Corollary~\ref{corollary:matern-supercritical}.
Thus, a single design sequence adapts to every $s>0$ without prior knowledge of $s$, paying only a logarithmic factor for adaptivity.

%The meaning of the boundary can be seen from a narrowing bump.
%Let $\eta$ be a smooth function supported on the unit ball with $\eta(0)=1$.
%For an interior point $x_0\in(0,1)^d$ and sufficiently small $h>0$, define $f_h(x)=\eta\{(x-x_0)/h\}$.
%For fixed $(d,s)$, its Sobolev norm satisfies
%\begin{equation}\label{equation:sobolev-spike-scaling}
%\norm{f_h}_{H^s([0,1]^d)}
%\asymp
%h^{d/2-s}.
%\end{equation}
%When $s<d/2$, a fixed-height spike can become arbitrarily narrow while its Sobolev norm converges to zero.
%At $s=d/2$, its norm remains of constant order.
%When $s>d/2$, its norm diverges.
%Subcritical Sobolev balls therefore permit highly localized fluctuations at low cost, while such fluctuations become increasingly costly above the boundary.

\begin{remark}[Comparison to Matched Pairs]
In the supplementary online appendix, we show that, for fixed $d\geq3$, the same design attains the minimax rate $n^{-2\beta/d}$, up to logarithmic factors, uniformly over $\beta$-H\"older outcome models for every $0<\beta\leq1$.
Unlike matching, however, its rate continues to improve under Sobolev smoothness above order one and becomes nearly parametric once $s\geq d/2$.
\end{remark}

Figure~\ref{figure:smoothness-simulation} illustrates the finite-sample efficiency gains from smoothness for kernelized GSW vs.\ matched pairs.
The relative efficiency improvement grows as the outcome model becomes smoother.
The figure also shows that a design using the fixed parameter $\nu_0=1$ captures much of the gain achieved by Oracle Mat\'ern GSW, which uses the true smoothness parameter, and the advantage increases with $n$.

For comparison, the figure also includes a nonparametric rerandomization design. 
In particular, we implement a kernelized version of the linear best-of-$m$ rerandomization in \cite{wang2025bestchoice}.
It selects, out of $m=10{,}000$ independent allocations $Z$, the one minimizing the worst-case imbalance $\mc I_W(Z)$ in Equation~\eqref{equation:rkhs-uniform-imbalance}, for the Mat\'ern kernel $W$ with the oracle value of $\nu$.
In contrast to its strong performance for the linear outcome model in Figure~\ref{figure:matching-reversal-simulation}, here rerandomization offers essentially no improvement over matched pairs, reflecting the difficulty of balancing a nonparametric function space through random search.

\begin{figure}[t]
\centering
\includegraphics[width=\textwidth]{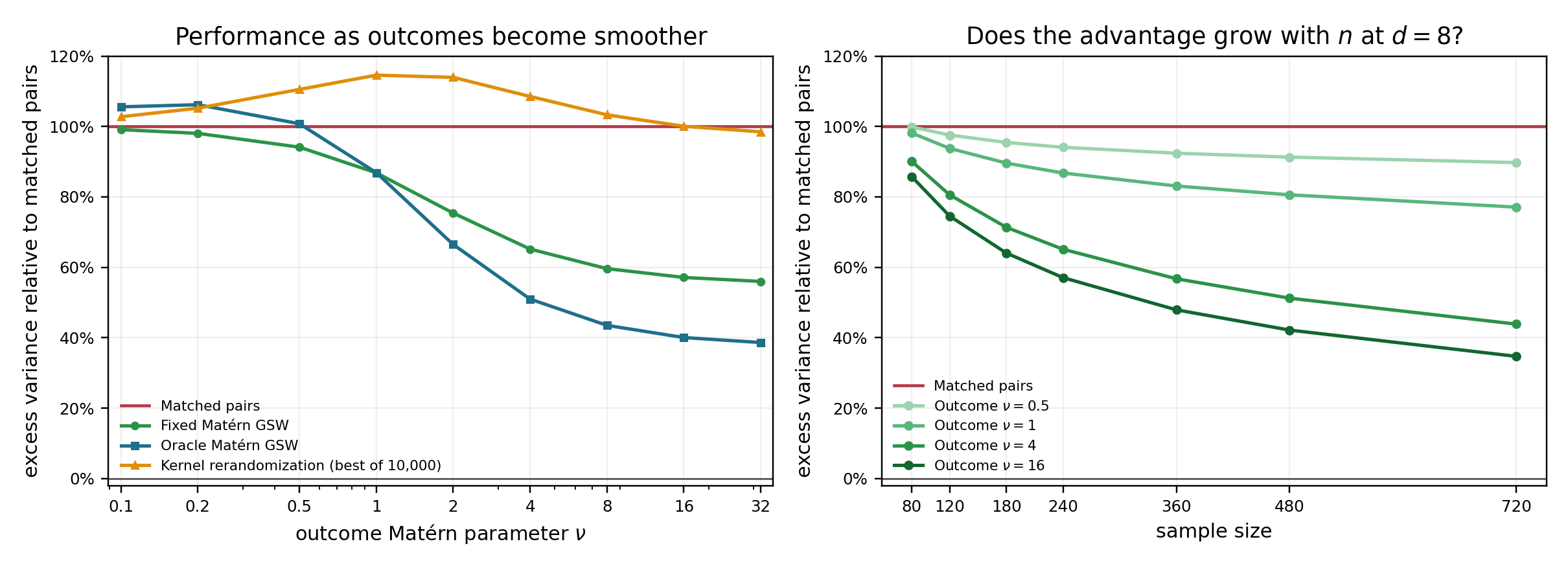}
\caption{Prior-averaged excess variance for outcome models $m(\psi)$ drawn from a Mat\'ern Gaussian process with parameter $\nu$. Covariates are uniform on $[0,1]^d$. The left panel fixes $n=240$ and $d=8$. Fixed Mat\'ern GSW uses $\nu_0=1$ independently of the true $\nu$, while Oracle Mat\'ern GSW uses $\nu_0=\nu$. Kernel rerandomization selects, among 10,000 draws, the allocation minimizing $Z'W_nZ$, where $W_n$ is the Mat\'ern kernel matrix constructed using the true $\nu$. The right panel fixes $d=8$ and varies $n$ for $\nu\in\{0.5,1,4,16\}$. The secondary online appendix gives the exact model and simulation protocol.}
\label{figure:smoothness-simulation}
\end{figure}

\subsection{Impossibility Result}\label{subsection:design-agnostic-lower-bound}

The adaptive rate in Theorem~\ref{theorem:sobolev-gsw} still has the basic form $n^{-2s/d}$, which becomes vacuous when $s$ is fixed and $d=d_n\geq a\log n$ for some $a>0$.
Similar to the H\"older lower bound in Theorem~\ref{theorem:design-agnostic-lower-bound}, the next theorem shows that this is not a limitation of kernelized GSW: every admissible design faces the same dimensionality barrier.

% PROOF-ID: sobolev-design-lower
\begin{thm}[Design-Agnostic Lower Bound]\label{theorem:sobolev-lower}
Fix $s>0$.
There is a constant $c_0>0$ depending only on $s$ such that, for every $d\geq1$ and every $n\geq2$,
\begin{equation}\label{equation:sobolev-lower}
\inf_{\sigma\in\mc D_n} \sup_{P\in\mc P_d^s}\mc V_n(\sigma,P) \geq c_0n^{-2s/d}.
\end{equation}
\end{thm}

The central implication is the dimensionality barrier.
For fixed $s$ and any sequence $d_n\geq a\log n$, Theorem~\ref{theorem:sobolev-lower} gives design-agnostic asymptotic inefficiency 
\begin{equation}\label{equation:sobolev-dimensionality-barrier}
\liminf_{n\to\infty}\inf_{\sigma\in\mc D_n}\sup_{P\in\mc P_{d_n}^s}\mc V_n(\sigma,P)>0.
\end{equation}
Hence, no admissible design can guarantee convergence to the full efficiency bound beyond the low-dimensional regime $d_n=o(\log n)$.
Motivated by this impossibility result, the next section relaxes the pursuit of full efficiency by balancing lower-complexity nonparametric working models.
When such a working model closely approximates the true outcome model, its variance benchmark remains close to the Hahn bound but can be approached much more rapidly in finite samples.

\begin{remark}[Fixed-Dimensional Minimaxity]
For fixed $d < \infty$ and $s\leq d/2$, combining Theorems~\ref{theorem:sobolev-gsw} and~\ref{theorem:sobolev-lower} gives the following finite-sample minimax comparison.
There are constants $0<c_1<C_1<\infty$, depending on $(d,s,c,\varphi,\overline p)$, such that, for every $n\geq2$,
\begin{equation}\label{equation:sobolev-minimax}
c_1 n^{-2s/d} \leq \inf_{\sigma\in\mc D_n} \sup_{P\in\mc P_d^s}\mc V_n(\sigma,P) \leq C_1\left(\frac n{\log n}\right)^{-2s/d}.
\end{equation}
In low-dimensional asymptotics with $d<\infty$ fixed while $n\to\infty$, kernelized GSW is thus minimax rate optimal for every $s\leq d/2$, up to logarithmic factors.
\end{remark}

% !TeX root = ../main.tex
\section{Efficient Designs Beyond Low Dimensions}\label{section:restricted-efficiency}
In this section, we construct kernelized GSW designs that balance low-complexity nonparametric working models chosen to approximate the outcome model $m(\psi)$, rather than pursuing uniform convergence to the full efficiency bound.
By imposing additive or low-order interaction structure, these designs can accommodate many more covariates than matching or full-dimensional kernelized GSW, while still providing substantial nonparametric efficiency gains.

\subsection{Structured Nonparametric Balance}\label{subsection:structured-balance}

Kernelized GSW controls the worst-case imbalance $n^{-1}\sum_{i=1}^nZ_ig(\psi_i)$ over the unit ball of the RKHS supplied to the algorithm.
Section~\ref{subsection:unit-gsw} used a full-dimensional Mat\'ern space.
Here we instead use structured sum RKHSs, beginning with two examples.

\begin{ex}[Nonparametric Main Effects]\label{example:structured-priority}
We can control each covariate nonparametrically by targeting additive functions of the form $g(\psi)=a+\sum_{j=1}^dg_j(\psi_j)$.
\end{ex}

\begin{ex}[Bivariate Nonparametric Effects]\label{example:structured-bivariate}
Alternatively, we can control every bivariate interaction nonparametrically by balancing functions
\begin{equation}\label{equation:bivariate-nonparametric-class}
g(\psi)=a+\sum_{j=1}^dg_j(\psi_j)+\sum_{j<\ell}g_{j\ell}(\psi_j,\psi_\ell).
\end{equation}
\end{ex}

Both classes can be implemented directly with kernelized GSW through kernel addition.
Let $W_1,\ldots,W_B$ be continuous positive semidefinite kernels satisfying $W_b(\psi,\psi)\leq1$ and set $K(\psi,\psi')=1+B^{-1}\sum_{b=1}^BW_b(\psi,\psi')$.
Then the RKHS $\mc H_K$ consists of functions of the form $g=a+\sum_{b=1}^Bf_b$ with $f_b\in\mc H_{W_b}$.
For instance, one kernel corresponding to Example~\ref{example:structured-priority} is $K_1(\psi,\psi')=1+d^{-1}\sum_{j=1}^dW_1^{\mathrm{Mat},\nu}(\psi_j,\psi_j')$.

% PROOF-SOURCE: impossibility/dva/section5_restructure/weighted_structured_kernel_master_theorem.md, Propositions 1--2.
% PROOF-SOURCE: impossibility/dva/section5_restructure/raw_bivariate_adaptive_master_theorem.md.
% PROOF-AUDIT: impossibility/dva/section5_restructure/raw_bivariate_matern_adversarial_audit.md (PASS for the canonical Lebesgue-ANOVA class).
% PROOF-AUDIT: impossibility/dva/section5_restructure/raw_bivariate_matern_centering_in_proof.md (raw uncentered kernel PASS; arbitrary raw-representation radius FAIL).

\paragraph{Variance Gap.}
The sum RKHS provides a general formula for low-complexity nonparametric approximations to $m$.
Let $\mc G$ denote the $L^2(P_\psi)$ closure of $\mc H_K$ and let $m_{\mc G}$ be the $L^2(P_\psi)$ projection of $m$ onto $\mc G$. We view this as a working model for $m$, with variance gap 
\begin{equation}\label{equation:working-model-approximation-gap}
\Delta_{\mc G}(P)\equiv4E\bigl((m(\psi)-m_{\mc G}(\psi))^2\bigr).
\end{equation}
The corresponding variance target is $V^*(P)+\Delta_{\mc G}(P)$.
When $\mc G=L^2(P_\psi)$, we have $m_{\mc G}=m$ and $\Delta_{\mc G}(P)=0$, recovering the full efficiency bound.
More generally, when $\mc G$ closely approximates $m$, this target remains close to $V^*(P)$, but can often be approached much more rapidly in finite samples.

\emph{Regularity Conditions.}
As in Section~\ref{section:structure}, obtaining uniform rates requires regularity of the structured approximation $m_{\mc G}$.
The component representations in Examples~\ref{example:structured-priority} and~\ref{example:structured-bivariate} are not unique.
For example, main effects can be absorbed into the bivariate components.
Because of this, we impose regularity conditions on the canonical functional ANOVA decomposition \citep{sobol1993}.
In particular, expand 
\begin{equation}\label{equation:canonical-low-order-decomposition}
m_{\mc G}(\psi)=a+\sum_{j=1}^dm_j(\psi_j)+\sum_{j<\ell}m_{j\ell}(\psi_j,\psi_\ell).
\end{equation}
We require $\int_0^1m_j(u)du=0$ as well as $\int_0^1m_{j\ell}(u,v)du=0$ for every $v$ and vice-versa.
These centering conditions make the corresponding component subspaces mutually orthogonal in $L^2([0,1]^d)$, so the decomposition is unique.
For smoothness $s>0$, we require the canonical components in Equation~\eqref{equation:canonical-low-order-decomposition} to satisfy the pooled Sobolev budget
\begin{equation}\label{equation:low-order-sobolev-budget}
\sum_{j=1}^d\norm{m_j}_{H_{\mathrm{av}}^s(\mr)}^2+\sum_{j<\ell}\norm{m_{j\ell}}_{H_{\mathrm{av}}^s(\mr^2)}^2\leq1.
\end{equation}
The pooled bound prevents the aggregate magnitude and roughness of the structured projection from growing mechanically with the number of recorded covariates.
For the main-effects working model, the bivariate terms are omitted.

We state the main result for the richer bivariate design in Example~\ref{example:structured-bivariate}, leaving the main-effects case to Theorem~\ref{theorem:proof-general-low-order-gsw} in the appendix.
For this result, take $\mc G_2$ to be the $L^2(P_\psi)$ closure of the class in Equation~\eqref{equation:bivariate-nonparametric-class} and let $\mc P_{d,2}^{s}$ contain the laws satisfying Assumption~\ref{assumption:sampling}, $E[m(\psi)^2]\leq1$, and Equation~\eqref{equation:low-order-sobolev-budget}.
Fix $c>0$ and set $\nu_n=c/\log n$, so each component RKHS lies $c/\log n$ above its critical Sobolev order.
Writing $\psi_{j\ell}=(\psi_j,\psi_\ell)$ and similarly for $\psi'$, define
\begin{equation*}
K_n(\psi,\psi')=1+\binom d2^{-1}\sum_{j<\ell}W_2^{\mathrm{Mat},\nu_n}(\psi_{j\ell},\psi_{j\ell}').
\end{equation*}
For every $n$, the $L^2(P_\psi)$ closure of $\mc H_{K_n}$ is $\mc G_2$.

% PROOF-SOURCE: impossibility/dva/section5_restructure/raw_bivariate_adaptive_master_theorem.md.
% PROOF-SOURCE: impossibility/dva/section4_matern_rebuild/v3_followup/adaptive_additive_sobolev_restricted_efficiency.md.
% PROOF-AUDIT: impossibility/dva/section5_restructure/raw_bivariate_matern_adversarial_audit.md (PASS for k=2 canonical class).
% PROOF-AUDIT: impossibility/dva/section4_matern_rebuild/v3_followup/adaptive_additive_interaction_sobolev_adversarial_audit.md (PASS for k=1 additive class).
\Needspace{9\baselineskip}
% PROOF-ID: adaptive-low-order-kernel-gsw
\begin{thm}[Bivariate Effects]\label{theorem:low-order-gsw}
Fix $n\geq2$, $d\geq2$, and smoothness $s>0$.
Let $\sigma_n=\mathrm{GSW}(K_n,\varphi_n)$, where $1-\varphi_n=1/2\wedge(n^{-1}d^2\log n)^{1/2}$.
Then, for a constant $C>0$ independent of $n$ and $d$, uniformly over $P\in\mc P_{d,2}^{s}$,
\begin{equation}\label{equation:low-order-gsw-approximation}
\mc V_n(\sigma_n,P)\leq\Delta_{\mc G_2}(P)+C\left(\frac{d^2\log n}{n}\right)^{(s\wedge1)/2}.
\end{equation}
\end{thm}

Theorem~\ref{theorem:proof-general-low-order-gsw} in the appendix extends this result to allow a fixed block of priority covariates, such as baseline outcomes, to enter jointly and have a different smoothness order.
The extension adds the corresponding fixed-dimensional approximation rate while retaining the bivariate rate above for the remaining covariates.

When the working model is accurate, $\Delta_{\mc G_2}(P)$ is small, so the design approaches a variance target close to the Hahn bound $V^*(P)$ much more rapidly than the $n^{-2s/d}$ rate obtained by the full-dimensional Mat\'ern designs in Section~\ref{section:structure}, up to logarithmic factors.
The lower-complexity specification also accommodates many more covariates than those designs can.
Whereas all of the preceding designs, including matching and full-dimensional GSW, require the very low-dimensional regime $d_n=o(\log n)$, Equation~\eqref{equation:low-order-gsw-approximation} allows the number of covariates to grow nearly as fast as $\sqrt n$:
\begin{equation*}
\mc V_n(\sigma_n,P)\leq\Delta_{\mc G_2}(P)+o(1)\quad\text{if}\quad d_n=o(\sqrt{n/\log n}).
\end{equation*}
Note also that, as in Section~\ref{section:structure}, the design does not require prior knowledge of smoothness since the near-critical kernel sequence adapts automatically.
The calibration above uses $\varphi_n\to1$, but the theorem does not suggest that this is necessary.
In simulations, small fixed values of $\varphi$ continue to perform well, and we conjecture that this technical requirement is an artifact of the current analysis.

\emph{Well Specification.} In Theorem~\ref{theorem:proof-general-low-order-gsw} in the appendix, we provide a sharper rate under correct specification of the working model $m_{\mc G}$.
We show that if $m=m_{\mc G_2}$, then any fixed $\varphi\in(0,1)$ gives $\sigma_n=\mathrm{GSW}(K_n,\varphi)$ satisfying $\mc V_n(\sigma_n,P)\lesssim(d^2\log n/n)^{s\wedge1}$.

\smallskip

\emph{Main effects.}
The additive nonparametric design in Example~\ref{example:structured-priority} permits still faster dimension growth.
Let $\mc G_1$ denote the $L^2(P_\psi)$ closure of the main-effects models in Example~\ref{example:structured-priority}.
Writing $\sigma_{n,1}$ for the corresponding main-effects GSW design, Theorem~\ref{theorem:proof-general-low-order-gsw} shows that $\mc V_n(\sigma_{n,1},P)\leq\Delta_{\mc G_1}(P)+o(1)$ whenever $d_n=o(n/\log n)$.
Thus, replacing bivariate interactions by additive main effects increases the allowable dimension from $d = o(\sqrt{n/\log n})$ to $d = o(n/\log n)$, at the cost of the potentially larger variance gap $\Delta_{\mc G_1}(P)$.

\subsection{Robustifying Structured Designs with Matching}\label{subsection:matching-structured-balance}

Section~\ref{subsection:structured-balance} showed that targeting a lower-complexity nonparametric working model can deliver strong variance reductions even when $d$ grows far beyond $o(\log n)$.
These gains come with a residual variance gap $\Delta_{\mc G}(P)$ from predictable outcome variation outside the working model.
We now robustify structured GSW by combining it with matching.
Structured GSW retains the fast rates above for balancing the modeled component $m_{\mc G}$ globally, while matching provides some protection against unmodeled components $m - m_{\mc G}$ using local comparisons.

\paragraph{Matched GSW.} Let $M=\{(i_p,j_p):p\in[n/2]\}$ be a matching as in Section~\ref{section:matching}.
Any assignment treating exactly one unit in each pair can be written as $Z_{i_p}=T_p$ and $Z_{j_p}=-T_p$ for a \emph{pair orientation} $T_p\in\{-1,1\}$.
In classic matched-pairs, the $T_p$ are iid random signs.
Here, we propose to use kernelized GSW to generate them.

The imbalance can be written $n^{-1}\sum_{i=1}^nZ_ig(\psi_i)=n^{-1}\sum_{p=1}^{n/2}T_p(g(\psi_{i_p})-g(\psi_{j_p}))$.
Using this formula and a reproducing-property calculation as in Section~\ref{subsection:unit-gsw}, one can show that in this case the squared worst-case imbalance over the RKHS unit ball is $n^{-2}T'G_n^MT$, where $G_n^M\in\mr^{(n/2)\times(n/2)}$ is the pair-difference matrix
\begin{equation}\label{equation:matched-pair-gram}
(G_n^M)_{pq}=K(\psi_{i_p},\psi_{i_q})-K(\psi_{i_p},\psi_{j_q})-K(\psi_{j_p},\psi_{i_q})+K(\psi_{j_p},\psi_{j_q}).
\end{equation}
As before, we seek to randomize in a way that makes $T'G_n^MT$ small.
This suggests normalizing $G_n^M$ and using it as the input to the Gram-Schmidt walk to generate the pair orientations $T$.
Let $\widehat\kappa_n=\max_p(G_n^M)_{pp}=\max_p\norm{K(\psi_{i_p},\cdot)-K(\psi_{j_p},\cdot)}_{\mc H_K}^2$ be the largest squared RKHS distance within a matched pair and apply GSW to
\begin{equation}\label{equation:matched-gsw-gram}
\Gamma_n^M=\varphi I_{n/2}+\frac{1-\varphi}{\widehat\kappa_n}G_n^M.
\end{equation}
If $\widehat\kappa_n=0$, set $\Gamma_n^M=\varphi I_{n/2}$.
The matrix $G_n^M$ is positive semidefinite, and the diagonal entries of $\Gamma_n^M$ are at most one, as required by the algorithm.
Write $\mathrm{GSW}_M(K,\varphi)$ for the resulting admissible assignment law, which treats exactly half the units.

Let $\kappa_n=E[\widehat\kappa_n]$ be the expected largest squared RKHS distance within a matched pair.
Impose the Section~\ref{subsection:unit-gsw} normalization $K=1+W$, with $\sup_\psi W(\psi,\psi)\leq1$, so that $0\leq\kappa_n\leq4$.
For a square-integrable function $r(\psi)$, define its expected pair energy by
\begin{equation}\label{equation:matched-pair-energy}
Q_n^M(r)=n^{-1}E\bigg[\sum_{(i,j)\in M}\bigl(r(\psi_i)-r(\psi_j)\bigr)^2\bigg].
\end{equation}

The following paired analogue of Proposition~\ref{proposition:kernel-oracle} formalizes the local-global division of labor: matching controls within-pair variation locally, while kernelized GSW controls structured imbalance globally.

% PROOF-SOURCE: impossibility/dva/section6_rebuild/paired_oracle_and_admissibility.md.
% PROOF-SOURCE: impossibility/dva/section6_rebuild/improvement_search/smooth_plus_rough_pair_rates.md, Sections 1--2.
% PROOF-ID: matched-kernel-gsw
\begin{thm}[Matched Oracle Inequality]\label{thm:matched-gsw}
Fix $\varphi\in(0,1)$ and let $\sigma=\mathrm{GSW}_M(K,\varphi)$.
\begin{equation}\label{equation:matched-gsw-oracle}
\mc V_n(\sigma,P)\leq\inf_{f\in\mc H_K}\left\{\frac4\varphi Q_n^M(m-f)+\frac{4\kappa_n}{n(1-\varphi)}\norm{f}_{\mc H_K}^2\right\}.
\end{equation}
\end{thm}

Relative to the unit-level oracle in Equation~\eqref{equation:kernel-oracle}, matching replaces the global approximation error $E[(m(\psi)-f(\psi))^2]$ by the within-pair energy $Q_n^M(m-f)$ and replaces the fixed RKHS penalty coefficient $8$ by $4\kappa_n$.
This can significantly attenuate both terms when the residual varies little within pairs and matched units are close in the kernel geometry.

\begin{ex}[Smooth Plus Rough Decomposition]
Suppose $m=a+g+r$ for some $g\in\mc H_K$ with $\norm{g}_{\mc H_K}\leq B$ and a residual $r$.
Taking $f=g$ in Theorem~\ref{thm:matched-gsw} gives
\begin{equation}\label{equation:matched-smooth-rough-specialization}
\mc V_n(\sigma,P)\leq\frac4\varphi Q_n^M(r)+\frac{4\kappa_nB^2}{n(1-\varphi)}.
\end{equation}
The first term uses matching to control the residual $r$, while the second asks GSW to balance the structured component $g$.
One can show $\kappa_n=o(1)$ under suitable regularity conditions, providing a further rate enhancement.
\end{ex}

We now instantiate this oracle inequality using the main-effects design in Example~\ref{example:structured-priority}.
Recall that $\mc G_1$ is the $L^2(P_\psi)$ closure of the class in Example~\ref{example:structured-priority}.
Let $\mc P_{d,1}^{s}$ be defined as $\mc P_{d,2}^{s}$ with $\mc G_2$ replaced by $\mc G_1$ and the bivariate components omitted from Equations~\eqref{equation:canonical-low-order-decomposition} and~\eqref{equation:low-order-sobolev-budget}.
Using the same $\nu_n=c/\log n$ as above, define $K_{n,1}(\psi,\psi')=1+d^{-1}\sum_{j=1}^dW_1^{\mathrm{Mat},\nu_n}(\psi_j,\psi_j')$.

% PROOF-SOURCE: impossibility/dva/section6_rebuild/paired_oracle_and_admissibility.md.
% PROOF-SOURCE: impossibility/dva/section4_matern_rebuild/v3_followup/adaptive_additive_sobolev_restricted_efficiency.md.
% PROOF-AUDIT: impossibility/dva/conference_draft_v3/edit_log/17-section5-second-review/matched_low_order_theorem_audit.md (PASS WITH REPAIRS for k=1 and k=2; the appendix proof implements the required penalty truncation).
\Needspace{8\baselineskip}
% PROOF-ID: matched-low-order-nonparametric-balance
\begin{thm}[Matched Main-Effects]\label{theorem:matched-low-order-gsw}
Fix $n\geq2$, $d$, smoothness $s>0$, and $\varphi\in(0,1)$.
Let $\sigma_n=\mathrm{GSW}_M(K_{n,1},\varphi)$.
For each $P\in\mc P_{d,1}^{s}$, define the residual $r=m-m_{\mc G_1}$.
Then, uniformly over $P\in\mc P_{d,1}^{s}$, for a constant independent of $n$ and $d$,
\begin{equation}\label{equation:matched-low-order-gsw}
\mc V_n(\sigma_n,P)\lesssim Q_n^M(r)+\left(\frac{d\log n}{n}\right)^{(2s)\wedge1}.
\end{equation}
\end{thm}

\smallskip
Unlike Theorem~\ref{theorem:low-order-gsw}, this bound contains no fixed variance gap: the residual enters through the within-pair energy $Q_n^M(r)$.
If $d_n=o(n/\log n)$, the structured term vanishes and
\begin{equation*}
\mc V_n(\sigma_n,P)\lesssim Q_n^M(r)+o(1).
\end{equation*}
If $Q_n^M(r)\to0$ as well, then $\mc V_n(\sigma_n,P)\to0$, and the design approaches the full Hahn bound.
If $d\geq3$, the matching $M$ is 2-swap-stable for a cost satisfying Assumption~\ref{assumption:matching-cost}, and $r$ is $\beta$-H\"older with coefficient $L_r$ for some $0<\beta\leq1$, then also
\begin{equation*}
\mc V_n(\sigma_n,P)\lesssim L_r^2(d\Lambda^2)^{2\beta}n^{-2\beta/d}+\left(\frac{d\log n}{n}\right)^{(2s)\wedge1}.
\end{equation*}
Thus only the residual pays the rough-class matching rate from Section~\ref{section:matching}, while GSW attains the faster main-effects rate for the structured component.

Figure~\ref{figure:structured-balance-simulation} illustrates the finite-sample gains from targeting a lower-complexity nonparametric function class when many covariates are recorded.
When the nonparametric main-effects working model is well specified, so that $m=m_{\mc G_1}$, both main-effects designs strongly outperform matched pairs and full-dimensional GSW.
When the true outcome model additionally contains interactions, $m\neq m_{\mc G_1}$ then matched pairs and full-dimensional GSW initially outperform main-effects GSW but deteriorate as the number of covariates increases and are eventually overtaken by it.

The version of main-effects GSW robustified with matching performs the best overall. 
It interpolates between the two regimes, attenuating the omitted interaction terms when the dimension is small but also preserving global balance when the dimension is large.

\begin{figure}[!htbp]
\centering
\includegraphics[width=0.88\textwidth]{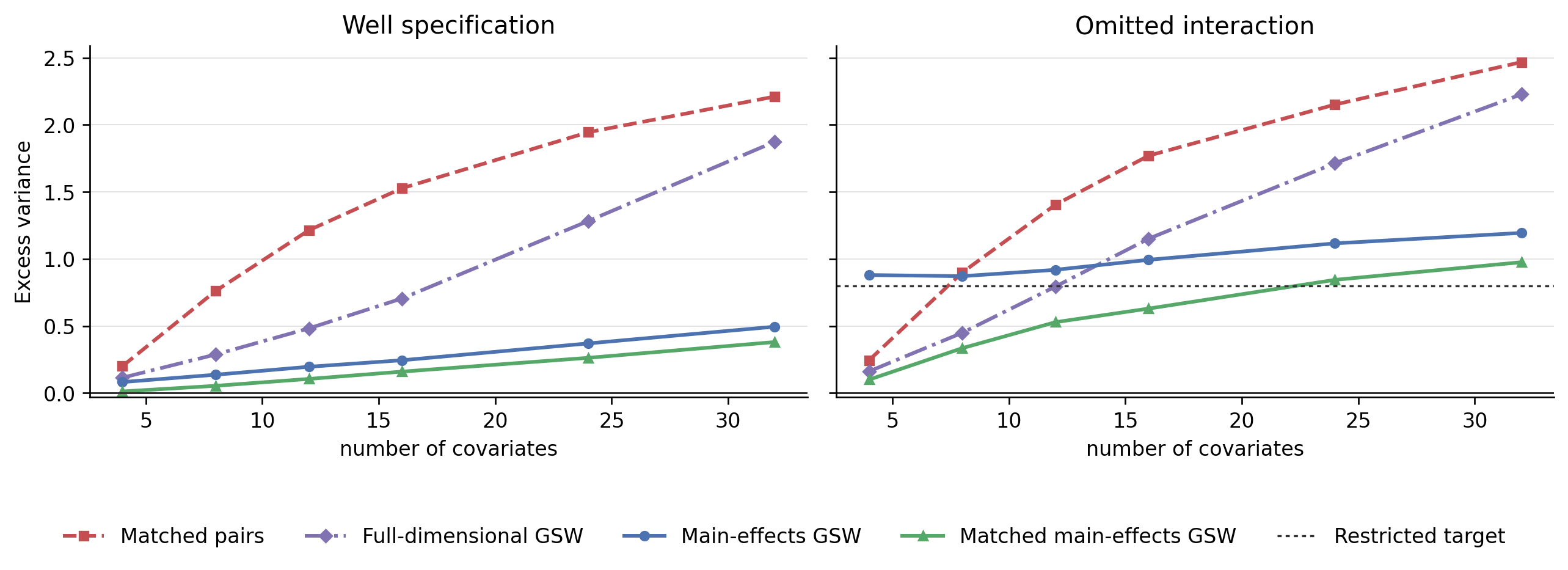}
\caption{Excess variance at $n=240$ for the main-effects design in Example~\ref{example:structured-priority}. The panels correspond to well specification and additional interactions between the covariates, respectively. The dotted line is the working-model excess-variance target $\Delta_{\mc G_1}(P)$. The main-effects designs use the additive kernel $K_1$ from Example~\ref{example:structured-priority}, an average of univariate Mat\'ern kernels. All GSW designs use $\varphi=0.03$ and $\nu=1$. The secondary online appendix gives the exact model and simulation protocol.}
\label{figure:structured-balance-simulation}
\end{figure}
\FloatBarrier

% !TeX root = ../main.tex

\section{Empirical Application}\label{section:empirical-application}

In this section, we evaluate our new designs using Monte Carlo simulations calibrated to 12 randomized experiments in recently published or forthcoming economics papers.
We compare classical matched-pairs randomization with the robust GSW designs from Section~\ref{section:restricted-efficiency}, with kernels ranging from the full-dimensional Mat\'ern to the reduced complexity nonparametric main effects specification.

The nonparametric additive and bivariate GSW designs in Section~\ref{subsection:structured-balance} reduce estimator variance relative to classical matched pairs in every experiment, while full-dimensional GSW as in Section~\ref{section:structure} performs similarly to matched pairs on average.
The matched versions of GSW in Section~\ref{subsection:matching-structured-balance} provide further robustness: every design reduces estimator variance relative to classical matched pairs in every experiment.
The main-effects specification in Example~\ref{example:structured-priority} delivers the largest variance reduction overall, which is reflected in shorter confidence intervals using our inference methods.

We selected 12 experiments from public replication data made available through the AEA, J-PAL, the World Bank, the \emph{Journal of Development Economics}, and the \emph{Quarterly Journal of Economics}.
The screening considered data availability, sample size, and the existence of meaningful pretreatment covariates and was fixed before any design comparisons.
Appendix~\ref{appendix:empirical-application} describes the selection procedure in detail.
The resulting experiments have sizes $n\in[96,1{,}200]$ and covariate dimensions $d\in[7,51]$.

\emph{Data Calibration.} For each experiment, we calibrate the Monte Carlo population to its empirical covariate distribution and estimated conditional potential-outcome laws.
For continuous outcomes, we model $Y_i(a)=m_a(\psi_i)+\sigma_a(\psi_i)\epsilon_{ia}$, $a\in\{0,1\}$, where each $\epsilon_{ia}$ has a normal distribution.
For binary outcomes, the conditional law is Bernoulli with success probability $m_a(\psi_i)$.
We use cross-validation to select a regression model from linear regression, elastic net with interactions, random forests, and gradient boosting, while a shallow random forest estimates $\sigma_a^2$ from squared cross-fitted residuals.
Clustered experiments are aggregated to the level of randomization using study-specific rules.
We then draw covariates $\psi_{1:n}$ from a smoothed bootstrap and potential outcomes from the fitted conditional laws.
See Appendix~\ref{appendix:empirical-application} for details.

\smallskip

\textbf{Designs and Estimators.}
We compare classical matched-pairs randomization with both unit-level and matched versions of the full-dimensional, bivariate, and additive GSW designs from Section~\ref{section:restricted-efficiency}.
All GSW designs set $\varphi=0.03$ and use the adaptive near-critical Mat\'ern calibration $\nu_n=1/\log n$.
We form a common set of pairs using the matching algorithm in \citet{cytrynbaum2022local}.
Under classical matched-pairs, the pair orientations $T_p\in\{-1,1\}$ are independent, whereas matched GSW chooses them jointly using the Gram-Schmidt walk.
Thus, in the simulation the matched designs differ only in how they jointly orient a common set of pairs.
All designs use difference in means as the point estimator.

\smallskip

\textbf{Inference.} Matched GSW makes the pair orientations dependent, so the usual variance estimators available for classical matched pairs are not valid.
Because of this, in Appendix \ref{appendix:matched-gsw-inference} we develop a new covariance corrected estimator for these designs.
Our main result in Theorem~\ref{theorem:matched-gsw-inference} establishes design-based finite-sample conservativeness for inference on the sample average treatment effect (SATE), while Proposition~\ref{theorem:matched-gsw-ate-calibration} calibrates it for asymptotically non-conservative inference on the ATE.
For each matched GSW design, we report the stabilized estimator $\widehat{\mc V}_{\mathrm{stab}}$, implemented using an allocation bank of size $2J = 1024$.
See the appendix for details.

For each simulated empirical environment, we draw 400 independent covariate populations and, within each population, 120 fresh potential outcome schedules and assignments.
We report estimator variance reduction relative to matched pairs, coverage of nominal 95\% ATE confidence intervals, and reduction in CI width.

\begin{figure}[t]
\centering
\includegraphics[width=\textwidth]{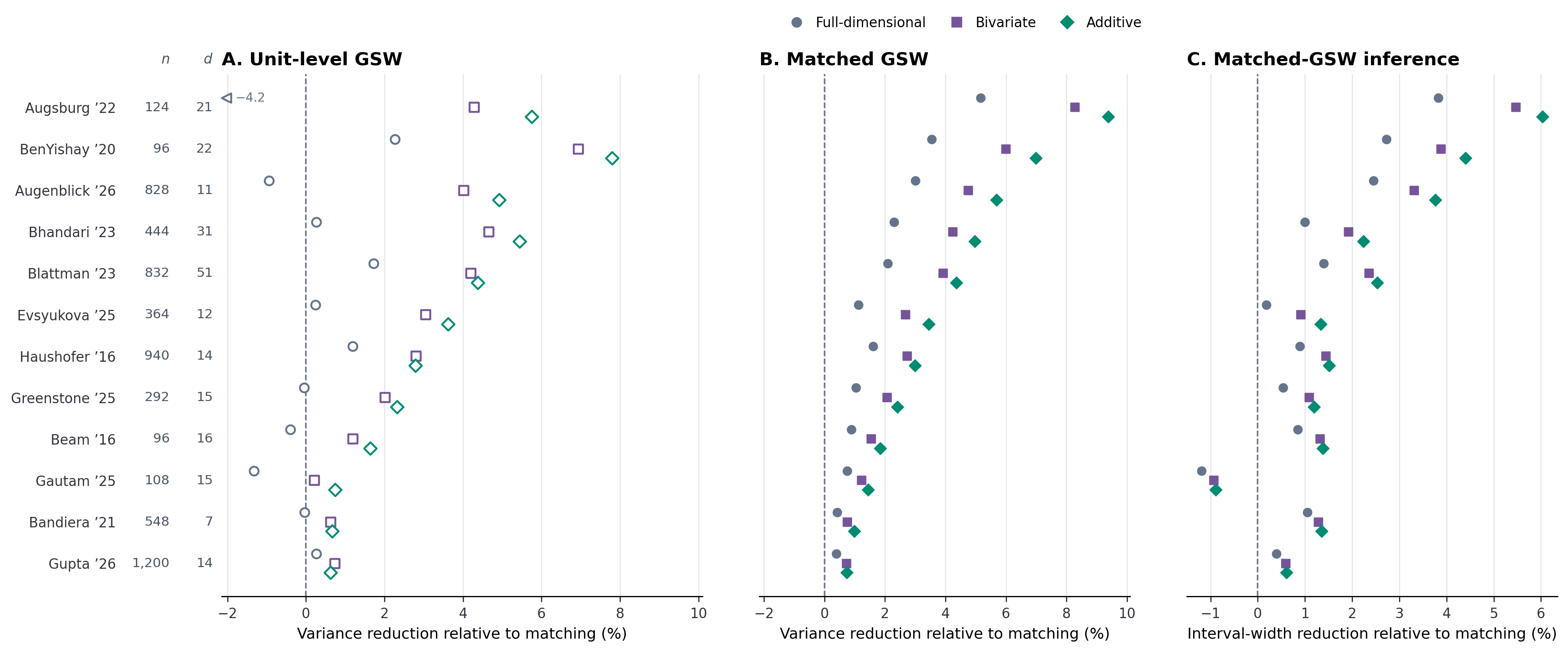}
\caption{GSW relative to matched pairs across 12 simulated experiments, with sample size and covariate dimension shown at left. Panels A and B report reductions in estimator variance for unit-level and matched GSW, respectively. Panel C reports the reduction in mean 95\% ATE confidence-interval width for matched GSW.}
\label{figure:empirical-application}
\end{figure}

\paragraph{Results.}
Figure~\ref{figure:empirical-application} documents broad efficiency gains: aside from unit-level full-dimensional GSW, every GSW design reduces estimator variance relative to classical matched pairs in every experiment, and all three matched designs do so uniformly.

The ranking among the unit-level designs illustrates the benefits from balancing lower-complexity nonparametric function spaces.
Recall that full-dimensional GSW targets the true semiparametric efficiency bound $V^*(P)$, whereas the bivariate and additive designs target the potentially larger efficiency bounds $V^*(P)+\Delta_{\mc G}(P)$.

These lower-complexity spaces are much easier to balance in finite samples.
As suggested by our rate theory, this advantage more than offsets the higher asymptotic variance target.
In particular, the nonparametric main effects design performs best in all settings, despite having the largest asymptotic variance target.

Panel C focuses on inference for matched GSW.
We develop new inference methods for these designs in Appendix~\ref{appendix:matched-gsw-inference}, which exploit the variance reduction to provide shorter confidence intervals.
Average coverage is 94.8\% for each matched design, with experiment-level coverage between 94.2\% and 95.2\%.

Taken together, these simulations show how the theoretical efficiency gains studied in the previous sections persist at the sample sizes and under the covariate distributions and potential-outcome relationships in actual experiments.
They support matched additive GSW as our preferred empirical design, since it combines the strongest efficiency gains with shorter confidence intervals and coverage close to nominal across all 12 experiments.

\section{Discussion and Recommendations for Practice}\label{section:discussion}

Both our theoretical and empirical results support a clear practical recommendation.
GSW designs targeting lower-complexity nonparametric function classes reduce estimator variance relative to classical matched pairs in every setting in our empirical application.
This reflects a basic difference between the designs.
Matching is subject to a severe curse of dimensionality, whereas discrepancy minimization designs like GSW can trade off covariate imbalances globally.
Combining the additive design with matching retains fast rates for nonparametric main effects while adding protection against misspecification, making matched additive GSW our preferred design overall.
Our covariance-corrected inference method also translates these efficiency gains into shorter confidence intervals, with coverage close to nominal.

Several directions merit further study.
The Gram-Schmidt walk is only one algorithmic implementation of global discrepancy minimization, and other algorithms different structured function classes may yield further improvements.
A more detailed investigation of inference for nonparametric GSW is beyond the scope of this paper but would be valuable, paralleling recent detailed investigations of inference for classical matched pairs designs \citep{fogarty2018,bai2026graph}.
Finally, it would be interesting to extend these methods beyond binary treatments to settings with multiple treatments or more complex experiments, as in the coupling-design approach of \citet{cytrynbaum2026coupling}.

\bibliographystyle{apalike}
\bibliography{references}

\appendix
\section{Appendix}

\subsection{Inference for Matched Gram-Schmidt Walk Designs}\label{appendix:matched-gsw-inference}

Matched GSW makes the orientations of different pairs dependent, so the usual matched-pairs variance estimator does not apply.
In this section, we construct a design-based variance estimator for matched GSW that is exactly conservative in finite samples for inference on the sample average treatment effect.
We also adapt the construction for inference on the population average treatment effect.

Suppose that $n=2N$, where $N$ is even, and write $M=\{(i_p,j_p):p\in[N]\}$ for the original matching.
Let $T_p=1$ when $i_p$ is treated and $T_p=-1$ otherwise.
For each unit, let $\tau_i=Y_i(1)-Y_i(0)$ and $\ylevel_i=(Y_i(1)+Y_i(0))/2$.
Define $A_p=(\tau_{i_p}+\tau_{j_p})/2$ and $B_p=\ylevel_{i_p}-\ylevel_{j_p}$.
The two potential treated-minus-control contrasts in pair $p$ are $C_p(1)=Y_{i_p}(1)-Y_{j_p}(0)$ and $C_p(-1)=Y_{j_p}(1)-Y_{i_p}(0)$.
Only the contrast selected by $T_p$ is observed, and it satisfies $C_p=C_p(T_p)=T_p(Y_{i_p}-Y_{j_p})=A_p+T_pB_p$.
Then
\begin{equation}\label{equation:matched-inference-estimator-decomposition}
\widehat\theta= \frac1N\sum_{p=1}^NA_p+\frac1NT'B.
\end{equation}

\emph{Orientation bank.}
Conditional on $\psi_{1:n}$ and the matching $M$, we draw independent matched-GSW orientation vectors $T^{(1)},\ldots,T^{(J)}$ and form the sign-symmetric bank of allocations $\mc B_J=\bigl(T^{(1)},-T^{(1)},\ldots,T^{(J)},-T^{(J)}\bigr)$.
We select the implemented orientation uniformly from the $2J$ indexed entries of $\mc B_J$.
Because the matched-GSW law is sign symmetric, this leaves its marginal assignment law unchanged.
Recall $\mathcal W_n=\sigma((\psi_i,Y_i(0),Y_i(1))_{i=1}^n)$.
Let $U_n$ be any exogenous randomness used during matching and set $\mathcal W_{n,J}=\sigma(\mathcal W_n,U_n,\mc B_J)$.
By sign symmetry, $\E[T|\mathcal W_{n,J}]=0$.
The conditional covariance matrix is $R\equiv\E[TT'|\mathcal W_{n,J}]=J^{-1}\sum_{j=1}^JT^{(j)}T^{(j)'}$.
Then
\begin{equation}\label{equation:matched-inference-bank}
n\var(\widehat\theta|\mathcal W_{n,J})=\frac2NB'RB.
\end{equation}

Following the quadratic-bound approach to design-based variance estimation of \citet{harshaw2026optimized}, let $L$ be any outcome-blind positive semidefinite $N\times N$ matrix with $L_{pp}=1$ that is measurable with respect to the realized covariates, matchings, and bank.
If $|R_{pq}|<1$ for every $p\neq q$, define variance estimator
\begin{equation}\label{equation:matched-inference-completion}
\widehat{\mc V}(L)\equiv\frac2N\left(\sum_{p=1}^NC_p^2+\sum_{p\neq q}\frac{T_pT_qR_{pq}+L_{pq}}{1+T_pT_qR_{pq}}C_pC_q\right).
\end{equation}
Define the sample average treatment effect $\SATE_n\equiv n^{-1}\sum_{i=1}^n\tau_i=N^{-1}\sum_{p=1}^NA_p$.
Sign symmetry gives $\E[\widehat\theta|\mathcal W_{n,J}]=\SATE_n$ and hence $\E[\widehat\theta|\mathcal W_n]=\SATE_n$.
The proof of the following result is given in Appendix~\ref{appendix:proofs-matched-gsw-inference}.

% PROOF-ID: matched-gsw-conditional-conservatism
\begin{thm}[Design-Based Conservativeness for SATE]\label{theorem:matched-gsw-inference}
For every matrix $L$ as above, $\E[\widehat{\mc V}(L)|\mathcal W_{n,J}]=n\var(\widehat\theta|\mathcal W_{n,J})+2A'LA/N$.
In particular, $\widehat{\mc V}(L)$ is conservative for design-based variance of $\widehat\theta$ about the $\SATE_n$:
\begin{equation}\label{equation:matched-inference-expectation}
\E[\widehat{\mc V}(L)|\mathcal W_n]\geq n\var(\widehat\theta|\mathcal W_n).
\end{equation}
\end{thm}

\emph{Inference on the SATE.}
Theorem~\ref{theorem:matched-gsw-inference} establishes finite-sample conservatism for inference on $\SATE_n$, with the choice of $L$ determining its degree.
To limit this conservativeness, apply the fine-stratification matcher of \citet{cytrynbaum2022local} to the covariate centroids of the $N$ original pairs, and let $\Pi$ denote the resulting outer matching.
For each $p\in[N]$, write $\pi(p)$ for its outer mate and define the swap matrix $S_\Pi$ by $(S_\Pi x)_p=x_{\pi(p)}$.
The matrix $L_\Pi=I_N-S_\Pi$ is PSD with unit diagonal, and we have
\begin{equation}\label{equation:matched-inference-sate-slack}
\E[\widehat{\mc V}(L_\Pi)|\mathcal W_{n,J}]-n\var(\widehat\theta|\mathcal W_{n,J})=\frac2NA'L_\Pi A=\frac2N\sum_{\{p,q\}\in\Pi}(A_p-A_q)^2.
\end{equation}
Thus, $\widehat{\mc V}(L_\Pi)$ is exactly conservative for $\SATE_n$, with smaller slack when the outer matching groups pairs with similar pair-average treatment effects.

\medskip
\emph{Inference on the ATE.}
The law of total variance gives $n\var(\widehat\theta)=\E[n\var(\widehat\theta|\mathcal W_{n,J})]+\var(\tau)$.
Then for inference on the ATE, we choose $L$ so the term $2A'LA/N$ converges to $\var(\tau)$.
Set $L_0=N(I_N-11'/N)/(N-1)$ and $L_{\mathrm{pop}}=(L_\Pi+L_0)/2$, which are PSD with unit diagonal.
Writing $\bar A=N^{-1}\sum_{p=1}^NA_p$, the new component satisfies $N^{-1}A'L_0A=(N-1)^{-1}\sum_{p=1}^N(A_p-\bar A)^2$ and measures the global dispersion of pair-average treatment effects.
The following result formalizes this calibration.
Its proof is given in Appendix~\ref{appendix:proofs-matched-gsw-inference}.

% PROOF-ID: matched-gsw-ate-calibration
\begin{prop}[ATE Calibration]\label{theorem:matched-gsw-ate-calibration}
Under Assumption~\ref{assumption:sampling}, suppose the matching satisfies Equation~\eqref{equation:tight-matching} and construct $\Pi$ as above.
Then $2A'L_{\mathrm{pop}}A/N\convp\var(\tau)$.
\end{prop}

The covariance-corrected estimator can be noisy.
By the Schur product theorem, $L_S=R\circ R$ is positive semidefinite with unit diagonal, and substituting it into Equation~\eqref{equation:matched-inference-completion} gives the always-nonnegative estimator $\widehat{\mc V}_S=2(D_TC)'R(D_TC)/N$, where $D_T$ is diagonal with entries $T_p$.
Take $\lambda_N=N^{-1/2}$, using $L=L_\Pi$ for SATE inference and $L=L_{\mathrm{pop}}$ for ATE inference.
When $|R_{pq}|<1$ for every $p\neq q$, we report the nominal interval $\widehat\theta\pm z_{1-\alpha/2}(\widehat{\mc V}_{\mathrm{stab}}(L)/n)^{1/2}$, where
\begin{equation}\label{equation:matched-inference-stabilized}
\widehat{\mc V}_{\mathrm{stab}}(L)\equiv\max\bigl((1-\lambda_N)\widehat{\mc V}(L)+\lambda_N\widehat{\mc V}_S,\lambda_N^2\widehat{\mc V}_S\bigr).
\end{equation}
Otherwise, we use $\widehat{\mc V}_S$ in place of $\widehat{\mc V}_{\mathrm{stab}}(L)$.

\subsection{Empirical Application Details}\label{appendix:empirical-application}

\paragraph{Study selection.}
We constructed a fixed candidate frame of 100 randomized studies from 247 records in five public replication repositories.
We sourced 30 studies from AEA/openICPSR, 20 each from the J-PAL Dataverse, World Bank Impact Evaluation Surveys, and the \emph{Journal of Development Economics} Mendeley collection, and 10 from the \emph{Quarterly Journal of Economics} Dataverse.
Within each source, we verified randomized assignment and selected the most recent qualifying studies.

We screened the replication packages for data containing a treatment, outcome, at least five meaningful pretreatment covariates, and a recoverable randomization unit.
We also required an adequate number of randomization units and legally usable files.
Thirty-five studies passed this initial screen.
We ranked these studies by their baseline-covariate counts, subject to a cap of six studies from any source, and retained 20 for a detailed file-level audit.

Neither treatment-effect estimates, balance statistics, estimator variances, nor the performance of any design entered the screening or ranking decisions.
One selected package could not be fitted because its raw microdata remained behind an access agreement, leaving 19 fitted experimental populations.

\begin{table}[!t]
\centering
\caption{Calibrated experimental populations in Figure~\ref{figure:empirical-application}. RF denotes a shallow random forest and ENet-2 denotes an elastic net with degree-two terms.}
\label{table:empirical-populations}
\small
\singlespacing
\begin{tabular}{@{}lll@{}}
\toprule
Study & Outcome & Mean model \\
\midrule
\citet{gupta2026paperwork} & Administrative pension receipt & RF \\
\citet{blattman2023cbt} & Antisocial-behavior index & RF \\
\citet{gautam2025sanitation} & Sanitation ownership rate & RF \\
\citet{bhandari2023able} & Evaluation of the incumbent & RF \\
\citet{augsburg2022nature} & Open-defecation rate & ENet-2 \\
\citet{benyishay2020gender} & Farmer knowledge score & ENet-2 \\
\citet{beam2016job} & Employment rate & RF \\
\citet{evsyukova2025linkedout} & LinkedIn connections & RF \\
\citet{augenblick2026retrieval} & Savings in maize-equivalent kilograms & OLS \\
\citet{greenstone2025pollution} & Log particulate emissions & RF \\
\citet{bandiera2021authority} & Item-adjusted log purchase price & RF \\
\citet{haushofer2016transfers} & Nondurable consumption & RF \\
\bottomrule
\end{tabular}
\end{table}

The final design-independent screen removed populations with insufficient covariate signal.
We retained fitted populations for which $4\var(m(\psi))/(V^*(P)+4\var(m(\psi)))\geq0.05$.
This ratio is the share of complete-randomization variance that any covariate-adaptive design could potentially remove.
When it is small, there is too little scope for evaluating the relative performance of different covariate-balancing designs.
The rule was fixed before any design comparisons and yielded 13 eligible fitted populations, of which Figure~\ref{figure:empirical-application} reports 12, with one excluded due to Monte Carlo time considerations.

\FloatBarrier

\paragraph{Fitted populations.}
For each study, we reproduce the paper's treatment contrast and outcome sample, applying documented study-specific sample restrictions and missing-data rules.
When treatment is assigned by cluster, we average outcomes and baseline covariates within each randomization cluster.
We verify the resulting sample by reproducing a treatment-effect estimate reported in the original paper.

For each arm $a\in\{0,1\}$, let $\widehat m_a(\psi)$ and $\widehat v_a(\psi)$ denote the fitted conditional mean and variance.
Five-fold cross-validation selects one mean-model family for each study from OLS, an elastic net with degree-two terms, a shallow random forest, and gradient-boosted trees.
For continuous outcomes, the score is out-of-sample $R^2$.
For binary outcomes, it is the improvement in Brier loss over the arm mean.
The family with the highest score averaged across the two arms is then fit separately within each arm on all observed units.
For continuous outcomes, a shallow random forest estimates the conditional variance from squared five-fold cross-fitted residuals.
For binary outcomes, the conditional variance is $\widehat m_a(\psi)(1-\widehat m_a(\psi))$.

To simulate covariates, we draw an observed covariate vector uniformly and add independent Gaussian jitter to continuous coordinates having at least 15 distinct values.
The jitter standard deviation is the mean nearest-neighbor spacing, capped at one quarter of the coordinate's sample standard deviation, and perturbed values are clipped to the observed range.
The fitted outcome models use these perturbed covariates on their original scale.
After perturbation, a separate copy of the design covariates is mapped to $[0,1]^d$ before constructing the kernels, using the empirical CDF for higher-cardinality coordinates and min-max scaling otherwise, with unordered categories represented by normalized indicator blocks.

At each Monte Carlo iteration, we evaluate the fitted conditional laws at the simulated covariates and draw both potential outcomes.
For continuous outcomes, $Y(a)=\widehat m_a(\psi)+\widehat v_a(\psi)^{1/2}\epsilon_a$ with Gaussian innovations, followed by clipping to any declared outcome bounds.
For binary outcomes, we draw Bernoulli potential outcomes with success probability $\widehat m_a(\psi)$.
The main specification takes the two innovations to be conditionally independent.
The sensitivity specifications use a Gaussian copula with correlation $0.5$ or a common latent rank, including the analogous latent-Gaussian construction for binary outcomes.
These alternatives change the across-study average variance reduction by at most $0.04$ percentage points for any GSW design.
The fitted models are held fixed, so the reported Monte Carlo uncertainty does not include model-estimation uncertainty.

For each study, we draw 400 independent covariate populations.
Within each population, we construct the matching and a sign-symmetric assignment bank, as in Appendix~\ref{appendix:matched-gsw-inference}, for each of the three unit-level and three matched GSW designs.
Each bank contains 512 independently generated walk allocations and their sign reversals, giving 1,024 possible allocations.
We then draw 120 potential-outcome schedules and one implemented allocation from each design for every schedule.

\section{Proofs}\label{appendix:proofs}

\paragraph{Notation.}
For a positive integer $r$, write $[r]=\{1,\ldots,r\}$ and $x_{1:r}=(x_1,\ldots,x_r)$.
The marginal law of $\psi$ under $P$ is denoted by $P_\psi$.
For vectors and matrices, $x'$ denotes transpose, $\norm{x}_2$ is the Euclidean norm, $I_r$ is the $r\times r$ identity matrix, $e_j$ is the $j$th coordinate vector, and $\mathds{1}_r$ is the vector of ones.
We write $\mathds{1}\{\mathcal E\}$ for the indicator of an event $\mathcal E$.
For symmetric matrices $A$ and $B$, $A\preceq B$ means that $B-A$ is positive semidefinite.
We use $a\wedge b=\min(a,b)$ and $a\vee b=\max(a,b)$.
For nonnegative quantities, $a\lesssim_\eta b$ means $a\leq C_\eta b$, while $a\asymp_\eta b$ means that both inequalities hold.
Generic positive constants $C$ and $c$ may change from line to line, with subscripts indicating permitted dependence.

\subsection{Setup}\label{appendix:proofs-setup}

For each unit, write $W_i=(\psi_i,Y_i(0),Y_i(1))$.
Let $\ylevel=(Y(1)+Y(0))/2$ and, for each unit, $\ylevel_i=(Y_i(1)+Y_i(0))/2$.

\begin{proof}[Proof of Proposition~\ref{proposition:variance-decomposition}]
Fix $\sigma\in\mathcal D_n$.
Let $\mathcal F_n=\sigma(\psi_{1:n})$ and $\mathcal W_n=\sigma(W_{1:n})$ denote the corresponding sigma-algebras, with $\mathcal F_n\subseteq\mathcal W_n$.
All expectations and variances below are taken under the joint law $(P,\sigma)$ unless stated otherwise.

Using Equation~\eqref{equation:ht-estimator}, write $\widehat\theta-\theta(P)=A_n+B_n$, where $A_n=n^{-1}\sum_{i=1}^n\tau_i-\theta(P)$ and $B_n=2n^{-1}\sum_{i=1}^nZ_i\yleveli$.
Definition~\ref{definition:admissible-designs} gives $Z\indep W_{1:n}|\mathcal F_n$ and $E[Z_i|\mathcal F_n]=0$.
It follows that $E[Z_i|\mathcal W_n]=E[Z_i|\mathcal F_n]=0$.
Hence $E[B_n|\mathcal W_n]=0$, while iid sampling gives $E[A_n]=0$.
It follows that $\widehat\theta$ is unbiased.
Moreover, $A_n$ is $\mathcal W_n$-measurable, so tower law gives $E[A_nB_n]=E[A_nE[B_n|\mathcal W_n]]=0$.
Then we have
\begin{equation}\label{equation:proof-ht-variance-split}
n\var(\widehat\theta)=nE[(A_n+B_n)^2]=\var(\tau)+\frac{4}{n}E\bigg[\bigg(\sum_{i=1}^nZ_i\yleveli\bigg)^2\bigg].
\end{equation}

Write $\varepsilon_i=\yleveli-m(\psi_i)$ and $\bar\sigma^2(x)=\var(\ylevel|\psi=x)$.
For a random element $A$, write $\mathcal L(A|B)$ for its conditional law given $B$.
By iid sampling, the conditional law factorizes as $\mathcal L(\varepsilon_{1:n}|\mathcal F_n)=\Pi_{i=1}^n\mathcal L(\varepsilon_i|\psi_i)$.
The conditional independence in Definition~\ref{definition:admissible-designs} gives $\varepsilon_{1:n}\indep Z|\mathcal F_n$, and therefore $\mathcal L(\varepsilon_{1:n}|\mathcal F_n,Z)=\mathcal L(\varepsilon_{1:n}|\mathcal F_n)$.
Taking the $i$th marginal of this law gives $E[\varepsilon_i|\mathcal F_n,Z]=E[\varepsilon_i|\mathcal F_n]=E[\varepsilon_i|\psi_i]=0$.
For $i\neq j$, the product law gives $E[\varepsilon_i\varepsilon_j|\mathcal F_n,Z]=E[\varepsilon_i|\psi_i]E[\varepsilon_j|\psi_j]=0$, while its $i$th marginal gives $E[\varepsilon_i^2|\mathcal F_n,Z]=E[\varepsilon_i^2|\psi_i]=\bar\sigma^2(\psi_i)$.
Set $C_n=\sum_{i=1}^nZ_im(\psi_i)$ and $R_n=\sum_{i=1}^nZ_i\varepsilon_i$.
The preceding moment identities give $E[R_n|\mathcal F_n,Z]=\sum_{i=1}^nZ_iE[\varepsilon_i|\mathcal F_n,Z]=0$.
They also give
\begin{equation*}
E[R_n^2|\mathcal F_n,Z]=\sum_{i,j=1}^nZ_iZ_jE[\varepsilon_i\varepsilon_j|\mathcal F_n,Z]=\sum_{i=1}^n\bar\sigma^2(\psi_i).
\end{equation*}
The last equality uses $Z_i^2=1$.
Since $C_n$ is measurable with respect to $\sigma(\mathcal F_n,Z)$,
\begin{equation*}
E\bigg[\bigg(\sum_{i=1}^nZ_i\yleveli\bigg)^2\bigg|\mathcal F_n,Z\bigg]=C_n^2+2C_nE[R_n|\mathcal F_n,Z]+E[R_n^2|\mathcal F_n,Z]=C_n^2+\sum_{i=1}^n\bar\sigma^2(\psi_i).
\end{equation*}
Taking expectations and using identical sampling gives
\begin{equation}\label{equation:proof-marginal-noise-split}
E\bigg[\bigg(\sum_{i=1}^nZ_i\yleveli\bigg)^2\bigg]=E\bigg[\bigg(\sum_{i=1}^nZ_im(\psi_i)\bigg)^2\bigg]+nE[\bar\sigma^2(\psi)].
\end{equation}

It remains to identify the design-invariant terms.
Let $\rho(x)=\cov(Y(1),Y(0)|\psi=x)$.
The law of total variance gives $\var(\tau)=\var\bigl(\tau(\psi)\bigr)+E[v_1^2(\psi)+v_0^2(\psi)-2\rho(\psi)]$.
The conditional variance formula for $\ylevel$ gives $4E[\bar\sigma^2(\psi)]=E[v_1^2(\psi)+v_0^2(\psi)+2\rho(\psi)]$.
Adding these identities cancels $\rho(\psi)$ and gives $\var(\tau)+4E[\bar\sigma^2(\psi)]=V^*(P)$.
Substituting Equation~\eqref{equation:proof-marginal-noise-split} and this identity into Equation~\eqref{equation:proof-ht-variance-split} proves Equation~\eqref{equation:variance-decomposition}.
\end{proof}

\subsection{Matching as a Conservative Design}\label{appendix:proofs-matching}

We first record the exact form of the excess variance under matched-pairs randomization.

\begin{lem}[Matched-Pairs Excess Variance]\label{lemma:proof-matched-pairs-excess}
Let $n\geq2$ be even and let $M$ be a matching with pairs $(i_p,j_p)$ for $p\in[n/2]$.
Then
\begin{equation}
\mc V_n(M,P)=\frac{4}{n}E\bigg[\sum_{p=1}^{n/2}\bigg(m(\psi_{i_p})-m(\psi_{j_p})\bigg)^2\bigg].
\end{equation}
\end{lem}

\begin{proof}
By Jensen's inequality and Assumption~\ref{assumption:sampling}, $E[m(\psi)^2] \leq E[\ylevel^2] <\infty$, so all second moments below are finite.
Let $U\indep W_{1:n}$ denote exogenous randomness that may be used to break ties in the pairing and set $\mathcal F_n=\sigma(\psi_{1:n},U)$.
Choose an $\mathcal F_n$-measurable enumeration $(i_p,j_p)_{p=1}^{n/2}$ of the pairs.
Then the within-pair treatments are $Z_{i_p}=-Z_{j_p} = T_p$ for iid Rademacher pair orientations $(T_p)_{p=1}^{n/2} \indep \mc W_n$.
Set within-pair gap $d_p=m(\psi_{i_p})-m(\psi_{j_p})$.
Then $\sum_{i=1}^nZ_im(\psi_i)=\sum_{p=1}^{n/2}T_pd_p$, and we may calculate
\begin{align*}
E\bigg[\bigg(\sum_{i=1}^nZ_im(\psi_i)\bigg)^2\bigg|\mathcal F_n\bigg]&=E\bigg[\bigg(\sum_{p=1}^{n/2}T_pd_p\bigg)^2\bigg|\mathcal F_n\bigg] =\sum_{p,q=1}^{n/2}d_pd_qE[T_pT_q|\mathcal F_n]=\sum_{p=1}^{n/2}d_p^2.
\end{align*}
Taking expectations and applying Proposition~\ref{proposition:variance-decomposition} proves the result.
\end{proof}

\subsubsection{A Lower Bound for Matching}

The next lemma gives a geometric lower bound that holds for every pairing rule.

\begin{lem}[Uniform Pair-Distance Bound]\label{lemma:proof-uniform-pair-distance}
Let $n\geq2$ be even, let $\psi_1,\ldots,\psi_n$ be iid uniform on $[0,1]^d$, let $U\indep\psi_{1:n}$, and let $(i_p,j_p)_{p=1}^{n/2}$ be any perfect pairing measurable with respect to $\sigma(\psi_{1:n},U)$.
If $\omega_d$ is the volume of the unit Euclidean ball in $\mathbb R^d$, then
\begin{equation}\label{equation:proof-uniform-pair-distance}
E\bigg[\frac{2}{n}\sum_{p=1}^{n/2}\|\psi_{i_p}-\psi_{j_p}\|_2^2\bigg]\geq\frac12\bigg(2\omega_dn\bigg)^{-2/d}.
\end{equation}
\end{lem}

\begin{proof}
For each unit, let $N_i=\min_{j\neq i}\|\psi_i-\psi_j\|_2$ denote its nearest-neighbor distance.
For every realization of $\psi_{1:n}$ and every pairing, a pair $\{i,j\}$ satisfies $\|\psi_i-\psi_j\|_2^2\geq\max\{N_i^2,N_j^2\}\geq(N_i^2+N_j^2)/2$.
Since $\sqcup_{p=1}^{n/2}\{i_p,j_p\}=[n]$, summing this inequality over the pairs gives
\begin{equation*}
\sum_{p=1}^{n/2}\|\psi_{i_p}-\psi_{j_p}\|_2^2\geq\sum_{p=1}^{n/2}\frac{N_{i_p}^2+N_{j_p}^2}{2}=\frac12\sum_{i=1}^nN_i^2.
\end{equation*}

Let $B(x,r)$ denote the Euclidean ball of radius $r$ centered at $x$.
Conditional on $\psi_i=x$, the remaining covariates are independent and uniform on $[0,1]^d$.
Hence, for every $r>0$, the union bound gives
\begin{align*}
P(N_i\leq r|\psi_i=x)&\leq\sum_{j\neq i}P\bigl(\psi_j\in B(x,r)\cap[0,1]^d|\psi_i=x\bigr)\\
&=(n-1)\mathrm{vol}\bigl(B(x,r)\cap[0,1]^d\bigr)\leq(n-1)\omega_dr^d<n\omega_dr^d.
\end{align*}
Set $r_0=(2\omega_d n)^{-1/d}$.
The preceding inequality gives $P(N_i>r_0)\geq1/2$.
Since $N_i^2\geq r_0^2\one\{N_i>r_0\}$, it follows that $E[N_i^2]\geq r_0^2P(N_i>r_0)\geq r_0^2/2$.
Taking expectations in the pathwise bound and using identical distributions gives
\begin{equation*}
E\bigg[\frac2n\sum_{p=1}^{n/2}\|\psi_{i_p}-\psi_{j_p}\|_2^2\bigg]\geq\frac1n\sum_{i=1}^nE[N_i^2]\geq\frac{r_0^2}{2}.
\end{equation*}
Substituting the definition of $r_0$ proves Equation~\eqref{equation:proof-uniform-pair-distance}.
\end{proof}

\begin{proof}[Proof of Theorem~\ref{theorem:matching-lower-bound}]
Fix $M\in\mathcal M_n$.
Let $U\indep W_{1:n}$ be its exogenous tie-breaking seed and set $\mathcal F_n=\sigma(\psi_{1:n},U)$, so the resulting pairs are $\mathcal F_n$-measurable.
Draw linear coefficient $\gamma\sim\Unif(\mathbb S^{d-1})$ with $\gamma \indep (\psi_{1:n},U)$.
For each $\gamma$, let $P_\gamma$ be the data law under which $\psi\sim\Unif([0,1]^d)$ and $Y(1)=Y(0)=m_\gamma(\psi)=\gamma'(\psi-\one_d/2)$.
Since $\|\gamma\|_2=1$, we have $P_\gamma\in\mathcal P_{\mathrm{lin}}$.
Under $P_\gamma$, $\tau=0$ and $v_a^2(\psi)=0$ for $a\in\{0,1\}$, so $V^*(P_\gamma)=\var\bigl(\tau(\psi)\bigr)+2E[v_1^2(\psi)]+2E[v_0^2(\psi)]=0$.

Write $A=E_\gamma[\gamma\gamma']$.
The standard isotropy calculation for the uniform law on $\mathbb S^{d-1}$ gives $A=d^{-1}I_d$, and hence $E_\gamma[(\gamma'u)^2]=u'Au=d^{-1}\|u\|_2^2$ for every $u\in\mathbb R^d$.

Set $\Delta_p=\psi_{i_p}-\psi_{j_p}$.
Since $P_\gamma\in\mathcal P_{\mathrm{lin}}$ for every $\gamma\in\mathbb S^{d-1}$,
\begin{align}
\sup_{P\in\mathcal P_{\mathrm{lin}}}\mc V_n(M,P)&\geq E_\gamma[\mc V_n(M,P_\gamma)]=\frac4nE\bigg[\sum_{p=1}^{n/2}E_\gamma\bigl[(\gamma'\Delta_p)^2|\mathcal F_n\bigr]\bigg] \notag\\
&=\frac{4}{nd}E\bigg[\sum_{p=1}^{n/2}\|\Delta_p\|_2^2\bigg]\geq\frac1d\bigl(2\omega_d\bigr)^{-2/d}n^{-2/d}. \label{equation:proof-linear-prior-bound}
\end{align}
The first inequality uses $P_\gamma\in\mathcal P_{\mathrm{lin}}$ for every $\gamma\in\mathbb S^{d-1}$.
Lemma~\ref{lemma:proof-matched-pairs-excess} and Tonelli's theorem give the first equality.
Independence of $\gamma$ from $\mathcal F_n$ and the preceding isotropy identity give the equality on the second line.
The final inequality is Lemma~\ref{lemma:proof-uniform-pair-distance}.

It remains lower bound $d^{-1} (2\omega_d)^{-2/d}$ uniformly in $d$.
For $d=1$, $\omega_1=2$, so this term equals $1/16>1/(2\pi e)$.
For $d\geq2$, let $r=d/2$.
Using $\omega_d=\pi^{d/2}/\Gamma(d/2+1)$ gives $d^{-1}(2\omega_d)^{-2/d}=(2r)^{-1}\Gamma(r+1)^{1/r}\pi^{-1}2^{-1/r}$.
Since $r\geq1$, the Stirling lower bound gives $\Gamma(r+1)\geq\sqrt{2\pi r}(r/e)^r\geq2(r/e)^r$, so that $\Gamma(r+1)^{1/r}\geq2^{1/r}r/e$.
Substitution into the preceding expression gives $(2r)^{-1}\Gamma(r+1)^{1/r}\pi^{-1}2^{-1/r}\geq(2r)^{-1}(2^{1/r}r/e)\pi^{-1}2^{-1/r}=1/(2\pi e)$.
The bound holds for the arbitrary design $M$ fixed at the start of the proof.
Taking the infimum over $M\in\mathcal M_n$ completes the proof.
\end{proof}

\begin{proof}[Proof of Corollary~\ref{corollary:matching-inefficiency}]
Let $R_n\equiv\inf_{M\in\mathcal M_n}\sup_{P\in\mathcal P_{\mathrm{lin}}}\mc V_n(M,P)$.
By Theorem~\ref{theorem:matching-lower-bound}, $R_n\geq(2\pi e)^{-1}n^{-2/d_n}$.
If $d_n\geq c\log n$, then $n^{-2/d_n}=\exp(-2\log n/d_n)\geq e^{-2/c}$.
This proves Equation~\eqref{equation:matching-inefficiency}.
\end{proof}

\subsubsection{Matching over Rough Outcome Classes}

The following geometric lemma turns 2-swap stability into a bound on the total within-pair distance.

\begin{lem}[Pair Distances]\label{lemma:proof-stable-pair-distance}
Suppose $n\geq2$ is even, $d\geq3$, $0<\beta\leq1$, and Assumption~\ref{assumption:matching-cost} holds.
For any $\psi_{1:n} \subseteq [0,1]^d$ and every 2-swap-stable pairing,
\begin{equation}
\frac{4}{n} \sum_{p=1}^{n/2} \|\psi_{i_p}-\psi_{j_p}\|_2^{2\beta} \leq 24\bigl(44d\Lambda^2\bigr)^{2\beta} n^{-2\beta/d}.
\end{equation}
\end{lem}

\begin{proof}
For each pair $p$, write $L_p\equiv\|\psi_{i_p}-\psi_{j_p}\|_2$.
For $t>0$, let $N(t)\equiv|\{p:L_p\geq t\}|$.
We first show
\begin{equation}\label{equation:proof-long-edge-count}
N(t) \leq 2\bigl(44d\Lambda^2\bigr)^d t^{-d} \quad \text{for every }t>0.
\end{equation}
Every pair distance is at most $\sqrt d$, so $N(t) = 0$ and the claim is immediate when $t>\sqrt d$.
Then fix $t\leq\sqrt d$ and partition the pairs with distance at least $t$ into the dyadic bands $\mathcal B_j\equiv\{p:2^jt\leq\|\psi_{i_p}-\psi_{j_p}\|_2<2^{j+1}t\}$ for $j\geq0$.

Fix a nonempty band, write $r=2^jt$, and set $h=r/(2\sqrt d\Lambda)$.
The band's nonemptiness and the diameter of $[0,1]^d$ give $r\leq\sqrt d$, so $h\leq1/2$ because $\Lambda\geq1$.
For $v\in\mathbb Z^d$, let cell $C_v\equiv h(v+[0,1)^d)$.
These half-open cells form a pairwise disjoint partition of $\mathbb R^d$.
Let $\mathcal V_h\equiv\{v:C_v\cap[0,1]^d\neq\varnothing\}$.
Let $v(i)\in\mathcal V_h$ be the unique index satisfying $\psi_i\in C_{v(i)}$.
Define $G_j$ to be the multigraph with vertex set $\mathcal V_h$ and one edge $e_p\equiv\{v(i_p),v(j_p)\}$ for each pair $p\in\mathcal B_j$.
We will show that $G_j$ is simple and has maximum degree at most $D_d\equiv(11\sqrt d\Lambda)^d$.
The edges of $G_j$ are in one-to-one correspondence with pairs $p\in\mathcal B_j$, so the degree-sum identity will then give $|\mathcal B_j|\leq|\mathcal V_h|D_d/2$.

We begin with simplicity of $G_j$.
Each cell has diameter $h\sqrt d=r/(2\Lambda)<r$.
If $p\in\mathcal B_j$, then $\|\psi_{i_p}-\psi_{j_p}\|_2\geq r$, so $v(i_p)\neq v(j_p)$ by this cell-diameter bound.
Thus $G_j$ has no loops.
Next, suppose for contradiction that two edges join the same vertices, $e_p=e_q=\{u,v\}$ for distinct $u,v\in\mathcal V_h$.
Relabeling endpoints if necessary, $\psi_{i_p},\psi_{i_q}\in C_u$ and $\psi_{j_p},\psi_{j_q}\in C_v$.
The alternative pairs $\{i_p,i_q\}$ and $\{j_p,j_q\}$ have lengths at most $r/(2\Lambda)$.
By Assumption~\ref{assumption:matching-cost}, their combined cost satisfies $c(\psi_{i_p},\psi_{i_q})+c(\psi_{j_p},\psi_{j_q})\leq2\overline\lambda(r/(2\Lambda))^q=2^{1-q}\underline\lambda r^q<2\underline\lambda r^q$, where the equality uses $\Lambda^q=\overline\lambda/\underline\lambda$.
By Assumption~\ref{assumption:matching-cost} again, the original pairs satisfy $c(\psi_{i_p},\psi_{j_p})+c(\psi_{i_q},\psi_{j_q})\geq2\underline\lambda r^q$.
This 2-swap would then strictly reduce their combined cost, contradicting \eqref{equation:swap-stability}.
Thus distinct pairs in $\mathcal B_j$ induce distinct edges, so $G_j$ is simple.

We next bound the maximum degree of $G_j$.
Suppose an edge joins $C_u$ and $C_v$, and write its endpoints as $x\in C_u$ and $y\in C_v$.
Its length is less than $2r$ by the definition of $\mathcal B_j$.
Because $x_\ell\in[hu_\ell,h(u_\ell+1))$ and $y_\ell\in[hv_\ell,h(v_\ell+1))$ for every index $\ell\in[d]$, we have $2r>\|x-y\|_2\geq|x_\ell-y_\ell|\geq h(|u_\ell-v_\ell|-1)$.
Thus $|u_\ell-v_\ell|<2r/h+1=4\sqrt d\Lambda+1$.
By this bound, for fixed $u_\ell$ the number of possible integer values of $v_\ell$ is at most $2(4\sqrt d\Lambda+1)+1=8\sqrt d\Lambda+3\leq11\sqrt d\Lambda$.
For fixed $u\in\mathcal V_h$, the number of vertices $v\in\mathcal V_h$ satisfying these $d$ coordinatewise restrictions is therefore at most $\Pi_{\ell=1}^d(11\sqrt d\Lambda)=(11\sqrt d\Lambda)^d=D_d$.
$G_j$ is simple, so at most one edge joins $u$ to each such $v$.
Then $G_j$ has maximum degree at most $D_d$.

It remains to count the graph's vertices.
If $v\in\mathcal V_h$, then for each coordinate $\ell$ the interval $[hv_\ell,h(v_\ell+1))$ must intersect $[0,1]$ for $v_\ell \in \mathbb Z$.
This implies $0\leq v_\ell\leq\lfloor h^{-1}\rfloor$, and hence $|\mathcal V_h|\leq(\lfloor h^{-1}\rfloor+1)^d\leq(h^{-1}+1)^d\leq(2/h)^d$, since $h\leq1/2$.
Then the degree-sum identity shows the number of edges $|\mathcal B_j|$ of $G_j$ is
\begin{equation}\label{equation:proof-stable-band-count}
|\mathcal B_j|=\frac12\sum_{u\in\mathcal V_h}\mathrm{deg}_j(u)\leq\frac12|\mathcal V_h|D_d\leq\frac12\left(\frac2h\right)^d\bigl(11\sqrt d\Lambda\bigr)^d\leq\bigl(44d\Lambda^2\bigr)^d(2^jt)^{-d}.
\end{equation}
The bands partition the pairs counted by $N(t)$, so summing Equation~\eqref{equation:proof-stable-band-count} over $j\geq0$ proves Equation~\eqref{equation:proof-long-edge-count} because $\sum_{j\geq0}2^{-jd}=(1-2^{-d})^{-1}\leq2$.
Set $C\equiv2(44d\Lambda^2)^d$, so Equation~\eqref{equation:proof-long-edge-count} becomes $N(t)\leq Ct^{-d}$ for every $t>0$.
Since $L_p\leq\sqrt d$ for every pair,
\begin{align*}
\sum_{p=1}^{n/2}L_p^{2\beta}&=2\beta\sum_{p=1}^{n/2}\int_0^{\sqrt d}t^{2\beta-1}\one\{L_p\geq t\}dt\\
&=2\beta\int_0^{\sqrt d}t^{2\beta-1}\sum_{p=1}^{n/2}\one\{L_p\geq t\}dt=2\beta\int_0^{\sqrt d}t^{2\beta-1}N(t)dt.
\end{align*}
The bound $N(t) \le Ct^{-d}$ diverges as $t \to 0$, whereas the trivial count gives $N(t)\leq n/2$ for any $t$.
Let $t_0\equiv(2C/n)^{1/d}$ be the unique point at which $Ct_0^{-d}=n/2$.
Then $n/2$ is the sharper bound for $t\leq t_0$, while $Ct^{-d}$ is the sharper bound for $t\geq t_0$.
Suppose first that $t_0\leq\sqrt d$.
Because $d\geq3$ and $\beta\leq1$, we have $2\beta<d$.
Splitting the integral,
\begin{align}
\sum_{p=1}^{n/2}L_p^{2\beta}&\leq2\beta\int_0^{t_0}t^{2\beta-1}\frac n2dt+2\beta C\int_{t_0}^{\sqrt d}t^{2\beta-d-1}dt\notag\\
&\leq\frac n2t_0^{2\beta}+\frac{2\beta}{d-2\beta}Ct_0^{2\beta-d}\leq\frac{3n}{2}t_0^{2\beta}. \label{equation:proof-stable-moment}
\end{align}
The final inequality uses $2\beta/(d-2\beta)\leq2$ and $Ct_0^{-d}=n/2$.
If $t_0>\sqrt d$, then $L_p<t_0$ for every pair, so the final bound in Equation~\eqref{equation:proof-stable-moment} continues to hold.
Consequently,
\begin{align}
\frac{4}{n}\sum_{p=1}^{n/2}\|\psi_{i_p}-\psi_{j_p}\|_2^{2\beta}&\leq6t_0^{2\beta}=6\cdot4^{2\beta/d}\bigl(44d\Lambda^2\bigr)^{2\beta}n^{-2\beta/d}\leq24\bigl(44d\Lambda^2\bigr)^{2\beta}n^{-2\beta/d}. \notag
\end{align}
The final inequality follows from $2\beta/d\leq2/3<1$.
\end{proof}

\begin{proof}[Proof of Theorem~\ref{theorem:stable-matching}]
Fix $P\in\mathcal P_\beta$.
Then the excess variance $\mc V_n(M_c,P)$ is equal to
\begin{align*}
\frac4nE\bigg[\sum_{p=1}^{n/2}\big(m(\psi_{i_p})-m(\psi_{j_p})\big)^2\bigg] \leq\frac4nE\bigg[\sum_{p=1}^{n/2}\|\psi_{i_p}-\psi_{j_p}\|_2^{2\beta}\bigg]\leq24(44d\Lambda^2)^{2\beta}n^{-2\beta/d}.
\end{align*}
The first expression follows from Lemma~\ref{lemma:proof-matched-pairs-excess}, and the first inequality from the H\"older condition $[m]_\beta\leq1$.
The second inequality follows from Lemma~\ref{lemma:proof-stable-pair-distance}.
The right hand side is independent of $P$.
Taking the supremum over $P\in\mathcal P_\beta$ proves the result.
\end{proof}

\clearpage
\hypersetup{pageanchor=false}
\pagenumbering{arabic}\renewcommand{\thepage}{\arabic{page}}

\begin{center}
{\Large Supplementary Online Appendix to ``The Limits of Experimental Design: Covariate Balance Beyond Low Dimension''}
\vskip 24pt
{\large Max Cytrynbaum}
\end{center}

\subsection{Design-Agnostic Lower Bound}\label{appendix:proofs-design-agnostic}

The all-design lower bound uses signed functions that are constant on the interior of small cells and vanish near their boundaries.

\begin{lem}[Signed Tabletop Functions]\label{lemma:proof-signed-tabletops}
Fix $0<\beta\leq1$, integers $d,q\geq1$, and $h=q^{-1}$.
Partition $[0,1]^d$ into $q^d$ cells $C$ of side length $h$.
There exist functions $F_C:[0,1]^d\to[0,1]$ and tabletop subsets $B_C\subseteq C$ such that $\mathrm{supp}(F_C)\subseteq C^\circ$ and $F_C|_{B_C}=1$.
Then $|B_C|\geq|C|/2=1/(2q^d)$.
For $(\omega_C)_C\in\{-1,1\}^{q^d}$, $[m_\omega]_\beta\leq1$ for
\begin{equation}\label{equation:proof-tabletop-witness}
m_\omega(x) \equiv a\sum_C\omega_C F_C(x), \qquad a \equiv 2^{-1-2\beta}h^\beta d^{-\beta}.
\end{equation}
\end{lem}

\begin{proof}
For a cell $C$, define $\delta_C(x)\equiv\inf_{y\in[0,1]^d\setminus C^\circ}\|x-y\|_2$.
Then $\delta_C(x)$ is the distance to the nearest face of $C$ when $x\in C^\circ$, and $\delta_C(x)=0$ otherwise.
Let $T(t)\equiv((t-1)\vee0)\wedge1$.
Define $F_C(x)=T(\delta_C(x)/b)$ for buffer $b=h/(8d)$.
Then $F_C(x)=0$ when $\delta_C(x)\leq b$, $F_C(x)=(\delta_C(x)-b)/b$ when $b<\delta_C(x)<2b$, and $F_C(x)=1$, the ``tabletop'', when $\delta_C(x)\geq2b$.
In particular, $F_C$ takes values in $[0,1]$ and is supported in $C^\circ$.

The triangle inequality gives $|\delta_C(x)-\delta_C(y)|\leq\|x-y\|_2$.
Each of the maps $t\mapsto t-1$, $t\mapsto t\vee0$, and $t\mapsto t\wedge1$ is one-Lipschitz, so $T$ is one-Lipschitz as their composition.
Therefore, $|F_C(x)-F_C(y)|\leq b^{-1}|\delta_C(x)-\delta_C(y)|\leq b^{-1}\|x-y\|_2$.

For a sign vector $\omega$, write $F_\omega\equiv\sum_C\omega_CF_C$.
We next show that $F_\omega$ is $b^{-1}$ Lipschitz on the full cube.
If $x,y\in C$, then $F_\omega(x)=\omega_CF_C(x)$ and $F_\omega(y)=\omega_CF_C(y)$, so the previous bound proves the claim.
Now let $x\in C$ and $y\in C'$ for distinct cells $C$ and $C'$.
Parameterize $\gamma(t)=(1-t)x+ty$ for $t\in[0,1]$.
Set $t_x=0$ when $x\in\partial C$, and otherwise set $t_x=\sup\{t\in[0,1]:\gamma(t)\in C^\circ\}$.
Set $t_y=1$ when $y\in\partial C'$, and otherwise set $t_y=\inf\{t\in[0,1]:\gamma(t)\in(C')^\circ\}$.
Because $\gamma$ is affine and the cell interiors are convex, the two preimages are intervals.
Their disjointness gives $0\leq t_x\leq t_y\leq1$, while continuity gives $z_x=\gamma(t_x)\in\partial C$ and $z_y=\gamma(t_y)\in\partial C'$.
Since every $F_C$ is supported in a cell interior, $F_\omega(z_x)=F_\omega(z_y)=0$.
Moreover $x$ and $z_x$ both lie in $C$, while $z_y$ and $y$ both lie in $C'$, so the same-cell bound above applies to each of these two pairs.
Then calculate
\begin{align*}
&|F_\omega(x)-F_\omega(y)|\leq|F_\omega(x)-F_\omega(z_x)|+|F_\omega(z_y)-F_\omega(y)|\\
&\leq b^{-1}\bigl(\|x-z_x\|_2+\|z_y-y\|_2\bigr)=b^{-1}(t_x+1-t_y)\|x-y\|_2\leq b^{-1}\|x-y\|_2.
\end{align*}
This proves the claim.
By disjoint support, $\|F_\omega\|_\infty\leq1$, so $\|m_\omega\|_\infty\leq a$ and $|m_\omega(x)-m_\omega(y)|\leq2a$.
The preceding Lipschitz bound is $|m_\omega(x)-m_\omega(y)|\leq(a/b)\|x-y\|_2$.
Interpolating between these two bounds using $\min\{u,v\}\leq u^{1-\beta}v^\beta$ for $u,v\geq0$,
\begin{align}
|m_\omega(x)-m_\omega(y)|&\leq\min\left\{2a,\frac ab\|x-y\|_2\right\}\notag\leq(2a)^{1-\beta}\left(\frac ab\|x-y\|_2\right)^\beta=\|x-y\|_2^\beta.
\end{align}
For the final equality, note that $b=h/(8d)$ and Equation~\eqref{equation:proof-tabletop-witness} give $a=2^{\beta-1}b^\beta$, so $(2a)^{1-\beta}(a/b)^\beta=2^{1-\beta}ab^{-\beta}=2^{1-\beta}2^{\beta-1}=1$.
Hence $[m_\omega]_\beta\leq1$.

Let $B_C\equiv\{x\in C:\delta_C(x)\geq2b\}$, so $F_C=1$ on $B_C$.
Write $C=\prod_{\ell=1}^d[c_\ell,c_\ell+h]$ and set $C^-=\prod_{\ell=1}^d[c_\ell+2b,c_\ell+h-2b]$.
If $x\in C^-$ and $y\in[0,1]^d\setminus C^\circ$, then some $\ell$ satisfies $y_\ell\leq c_\ell$ or $y_\ell\geq c_\ell+h$, and hence $\|x-y\|_2\geq|x_\ell-y_\ell|\geq2b$.
Taking the infimum over $y$ gives $\delta_C(x)\geq2b$, so $C^-\subseteq B_C$ and
\begin{equation*}
|B_C|\geq(h-4b)^d=h^d\bigg(1-\frac{1}{2d}\bigg)^d\geq\frac{h^d}{2}=\frac{|C|}{2}.
\end{equation*}
The final inequality is Bernoulli's inequality.
\end{proof}

\begin{proof}[Proof of Theorem~\ref{theorem:design-agnostic-lower-bound}]
Let $(F_C,B_C)_C$ be the functions and sets from Lemma~\ref{lemma:proof-signed-tabletops} for cells $C$ of side length $h=Q^{-1}$, where $Q=\lceil n^{1/d}\rceil$.
The number of cells is $K=Q^d$.
For each sign vector $\omega=(\omega_C)_C\in\{-1,1\}^K$, define a law $P_\omega$ by $\psi\sim\mathrm{Unif}([0,1]^d)$ and $Y(1)=Y(0)=m_\omega(\psi)$, as in Equation~\eqref{equation:proof-tabletop-witness}.
Then $P_\omega$ satisfies Assumption~\ref{assumption:sampling}, and Lemma~\ref{lemma:proof-signed-tabletops} gives $[m_\omega]_\beta\leq1$.
Thus $P_\omega\in\mathcal P_\beta$.
Moreover, $\tau(\psi)=0$ and $v_1^2(\psi)=v_0^2(\psi)=0$, so Equation~\eqref{equation:efficiency-bound} gives $V^*(P_\omega)=0$.

Now fix $\sigma\in\mathcal D_n$.
Let $E_\psi$ denote expectation under $\psi_i\simiid\mathrm{Unif}([0,1]^d)$, and let $E_{\psi,\sigma}$ denote expectation under the joint law of $(\psi_{1:n},Z)$.
For each cell, define imbalance $G_C=\sum_{i=1}^nZ_iF_C(\psi_i)$.
Let $\omega$ be distributed uniformly on $\{-1,1\}^K$ independently of $(\psi_{1:n},Z)$, and let $E_\omega$ denote expectation over $\omega$.
Then
\begin{align}
\sup_{P\in\mathcal P_\beta}\mc V_n(\sigma,P)&\geq E_\omega[\mc V_n(\sigma,P_\omega)]=\frac4nE_{\psi,\sigma} E_\omega\bigg[\bigg(\sum_{i=1}^nZ_im_\omega(\psi_i)\bigg)^2\bigg]\notag\\
&=\frac{4a^2}{n}E_{\psi,\sigma} E_\omega\bigg[\bigg(\sum_C\omega_CG_C\bigg)^2\bigg]=\frac{4a^2}{n}E_{\psi,\sigma}\bigg[\sum_CG_C^2\bigg]. \label{equation:proof-bump-prior-average}
\end{align}
The inequality follows since $P_\omega\in\mathcal P_\beta$ for every $\omega$.
The first equality follows from Proposition~\ref{proposition:variance-decomposition} and Tonelli's theorem applied to the independent law of $\omega$ and the joint law of $(\psi_{1:n},Z)$.
The second equality uses Equation~\eqref{equation:proof-tabletop-witness} and the definition of $G_C$.
The final equality follows from $E_\omega[\omega_C\omega_{C'}]=\one\{C=C'\}$.

For each cell, set $I_C=\{i\in[n]:\psi_i\in C\}$ and let $\mathcal S=\{C:I_C=\{i\}\text{ for some }i\in[n]\text{ with }\psi_i\in B_C\}$ be the set of cells containing exactly one sampled unit, with that unit on the tabletop $B_C$.
If $C\in\mathcal S$ and $I_C=\{i\}$, then $F_C(\psi_j)=0$ for $j\neq i$ and $F_C(\psi_i)=1$.
Therefore, $G_C^2=(Z_iF_C(\psi_i))^2=Z_i^2F_C(\psi_i)^2=1$, so $\sum_CG_C^2\geq|\mathcal S|$ for every realization of $Z$.
Taking expectations, since $\mathcal S$ is measurable with respect to $\psi_{1:n}$, we have $E_{\psi,\sigma}[\sum_CG_C^2]\geq E_{\psi,\sigma}[|\mathcal S|]=E_\psi[|\mathcal S|]$.
It remains to lower bound $E_\psi[|\mathcal S|]$.
\begin{align*}
P(C\in\mathcal S)&=\sum_{i=1}^nP\bigl(\psi_i\in B_C,\ \psi_j\notin C\text{ for every }j\neq i\bigr)\\
&=\sum_{i=1}^n|B_C|\bigg(1-\frac1K\bigg)^{n-1}=n|B_C|\bigg(1-\frac1K\bigg)^{n-1}\geq\frac{n}{2K}\bigg(1-\frac1K\bigg)^{n-1}.
\end{align*}
The first equality partitions the event $\{C\in\mathcal S\}$ according to the unique sampled unit in $C$.
The second equality follows by uniformity and independence, since $|C|=1/K$.
The final inequality follows from Lemma~\ref{lemma:proof-signed-tabletops}, which gives $|B_C|\geq1/(2K)$.
Therefore,
\begin{align*}
E_\psi[|\mathcal S|]&=E_\psi\bigg[\sum_C\one\{C\in\mathcal S\}\bigg]=\sum_CP(C\in\mathcal S) \geq\frac n2\bigg(1-\frac1K\bigg)^{n-1}\geq\frac{n}{2e}.
\end{align*}
The first inequality sums the preceding bound over the $K$ total cells.
For the final inequality, note that $K\geq n$, so $(1-1/K)^{n-1}\geq(1-1/K)^{K-1}$.
Then calculate $(1-1/K)^{K-1}=((K-1)/K)^{K-1}=(1+1/(K-1))^{-(K-1)}\geq e^{-1}$, where the final inequality uses $(1+1/m)^m\leq e$.
Equation~\eqref{equation:proof-bump-prior-average} and the work above give
\begin{align*}
\sup_{P\in\mathcal P_\beta}\mc V_n(\sigma,P)&\geq\frac{4a^2}{n}E_\psi[|\mathcal S|]\geq\frac{2a^2}{e} =\frac{2^{-1-4\beta}}{e}h^{2\beta}d^{-2\beta}\geq\frac{2^{-1-6\beta}}{e}d^{-2\beta}n^{-2\beta/d}.
\end{align*}
The equality is from the definition of $a$ in Equation~\eqref{equation:proof-tabletop-witness}.
For the final inequality, $Q=\lceil n^{1/d}\rceil\leq2n^{1/d}$, so $h=Q^{-1}\geq2^{-1}n^{-1/d}$.
Because $\sigma\in\mathcal D_n$ was arbitrary, taking the infimum over $\sigma$ proves the theorem with constant $2^{-1-6\beta}d^{-2\beta}/e$.
\end{proof}

\subsection{Proofs for Section~\ref{section:structure}}\label{appendix:proofs-section4}

Recall $\mathcal F_n=\sigma(\psi_{1:n})$. For a continuous positive semidefinite kernel $W$ satisfying $\sup_{\psi}W(\psi,\psi)\leq1$, write $W_n=[W(\psi_i,\psi_j)]_{i,j=1}^n$ and $K_n=[K(\psi_i,\psi_j)]_{i,j=1}^n$ for $K=1+W$, and set $\Gamma_n\equiv\varphi I_n+(1-\varphi)K_n/2$. Then $\Gamma_n\succeq\varphi I_n$ and $\Gamma_{n,ii}=\varphi+(1-\varphi)K(\psi_i,\psi_i)/2\leq1$.

Let $B_n=\Gamma_n^{1/2}$ and let $b_{n,i}$ be its $i$th column. Since $B_n'B_n=\Gamma_n$, we have $\norm{b_{n,i}}_2^2=e_i'B_n'B_ne_i=\Gamma_{n,ii}\leq1$. Theorem~6.3 of \citet{harshaw2024gsw} therefore gives $\cov(B_nZ|\mathcal F_n)\preceq B_n\Gamma_n^{-1}B_n'=I_n$. Moreover, $\Gamma_n\succeq\varphi I_n$ makes $B_n$ invertible, so premultiplying and postmultiplying by $B_n^{-1}$ gives $\cov(Z|\mathcal F_n)\preceq\Gamma_n^{-1}$. Lemma~6.1 of \citet{harshaw2024gsw} also gives $E[Z|\mathcal F_n]=0$, and hence
\begin{equation}\label{equation:proof-section4-gsw-covariance}
E[Z|\mathcal F_n]=0, \qquad \cov(Z|\mathcal F_n)\preceq \Gamma_n^{-1}.
\end{equation}

The next lemma records the algebraic form of the covariance bound.

\begin{lem}[GSW Ridge Bound]\label{lemma:proof-gsw-ridge}
Let $\mathcal H$ be a real Hilbert space and let $\mathcal A$ be a sigma-algebra. Let $x_1,\ldots,x_N\in\mathcal H$, $a,b>0$, and $y\in\mathbb R^N$ be $\mathcal A$-measurable, and define $F:\mathcal H\to\mathbb R^N$ by $(Fh)_i=\langle h,x_i\rangle_{\mathcal H}$. Suppose a random vector $S\in\mathbb R^N$ satisfies $E[S|\mathcal A]=0$ and $\cov(S|\mathcal A)\preceq(aI_N+bFF^*)^{-1}$. Then
\begin{equation}\label{equation:proof-gsw-ridge-bound}
E[(S'y)^2|\mathcal A]\leq\min_{h\in\mathcal H}\left\{\frac1a\norm{y-Fh}_2^2+\frac1b\norm{h}_{\mathcal H}^2\right\}.
\end{equation}
\end{lem}

\begin{proof}
The argument is pointwise in $\mathcal A$. Fix a realization, so $a$, $b$, $F$, and $y$ are deterministic.
Define $L(h)\equiv a^{-1}\norm{y-Fh}_2^2+b^{-1}\norm{h}_{\mathcal H}^2$. For any $\delta\in\mathcal H$, its directional derivative is $(d/dt)L(h+t\delta)|_{t=0}=-2a^{-1}\langle y-Fh,F\delta\rangle+2b^{-1}\langle h,\delta\rangle_{\mathcal H}$.
A second calculation gives the second directional derivative $2a^{-1}\norm{F\delta}_2^2+2b^{-1}\norm{\delta}_{\mathcal H}^2>0$ for $\delta\neq0$, so $L$ is strictly convex. Writing $\langle y-Fh,F\delta\rangle=\langle F^*(y-Fh),\delta\rangle_{\mathcal H}$ and setting the derivative to zero for every $\delta$ gives $bF^*(y-Fh)=ah$, so the first-order condition is $(aI+bF^*F)h=bF^*y$. For every $u\in\mathcal H$, $\langle u,(aI+bF^*F)u\rangle_{\mathcal H}=a\norm{u}_{\mathcal H}^2+b\norm{Fu}_2^2\geq a\norm{u}_{\mathcal H}^2$, so $aI+bF^*F$ is invertible. The first-order condition therefore has the unique solution $h_*=b(aI+bF^*F)^{-1}F^*y$, which is the unique minimizer of $L$.

To express the solution using the inverse in the covariance bound, write $R=(aI_N+bFF^*)^{-1}$. The push-through identity $(aI+bF^*F)^{-1}F^*=F^*R$ follows from $(aI+bF^*F)F^*=F^*(aI_N+bFF^*)$. Hence $h_*=bF^*Ry$. Moreover, $(aI_N+bFF^*)Ry=y$, so $y=aRy+bFF^*Ry$ and therefore $y-Fh_*=y-bFF^*Ry=aRy$.
Then
\begin{align*}
L(h_*)&=\frac1a\norm{aRy}_2^2+\frac1b\norm{bF^*Ry}_{\mathcal H}^2\\
&=ay'R^2y+by'RFF^*Ry=y'R(aI_N+bFF^*)Ry=y'Ry.
\end{align*}
By the covariance bound and $E[S|\mathcal A]=0$, $E[(S'y)^2|\mathcal A]\leq y'Ry=L(h_*)=\min_{h\in\mathcal H}L(h)$. This proves Equation~\eqref{equation:proof-gsw-ridge-bound}.
\end{proof}

The Fourier cutoff yields the following penalized approximation bound for near-critical Mat\'ern kernels.

\begin{lem}[Penalized Mat\'ern Approximation]\label{lemma:proof-penalized-matern-approximation}
Fix $c,s>0$.
There are constants $c_0,C>0$ and $n_0\geq3$, depending only on $(c,s)$, such that the following holds for every integer $d\geq1$ and $n\geq n_0$.
Set $\nu_n=c/\log n$ and $W_n=W_d^{\mathrm{Mat},\nu_n}$.
For every penalty $\lambda_n\in[n^{-1},c_0/\log n]$ and $h\in H_{\mathrm{av}}^s(\mathbb R^d)$,
\begin{equation}\label{equation:proof-penalized-matern-approximation}
\inf_{v\in\mc H_{W_n}}\left\{\norm{h-v}_{L^2([0,1]^d)}^2+\lambda_n\norm{v}_{\mc H_{W_n}}^2\right\}\leq C\norm{h}_{H_{\mathrm{av}}^s(\mathbb R^d)}^2(\lambda_n\log n)^{(2s/d)\wedge1}.
\end{equation}
If $X$ has a density on $[0,1]^d$ bounded above by $\overline p$, Equation \ref{equation:proof-penalized-matern-approximation} holds with the first term inside the braces replaced by $E[(h(X)-v(X))^2]$, with $C$ also depending on $\overline p$.
\end{lem}

\begin{proof}
First suppose that $s\leq d/2$.
Recall the inverse Fourier transform $\mathcal F^{-1}$.
For $R\geq1$, define $\widehat h_R(\omega)=\widehat h(\omega)\one\{1+\norm{\omega}_2^2/d\leq R\}$ and $h_R=\mathcal F^{-1}(\widehat h_R)$.
Since $h\in H_{\mathrm{av}}^s(\mathbb R^d)\subseteq L^2(\mathbb R^d)$, Plancherel's theorem gives $\widehat h\in L^2(\mathbb R^d)$.
Hence $\widehat h_R\in L^2(\mathbb R^d)$, and another application of Plancherel's theorem gives $h_R\in L^2(\mathbb R^d)$, so
\begin{align*}
&\norm{h-h_R}_{L^2([0,1]^d)}^2\leq\int_{\mathbb R^d}|h(x)-h_R(x)|^2dx=\int_{\mathbb R^d}|\widehat h(\omega)-\widehat h_R(\omega)|^2d\omega\\
&=\int_{1+\norm{\omega}_2^2/d>R}|\widehat h(\omega)|^2d\omega\leq R^{-s}\int_{\mathbb R^d}|\widehat h(\omega)|^2\left(1+\frac{\norm{\omega}_2^2}{d}\right)^sd\omega =R^{-s}\norm{h}_{H_{\mathrm{av}}^s(\mathbb R^d)}^2.
\end{align*}
The first equality is Plancherel's identity.
For the second inequality, note that $(1+\norm{\omega}_2^2/d)^{-s}\leq R^{-s}$ on the domain of integration.
Because $\widehat h_R$ is supported on $\{1+\norm{\omega}_2^2/d\leq R\}$, we have $h_R\in H^{d/2+\nu_n}(\mathbb R^d)=\mc H_{W_n}$.
Write $\Lambda_n=\Lambda_{d/2+\nu_n,d}$ and $\Theta_n=\Theta_{d/2+\nu_n,d}$.
Similarly, the RKHS norm of this approximation satisfies
\begin{align*}
\norm{h_R}_{\mc H_{W_n}}^2&=\Lambda_n\int_{1+\norm{\omega}_2^2/d\leq R}|\widehat h(\omega)|^2(1+\norm{\omega}_2^2)^{d/2+\nu_n}d\omega\\
&\leq\Theta_nR^{d/2+\nu_n-s}\int_{\mathbb R^d}|\widehat h(\omega)|^2\left(1+\frac{\norm{\omega}_2^2}{d}\right)^sd\omega=\Theta_nR^{d/2+\nu_n-s}\norm{h}_{H_{\mathrm{av}}^s(\mathbb R^d)}^2.
\end{align*}
The first equality follows from the Mat\'ern RKHS norm calculation in the secondary online appendix, noting that Fourier inversion gives $\mathcal Fh_R=\mathcal F(\mathcal F^{-1}(\widehat h_R))=\widehat h_R$ almost everywhere, and the definition of $\widehat h_R$ then restricts the integral to the displayed domain.
Since $1+\norm{\omega}_2^2\leq d(1+\norm{\omega}_2^2/d)$ and $d/2+\nu_n\geq d/2\geq s$, we have $(1+\norm{\omega}_2^2/d)^{d/2+\nu_n-s}\leq R^{d/2+\nu_n-s}$ on this domain, which gives the inequality.

Using $h_R$ as a witness in the infimum on the left hand side of Equation~\eqref{equation:proof-penalized-matern-approximation} gives the upper bound $\norm{h}_{H_{\mathrm{av}}^s(\mathbb R^d)}^2\{R^{-s}+\lambda_n\Theta_nR^{d/2+\nu_n-s}\}$.
The dimension-uniform bound from that calculation, applied with $\bar\nu=c$, gives $\Theta_n\leq C_c/\nu_n=(C_c/c)\log n$ uniformly over $d\geq1$.
Set $y=\lambda_n\log n$, so $\log n/n\leq y\leq c_0<1$ and $\lambda_n\Theta_n\leq(C_c/c)y$.
Decrease $c_0$ if necessary so that $\lambda_n\Theta_n\leq(C_c/c)c_0\leq1$, and take $R=(\lambda_n\Theta_n)^{-1/(d/2+\nu_n)}\geq1$.
Since $s/(d/2+\nu_n)\leq1$,
\begin{align*}
R^{-s}+\lambda_n\Theta_nR^{d/2+\nu_n-s}&=2(\lambda_n\Theta_n)^{s/(d/2+\nu_n)}\leq Cy^{s/(d/2+\nu_n)}\\
&=Cy^{2s/d}\exp\left(\left(\frac{2s}{d}-\frac{s}{d/2+\nu_n}\right)\log(1/y)\right)\leq Cy^{2s/d}.
\end{align*}
To see the final inequality, note that
\begin{equation*}
0\leq\frac{2s}{d}-\frac{s}{d/2+\nu_n}=\frac{4s\nu_n}{d(d+2\nu_n)}\leq\frac{4sc}{\log n}, \quad \log(1/y)\leq\log n.
\end{equation*}
Thus the exponential factor is at most $e^{4sc}$, which is absorbed into $C$.
This proves Equation~\eqref{equation:proof-penalized-matern-approximation} when $s\leq d/2$.

Now suppose $s>d/2$.
Only finitely many integer dimensions satisfy this inequality.
If this set is nonempty, define $\gamma_s=\min_{d\in\mathbb N:1\leq d<2s}(s-d/2)>0$.
For all sufficiently large $n$, $\nu_n\leq\gamma_s$, so $d/2+\nu_n\leq s$ simultaneously for these dimensions.
Taking $v=h$ and using the Mat\'ern RKHS norm formula in the secondary online appendix and monotonicity of the spectral weights gives zero approximation error and
\begin{equation*}
\lambda_n\norm{h}_{\mc H_{W_n}}^2\leq\lambda_n\Theta_n\norm{h}_{H_{\mathrm{av}}^s(\mathbb R^d)}^2\leq(C_c/c)y\norm{h}_{H_{\mathrm{av}}^s(\mathbb R^d)}^2.
\end{equation*}
This proves Equation~\eqref{equation:proof-penalized-matern-approximation} when $s>d/2$.

If $X$ has density bounded above by $\overline p$, then $E[(h(X)-v(X))^2]\leq\overline p\norm{h-v}_{L^2([0,1]^d)}^2$ for every $v$.
Since $\overline p\geq1$, the population criterion is at most $\overline p$ times the criterion on the left hand side of Equation~\eqref{equation:proof-penalized-matern-approximation}.
This proves the final statement.
\end{proof}

\begin{proof}[Proof of Corollary~\ref{corollary:kernel-pointwise}]
Recall $m(\psi)=E[\ylevel|\psi]$. By conditional Jensen's inequality, tower law, and Assumption~\ref{assumption:sampling}, $E[m(\psi)^2]\leq E[\ylevel^2]\leq\frac12\sum_{a=0}^1E[Y(a)^2]<\infty$, so $m\in L^2(P_\psi)$. Fix $\epsilon>0$. Density of $\mc H_K$ provides $g_\epsilon\in\mc H_K$ such that $E[(m(\psi)-g_\epsilon(\psi))^2]\leq\epsilon$. Using $g_\epsilon$ as a witness in Proposition~\ref{proposition:kernel-oracle} gives $\mc V_n(P)\leq4\epsilon/\varphi+8\norm{g_\epsilon}_{\mc H_K}^2/\{(1-\varphi)n\}$. Hence $\limsup_{n\to\infty}\mc V_n(P)\leq4\epsilon/\varphi$. Since $\epsilon$ was arbitrary and $\mc V_n(P)\geq0$ by Proposition~\ref{proposition:variance-decomposition}, the conclusion follows.
\end{proof}

\begin{proof}[Proof of Corollary~\ref{corollary:kernel-ball}]
For $P\in\mc P_K(B)$, take $g=m$ in Proposition~\ref{proposition:kernel-oracle}. The approximation term vanishes and $\norm{m}_{\mc H_K}\leq B$, so $\mc V_n(P)\leq8B^2/\{(1-\varphi)n\}$. The right hand side is independent of $P$, and taking the supremum proves the result.
\end{proof}

\begin{proof}[Proof of Corollary~\ref{corollary:matern-supercritical}]
Put $r=d/2+\nu\leq s$ and fix $P\in\mc P_d^s$.
Recall that $K=1+W$. By the sums-of-kernels theorem \citep[Theorem~5.4]{paulsen2016rkhs}, $\norm{f}_{\mc H_K}^2=\min_{f=a+w}\{a^2+\norm{w}_{\mc H_W}^2\}$, where $a\in\mathbb R$ and $w\in\mc H_W$. Taking $a=0$ and $w=g$ gives $\norm{g}_{\mc H_K}\leq\norm{g}_{\mc H_W}$ for every $g\in\mc H_W$. Since $r\leq s$, the Mat\'ern RKHS norm formula in the secondary online appendix and monotonicity of the spectral weights give
\begin{align*}
\norm{m}_{\mc H_W}^2&=\Lambda_{r,d}\int_{\mathbb R^d}|\widehat m(\omega)|^2(1+\norm{\omega}_2^2)^r d\omega\\
&\leq\Theta_{r,d}\int_{\mathbb R^d}|\widehat m(\omega)|^2\left(1+\frac{\norm{\omega}_2^2}{d}\right)^s d\omega=\Theta_{r,d}\norm{m}_{H_{\mathrm{av}}^s(\mathbb R^d)}^2\leq\Theta_{r,d}.
\end{align*}
Thus, $\norm{m}_{\mc H_K}^2\leq\Theta_{r,d}$. Apply Corollary~\ref{corollary:kernel-ball} with $B^2=\Theta_{r,d}$. The constant $8\Theta_{r,d}/(1-\varphi)$ depends only on $(d,\nu,\varphi)$, as claimed.
\end{proof}

\begin{proof}[Proof of Theorem~\ref{theorem:sobolev-gsw}]
Put $W_n=W_d^{\mathrm{Mat},\nu_n}$ and recall that $K_n=1+W_n$, where $\nu_n=c/\log n$.
Fix $P\in\mc P_d^s$ and set $\lambda_n=\max\{2\varphi/((1-\varphi)n),n^{-1}\}$.
By the sums-of-kernels theorem \citep[Theorem~5.4]{paulsen2016rkhs}, every $v\in\mc H_{W_n}$ belongs to $\mc H_{K_n}$ with $\norm{v}_{\mc H_{K_n}}\leq\norm{v}_{\mc H_{W_n}}$.
Proposition~\ref{proposition:kernel-oracle} therefore gives
\begin{align*}
\mc V_n(P)&\leq\frac4\varphi\inf_{v\in\mc H_{W_n}}\left\{E[(m(\psi)-v(\psi))^2]+\frac{2\varphi}{(1-\varphi)n}\norm{v}_{\mc H_{W_n}}^2\right\}\\
&\leq\frac4\varphi\inf_{v\in\mc H_{W_n}}\left\{E[(m(\psi)-v(\psi))^2]+\lambda_n\norm{v}_{\mc H_{W_n}}^2\right\}.
\end{align*}
For fixed $\varphi$, $n^{-1}\leq\lambda_n\leq C_\varphi/n$, so the penalty rate condition in Lemma~\ref{lemma:proof-penalized-matern-approximation} holds for all sufficiently large $n$.
Apply the population conclusion of that lemma in dimension $d$ at smoothness $s$, with $h=m$ and $X=\psi$.
Since $\norm{m}_{H_{\mathrm{av}}^s(\mathbb R^d)}\leq1$, we obtain
\begin{equation*}
\mc V_n(P)\leq C_{s,c,\varphi,\overline p}(\lambda_n\log n)^{(2s/d)\wedge1}\leq C_{s,c,\varphi,\overline p}\left(\frac{\log n}{n}\right)^{(2s/d)\wedge1}.
\end{equation*}
Thus, there is $n_0\geq3$ depending only on $(s,c,\varphi)$ such that the claimed bound holds for every $n\geq n_0$, uniformly over $d$ and $P$.
For $2\leq n<n_0$, Proposition~\ref{proposition:kernel-oracle} evaluated at $g=0$ and the density upper bound give $\mc V_n(P)\leq4E[m(\psi)^2]/\varphi\leq4\overline p/\varphi$.
Moreover, $(\log n/n)^{(2s/d)\wedge1}\geq\log n/n$, whose minimum over the finite set $2\leq n<n_0$ is strictly positive.
Enlarging $C_{s,c,\varphi,\overline p}$ therefore establishes the same bound for these remaining values of $n$.
Taking the supremum over $P\in\mc P_d^s$ proves the theorem.
\end{proof}

\begin{proof}[Proof of Theorem~\ref{theorem:sobolev-lower}]
Fix $\sigma\in\mc D_n$. Choose a smooth function $\eta:\mathbb R\to\mathbb R$ supported in a compact subinterval of $(0,1)$ and normalized so that $\int_0^1\eta(u)du=0$ and $\int_0^1\eta(u)^2du=1$. Set $q\equiv\lceil n^{1/d}\rceil$, $h\equiv q^{-1}$, and $K\equiv q^d$.
Choose one cell $C$ of the regular partition of $[0,1]^d$ into $K$ cubes of side length $h$, let $b$ be its lower-left corner, and define the whole-space bump
\begin{equation*}
g_h(x)\equiv\Pi_{\ell=1}^d\eta\left(\frac{x_\ell-b_\ell}{h}\right).
\end{equation*}
We claim that, for a constant $C_s<\infty$ depending only on $s$, $\norm{g_h}_{H_{\mathrm{av}}^s(\mathbb R^d)}^2\leq C_sh^{d-2s}$.
The function $g_h$ is smooth and supported in the interior of $C$. Under the Fourier convention above, $\widehat g_h(\omega)=(2\pi)^{-d/2}\int_{\mathbb R^d}e^{-i\omega'x}\Pi_{\ell=1}^d\eta((x_\ell-b_\ell)/h)dx$. Setting $x=b+hz$,
\begin{align*}
\widehat g_h(\omega)&=h^de^{-i\omega'b}(2\pi)^{-d/2}\int_{\mathbb R^d}e^{-ih\omega'z}\Pi_{\ell=1}^d\eta(z_\ell)dz =h^de^{-i\omega'b}\Pi_{\ell=1}^d\widehat\eta(h\omega_\ell).
\end{align*}
The second equality is from $e^{-ih\omega'z}=\Pi_{\ell=1}^de^{-ih\omega_\ell z_\ell}$ and Fubini. 
Since $|e^{-i\omega'b}|=1$,
\begin{align}
\norm{g_h}_{H_{\mathrm{av}}^s(\mathbb R^d)}^2&=h^{2d}\int_{\mathbb R^d}\Pi_{\ell=1}^d|\widehat\eta(h\omega_\ell)|^2\left(1+\frac{\norm{\omega}_2^2}{d}\right)^sd\omega\nonumber\\
&=h^d\int_{\mathbb R^d}\Pi_{\ell=1}^d|\widehat\eta(u_\ell)|^2\left(1+\frac{h^{-2}\norm{u}_2^2}{d}\right)^sdu.\label{equation:proof-sobolev-bump-scaling}
\end{align}

The final equality uses $u=h\omega$, so $d\omega=h^{-d}du$ and $\norm{\omega}_2^2=h^{-2}\norm{u}_2^2$.
By Plancherel's identity, $\int_{\mathbb R}|\widehat\eta(u)|^2du=\int_{\mathbb R}|\eta(u)|^2du=\int_0^1\eta(u)^2du=1$, where the second equality uses the support of $\eta$. Thus, $d\mu(u)=|\widehat\eta(u)|^2du$ defines a probability measure. Let $U_1,\ldots,U_d$ be iid with law $\mu$.
Since $h\leq1$, Equation~\eqref{equation:proof-sobolev-bump-scaling} gives
\begin{equation}\label{equation:proof-sobolev-bump-moment}
\norm{g_h}_{H_{\mathrm{av}}^s(\mathbb R^d)}^2\leq h^{d-2s}E\bigg[\bigg(1+\frac1d\sum_{\ell=1}^dU_\ell^2\bigg)^s\bigg].
\end{equation}

Suppose first that $0<s\leq1$. Jensen's inequality gives $E[(1+d^{-1}\sum_{\ell=1}^dU_\ell^2)^s]\leq(E[1+d^{-1}\sum_{\ell=1}^dU_\ell^2])^s=(1+E[U_1^2])^s$. Set $C_s=(1+E[U_1^2])^s$. Then the display above gives $\norm{g_h}_{H_{\mathrm{av}}^s(\mathbb R^d)}^2\leq C_sh^{d-2s}$. It remains to show $E[U_1^2]<\infty$. Integration by parts gives $\widehat{\eta'}(u)=iu\widehat\eta(u)$, and hence $E[U_1^2]=\int_{\mathbb R}u^2|\widehat\eta(u)|^2du=\int_{\mathbb R}|\widehat{\eta'}(u)|^2du=\norm{\eta'}_{L^2(\mathbb R)}^2<\infty$. The final equality is Plancherel's identity, and finiteness follows because $\eta$ is smooth and compactly supported.

Suppose next that $s\geq1$. By Jensen, $(1+d^{-1}\sum_{\ell=1}^dU_\ell^2)^s=d^{-s}(\sum_{\ell=1}^d(1+U_\ell^2))^s\leq d^{-1}\sum_{\ell=1}^d(1+U_\ell^2)^s$. Taking expectations and applying Equation~\eqref{equation:proof-sobolev-bump-moment} gives $\norm{g_h}_{H_{\mathrm{av}}^s(\mathbb R^d)}^2\leq E[(1+U_1^2)^s]h^{d-2s}$. Thus the same bound holds with $C_s=E[(1+U_1^2)^s]$, provided that $E[|U_1|^{2s}]<\infty$. To verify this, choose an integer $k\geq s$. Since $|u|^{2s}\leq1+|u|^{2k}$ and repeated integration by parts gives $\widehat{\eta^{(k)}}(u)=(iu)^k\widehat\eta(u)$, Plancherel's identity yields $E[|U_1|^{2s}]\leq1+\int_{\mathbb R}|u|^{2k}|\widehat\eta(u)|^2du=1+\norm{\eta^{(k)}}_{L^2(\mathbb R)}^2<\infty$. The final inequality follows because $\eta\in C^\infty(\mathbb R)$ is compactly supported.

Above we showed $\norm{g_h}_{H_{\mathrm{av}}^s(\mathbb R^d)}^2\leq C_sh^{d-2s}$. Then define $A=C_s^{-1/2}h^{s-d/2}$ and $m(x)=Ag_h(x)$ so that $\norm{m}_{H_{\mathrm{av}}^s(\mathbb R^d)}\leq1$. Let $\psi$ be uniform on $[0,1]^d$ and set $Y(1)=Y(0)=m(\psi)$. This law belongs to $\mc P_d^s$ and has $V^*(P)=0$.

Let $\mc E_C$ denote the event $\sum_{i=1}^n\one(\psi_i\in C)=1$, and let $i(C)$ denote the unique index in $C$ on this event. Since $m$ vanishes outside $C$, on $\mc E_C$ we have $\sum_{i=1}^nZ_im(\psi_i)=AZ_{i(C)}g_h(\psi_{i(C)})$ and hence $(\sum_{i=1}^nZ_im(\psi_i))^2=A^2g_h(\psi_{i(C)})^2$. 
Then pointwise we have $\one(\mc E_C)(\sum_{i=1}^nZ_im(\psi_i))^2=A^2\sum_{i=1}^n\one(\psi_i\in C)\prod_{j\neq i}\one(\psi_j\notin C)g_h(\psi_i)^2$. Proposition~\ref{proposition:variance-decomposition} therefore gives
\begin{equation*}
\mc V_n(\sigma,P)\geq\frac4nE[\one(\mc E_C)(\sum_{i=1}^nZ_im(\psi_i))^2] =\frac{4A^2}{n}\sum_{i=1}^n\int_Cg_h(x)^2dxP(\psi\notin C)^{n-1}.
\end{equation*}

The equality follows by taking expectations in the pointwise identity and using independence and uniformity of the sampled covariates. Uniformity gives $P(\psi\notin C)=1-1/K$. Since $K\geq n$, $(1-1/K)^{n-1}\geq(1-1/K)^{K-1}=(1+1/(K-1))^{-(K-1)}\geq e^{-1}$, where the final inequality uses $(1+1/r)^r\leq e$. Together with $\int_Cg_h(x)^2dx=h^d$, this shows $\mc V_n(\sigma,P)\geq4A^2h^d/e=4h^{2s}/(eC_s)$. Finally, $q\leq2n^{1/d}$, so $h^{2s}\geq2^{-2s}n^{-2s/d}$. Because the constructed law $P$ does not depend on $\sigma$, the argument in fact gives $\sup_{P\in\mc P_d^s}\inf_{\sigma\in\mc D_n}\mc V_n(\sigma,P)\geq4n^{-2s/d}/(eC_s2^{2s})$. The theorem follows from $\inf_\sigma\sup_P\mc V_n(\sigma,P)\geq\sup_P\inf_\sigma\mc V_n(\sigma,P)$, with $c_0=4/(eC_s2^{2s})$.
\end{proof}

\subsection{Proofs for Section~\ref{section:restricted-efficiency}}\label{appendix:proofs-restricted-efficiency}

We first record an oracle inequality separating working-model approximation from design error.

\begin{prop}[Approximation Oracle]\label{proposition:proof-restricted-kernel-comparison}
Let $K=1+W$ be a design kernel, where $W$ is PSD, let $\mc G$ be the $L^2(P_\psi)$ closure of $\mc H_K$, and let $m_{\mc G}$ be the $L^2(P_\psi)$ projection of $m$ onto $\mc G$.
Define $\Delta_{\mc G}(P)=4E[(m(\psi)-m_{\mc G}(\psi))^2]$.
For $\sigma=\mathrm{GSW}(K,\varphi)$ with $\varphi\in(0,1)$,
\begin{align*}
\mc V_n(\sigma,P)\leq\frac{\Delta_{\mc G}(P)}{\varphi}+\inf_{f\in\mc H_K}\left\{\frac4\varphi E[(m_{\mc G}(\psi)-f(\psi))^2]+\frac8{n(1-\varphi)}\norm{f}_{\mc H_K}^2\right\}.
\end{align*}
\end{prop}

\begin{proof}
Set $A=E[(m(\psi)-m_{\mc G}(\psi))^2]$, so that $\Delta_{\mc G}(P)=4A$.
Because $m-m_{\mc G}$ is orthogonal in $L^2(P_\psi)$ to $m_{\mc G}-f$ for every $f\in\mc H_K$, we have $E[(m(\psi)-f(\psi))^2]=A+E[(m_{\mc G}(\psi)-f(\psi))^2]$.
By Proposition~\ref{proposition:kernel-oracle},
\begin{align*}
\mc V_n(\sigma,P)&\leq\inf_{f\in\mc H_K}\left\{\frac4\varphi\left(A+E[(m_{\mc G}(\psi)-f(\psi))^2]\right)+\frac8{n(1-\varphi)}\norm{f}_{\mc H_K}^2\right\}\\
&=\frac{\Delta_{\mc G}(P)}{\varphi}+\inf_{f\in\mc H_K}\left\{\frac4\varphi E[(m_{\mc G}(\psi)-f(\psi))^2]+\frac8{n(1-\varphi)}\norm{f}_{\mc H_K}^2\right\}.
\end{align*}
\end{proof}

We next give the general result underlying the streamlined main-text theorem.
In addition to main effects and bivariate interactions, it allows a fixed block $J\subseteq[d]$ of priority covariates, such as baseline outcomes, to enter jointly and have a different smoothness order.
Write $d_0=|J|$ and $q=d-d_0$.
Let $\mc G_1$ be the $L^2(P_\psi)$ closure of functions $g(\psi)=a+g_J(\psi_J)+\sum_{j\notin J}g_j(\psi_j)$, and let $\mc G_2$ be the closure of functions $g(\psi)=a+g_J(\psi_J)+\sum_{\substack{j<\ell\\j,\ell\notin J}}g_{j\ell}(\psi_j,\psi_\ell)$.
For $k\in\{1,2\}$, write the canonical projection as
\begin{equation}\label{equation:appendix-canonical-low-order-decomposition}
m_{\mc G_k}(\psi)=a+m_J(\psi_J)+\sum_{j\notin J}m_j(\psi_j)+\sum_{\substack{j<\ell\\j,\ell\notin J}}m_{j\ell}(\psi_j,\psi_\ell).
\end{equation}
The bivariate terms are omitted for $k=1$.
The priority component and main effects have mean zero under Lebesgue measure, and each bivariate component integrates to zero in either argument.
For $S=(s_J,s)$ with $s_J,s>0$, let $\mc P_{d,k}^{S}$ contain the laws satisfying Assumption~\ref{assumption:sampling}, $E[m(\psi)^2]\leq1$, and
\begin{equation}\label{equation:appendix-low-order-sobolev-budget}
\norm{m_J}_{H_{\mathrm{av}}^{s_J}(\mr^{d_0})}^2+\sum_{j\notin J}\norm{m_j}_{H_{\mathrm{av}}^s(\mr)}^2+\sum_{\substack{j<\ell\\j,\ell\notin J}}\norm{m_{j\ell}}_{H_{\mathrm{av}}^s(\mr^2)}^2\leq1.
\end{equation}
Again, the bivariate terms are omitted for $k=1$.
When $J$ is empty, every priority-block term is omitted.
Fix $c>0$, set $\nu_n=c/\log n$, and write $W_{r,n}=W_r^{\mathrm{Mat},\nu_n}$.
When $J$ is nonempty, define
\begin{align*}
K_{n,1}(\psi,\psi')&=1+\frac12W_{d_0,n}(\psi_J,\psi_J')+\frac1{2q}\sum_{j\notin J}W_{1,n}(\psi_j,\psi_j'),\\
K_{n,2}(\psi,\psi')&=1+\frac12W_{d_0,n}(\psi_J,\psi_J')+\frac1{2\binom q2}\sum_{\substack{j<\ell\\j,\ell\notin J}}W_{2,n}(\psi_{j\ell},\psi_{j\ell}').
\end{align*}
When $J$ is empty, omit the priority term and normalize the remaining nonconstant kernel weight to one.
The canonical component spaces are mutually orthogonal and closed in Lebesgue $L^2$, and the density bounds make the Lebesgue and $L^2(P_\psi)$ norms equivalent.
Since each Mat\'ern RKHS is dense in its corresponding component $L^2$ space, $\mc G_k$ is exactly the corresponding structured function class written in canonical form.
For $k\in\{1,2\}$, let $m_{\mc G_k}$ be the $L^2(P_\psi)$ projection of $m$ onto $\mc G_k$ and write $\Delta_{\mc G_k}(P)=4E[(m(\psi)-m_{\mc G_k}(\psi))^2]$.

\begin{thm}[General Low-Order Nonparametric Balance]\label{theorem:proof-general-low-order-gsw}
Fix $n\geq2$, $d$, smoothness $S=(s_J,s)$ with $s_J,s>0$, $c>0$, and $k\in\{1,2\}$, where $d_0=|J|$ and $d-d_0\geq k$.
Define $\alpha_k=(2s/k)\wedge1$ and, when $J$ is nonempty, define $\alpha_J=(2s_J/d_0)\wedge1$.
When $J$ is empty, omit $\alpha_J$ and every priority-block term below.
\begin{enumerate}[label={\rm(\roman*)}]
\item Let $\sigma_n=\mathrm{GSW}(K_{n,k},\varphi_n)$, where $1-\varphi_n=1/2\wedge(n^{-1}(d-d_0)^k\log n)^{1/2}$.
Then, for some $C>0$ depending on $(d_0,s_J,s,c,\underline p,\overline p)$, uniformly over $P\in\mc P_{d,k}^{S}$,
\begin{equation}\label{equation:proof-general-low-order-gsw}
\mc V_n(\sigma_n,P)\leq\Delta_{\mc G_k}(P)+C\bigg(\left(\frac{\log n}{n}\right)^{\alpha_J}+\left(\frac{(d-d_0)^k\log n}{n}\right)^{\alpha_k}\bigg)^{1/2}.
\end{equation}
\item On the subclass of $\mc P_{d,k}^{S}$ with $m=m_{\mc G_k}$, any fixed $\varphi\in(0,1)$ yields the same bound for $\sigma_n=\mathrm{GSW}(K_{n,k},\varphi)$ with $\Delta_{\mc G_k}(P)=0$ and the outer square root removed. The constant may additionally depend on $\varphi$.
\end{enumerate}
\end{thm}

The next lemma collects the approximation calculation shared by the two structured designs.
When $J$ is empty, every priority-block term in the statement and proof is omitted.

\begin{lem}[Low-Order Mat\'ern Approximation]\label{lemma:proof-low-order-matern-approximation}
Fix $k\in\{1,2\}$ and smoothness $s_0,s>0$.
Let $g$ have the decomposition in Equation~\eqref{equation:appendix-canonical-low-order-decomposition}, with canonical components $g_J$, $g_j$, and $g_{j\ell}$. Suppose these satisfy the pooled smoothness budget in Equation~\eqref{equation:appendix-low-order-sobolev-budget} and also that $E[g(\psi)^2]\leq1$.
Set $\alpha_k=(2s/k)\wedge1$ and, when $d_0>0$, set $\alpha_J=(2s_0/d_0)\wedge1$.
There are constants $c_0,C>0$ and $n_0\geq1$ depending only on $(d_0,s_0,s,c,\underline p,\overline p)$ such that, for every $n\geq n_0$ satisfying $(d-d_0)^k\log n\leq c_0n$ and every $\lambda_n$ satisfying $n^{-1}\leq\lambda_n\leq c_0/((d-d_0)^k\log n)$,
\begin{equation*}
\inf_{f\in\mc H_{K_{n,k}}}\left\{E[(g(\psi)-f(\psi))^2]+\lambda_n\norm{f}_{\mc H_{K_{n,k}}}^2\right\}\leq C\bigl((\lambda_n\log n)^{\alpha_J}+((d-d_0)^k\lambda_n\log n)^{\alpha_k}\bigr).
\end{equation*}
\end{lem}

\begin{proof}
As in Equation~\eqref{equation:appendix-canonical-low-order-decomposition}, write $g=a+g_J+\sum_jg_j+\sum_{j<\ell}g_{j\ell}$, where every sum in the proof ranges over indices outside $J$.
We construct $f=a+f_J+\sum_jf_j+\sum_{j<\ell}f_{j\ell}\in\mc H_{K_{n,k}}$ as a witness for the inequality above.
The construction has three steps.
\begin{enumerate}
\item \emph{Componentwise approximation.} For each component $g_J$, $g_j$, and $g_{j\ell}$, we construct a raw Mat\'ern-RKHS approximation, denoted by $v_J$, $v_j$, or $v_{j\ell}$, satisfying the componentwise approximation profile in Lemma~\ref{lemma:proof-penalized-matern-approximation}.
\item \emph{Canonicalization.} The initial approximations may not satisfy the same canonical centering restrictions as their targets. For example, $\int_0^1g_{j\ell}(u,y)du=\int_0^1g_{j\ell}(x,u)du=0$ for every $x,y\in[0,1]$, whereas these equalities may fail for $v_{j\ell}$. We therefore transform the initial approximations using appropriate $L^2$ orthogonal projections to construct $f_J$, $f_j$, and $f_{j\ell}$ satisfying these restrictions, while preserving each approximation profile up to a constant factor.
\item \emph{Assembly.} We set $f=a+f_J+\sum_jf_j+\sum_{j<\ell}f_{j\ell}$. The component errors $g_J-f_J$, $g_j-f_j$, and $g_{j\ell}-f_{j\ell}$ lie in mutually orthogonal subspaces under Lebesgue measure. Hence the squared Lebesgue $L^2$ error is the sum of their squared component errors.
The assumed density upper bound in Assumption~\ref{assumption:sampling} transfers this to $E[(g(\psi)-f(\psi))^2]$. The RKHS norm of $f$ is also controlled by the corresponding component RKHS norms by the sum-kernel norm formula. Combining these bounds with the pooled Sobolev budget gives the required witness bound.
\end{enumerate}

\emph{Step 1: Componentwise approximation.}
For a penalty $\lambda_n$ and component kernel $W$, define the penalized approximation profile $\mathcal R_{\lambda_n,W}(h,v)=\norm{h-v}_2^2+\lambda_n\norm{v}_{\mc H_W}^2$, where the first norm is Lebesgue $L^2$ on the appropriate cube.
After decreasing $c_0$ if necessary, the restrictions on $\lambda_n$ imply that both $\lambda_n$ and $(d-d_0)^k\lambda_n$ belong to $[n^{-1},c_0/\log n]$.
Lemma~\ref{lemma:proof-penalized-matern-approximation} therefore allows us to choose a raw priority block witness $v_J$ and, when $k=2$, raw interaction witnesses $v_{j\ell}$ satisfying
\begin{align*}
\mathcal R_{\lambda_n,W_{d_0,n}}(g_J,v_J)&\leq C\norm{g_J}_{H_{\mathrm{av}}^{s_0}(\mr^{d_0})}^2(\lambda_n\log n)^{\alpha_J}\\
\mathcal R_{(d-d_0)^2\lambda_n,W_{2,n}}(g_{j\ell},v_{j\ell})&\leq C\norm{g_{j\ell}}_{H_{\mathrm{av}}^s(\mr^2)}^2((d-d_0)^2\lambda_n\log n)^{\alpha_k}.
\end{align*}
The first row is omitted when $J$ is empty, and the second is used only when $k=2$.
When $k=1$, applying the same lemma directly to each main effect gives
\begin{equation*}
\mathcal R_{(d-d_0)\lambda_n,W_{1,n}}(g_j,v_j)\leq C\norm{g_j}_{H_{\mathrm{av}}^s(\mr)}^2((d-d_0)\lambda_n\log n)^{\alpha_k}.
\end{equation*}
When $k=2$, the canonical decomposition of $g$ contains main effects $g_j\in H_{\mathrm{av}}^s(\mr)$, but the design RKHS only contains bivariate components $f_{j\ell}\in\mc H_{W_{2,n}}$.
List the coordinates outside $J$ as $j_1<\cdots<j_q$, where $q=d-d_0$, and set $\pi(j_1)=\pi(j_2)=(j_1,j_2)$ and $\pi(j_r)=(j_1,j_r)$ for $r=3,\ldots,q$.
For each $j\notin J$, we use the bivariate kernel indexed by $\pi(j)$ to approximate $g_j$.
Viewed as a function on $\mr^2$ that is constant in the other coordinate, the main effect need not be square-integrable.
We therefore fix a smooth compactly supported function $\chi$ equal to one on $[0,1]$ and define $g_j^\chi(x,y)=g_j(x)\chi(y)$ if $j$ is the first coordinate of $\pi(j)$ and $g_j^\chi(x,y)=\chi(x)g_j(y)$ if it is the second.
On the unit square, $g_j^\chi$ agrees with $g_j$ evaluated at coordinate $j$.
When $j$ is the first coordinate, Fourier factorization gives $\widehat{g_j^\chi}(\omega_1,\omega_2)=\widehat g_j(\omega_1)\widehat\chi(\omega_2)$, while $1+(\omega_1^2+\omega_2^2)/2\leq(1+\omega_1^2)(1+\omega_2^2)$.
Then
\begin{align*}
\norm{g_j^\chi}_{H_{\mathrm{av}}^s(\mr^2)}^2&=\int_{\mr^2}|\widehat g_j(\omega_1)|^2|\widehat\chi(\omega_2)|^2\left(1+\frac{\omega_1^2+\omega_2^2}{2}\right)^sd\omega_1d\omega_2\\
&\leq\norm{g_j}_{H_{\mathrm{av}}^s(\mr)}^2\norm{\chi}_{H^s(\mr)}^2\leq C\norm{g_j}_{H_{\mathrm{av}}^s(\mr)}^2.
\end{align*}
The display treats the case in which $j$ is the first coordinate, and the other case is identical after exchanging $\omega_1$ and $\omega_2$.
Then $g_j^\chi\in H_{\mathrm{av}}^s(\mr^2)$, so once again Lemma~\ref{lemma:proof-penalized-matern-approximation} supplies a raw witness $v_j$ satisfying approximation rate $\mathcal R_{(d-d_0)^2\lambda_n,W_{2,n}}(g_j^\chi,v_j)\leq C\norm{g_j}_{H_{\mathrm{av}}^s(\mr)}^2((d-d_0)^2\lambda_n\log n)^{\alpha_k}$.

\medskip

\emph{Step 2: Canonicalization.}
Fix one of the profiles displayed in Step 1, and let $h$, $v$, $W$, and $\lambda_{h,n}$ denote its target, raw witness, kernel, and penalty.
We transform $v$ into a witness $f$ satisfying the canonical restrictions with 
\begin{equation*}
\mathcal R_{\lambda_{h,n},W}(h,f)\leq C\mathcal R_{\lambda_{h,n},W}(h,v).
\end{equation*}
First, we impose the scalar centering restrictions.
The priority block and main effect witnesses $f_J$ and $f_j$ must satisfy $\int_{[0,1]^{d_0}}f_J=0$ and $\int_0^1f_j=0$ when $k=1$.
Consider the priority witness $v_J$.
Simply demeaning it by setting $f_J = v_J-\int_{[0,1]^{d_0}}v_J$ is not feasible, because nonzero constant functions do not belong to the RKHS $\mc H_{W_{d_0,n}}=H^{d_0/2+c/\log n}(\mr^{d_0})$.
To handle this technical complication, let $\chi_{d_0}$ be a fixed smooth compactly supported function equal to one on the cube, set $\mu_J=\int_{[0,1]^{d_0}}v_J$, and define $f_J=v_J-\mu_J\chi_{d_0}$.
Then $\int_{[0,1]^{d_0}}f_J=\int_{[0,1]^{d_0}}g_J=0$, and
\begin{align*}
\norm{g_J-v_J}_2^2&=\norm{g_J-f_J}_2^2+\mu_J^2+2\int_{[0,1]^{d_0}}(g_J-f_J)(f_J-v_J) =\norm{g_J-f_J}_2^2+\mu_J^2.
\end{align*}
The cross term vanishes because $f_J-v_J=-\mu_J$ on the cube and $g_J-f_J$ has mean zero.
In particular, $\mu_J^2\leq\norm{g_J-v_J}_2^2$ and the approximation error contracts: $\norm{g_J-f_J}_2^2\leq\norm{g_J-v_J}_2^2$.
The dimension-uniform Mat\'ern RKHS norm bound in the secondary online appendix gives $\norm{\chi_{d_0}}_{\mc H_{W_{d_0,n}}}^2\leq C\log n$.
Therefore,
\begin{align*}
\lambda_n\norm{f_J}_{\mc H_{W_{d_0,n}}}^2&=\lambda_n\norm{v_J-\mu_J\chi_{d_0}}_{\mc H_{W_{d_0,n}}}^2\leq2\lambda_n\norm{v_J}_{\mc H_{W_{d_0,n}}}^2+2\lambda_n\mu_J^2\norm{\chi_{d_0}}_{\mc H_{W_{d_0,n}}}^2\\
&\leq2\lambda_n\norm{v_J}_{\mc H_{W_{d_0,n}}}^2+C\lambda_n\log n\norm{g_J-v_J}_2^2.
\end{align*}
The first inequality is Young's and the second uses bounds above on $\mu_J$ and $\chi_{d_0}$.
Since $\lambda_n\log n\leq c_0$, the $L^2$ contraction and RKHS-norm bound give
\begin{align*}
\mathcal R_{\lambda_n,W_{d_0,n}}(g_J,f_J)&=\norm{g_J-f_J}_2^2+\lambda_n\norm{f_J}_{\mc H_{W_{d_0,n}}}^2\\
&\leq(1+C\lambda_n\log n)\norm{g_J-v_J}_2^2+2\lambda_n\norm{v_J}_{\mc H_{W_{d_0,n}}}^2\\
&\leq C\big(\norm{g_J-v_J}_2^2+\lambda_n\norm{v_J}_{\mc H_{W_{d_0,n}}}^2 \big)=C\mathcal R_{\lambda_n,W_{d_0,n}}(g_J,v_J).
\end{align*}
When $k=1$, the same construction with penalty $(d-d_0)\lambda_n$, produces centered $f_j$ from $v_j$ that preserve the approximation profile up to the same constant factor.

\emph{Interactions.}
We next canonicalize the bivariate interaction witnesses used when $k=2$.
Each $f_{j\ell}$ must satisfy $\int_0^1f_{j\ell}(x,z)dz=\int_0^1f_{j\ell}(z,y)dz=0$ for every $x,y\in[0,1]$.
To describe the required projection, let $u\in L^2([0,1]^2)$ and set $A_1u(x)=\int_0^1u(x,z)dz$, $A_2u(y)=\int_0^1u(z,y)dz$, and $A_0u=\int_0^1\int_0^1u(z,w)dzdw$.
The orthogonal projection of $u$ onto the subspace of $L^2([0,1]^2)$ satisfying the two canonical restrictions is
\begin{equation}\label{equation:proof-anova-projection}
(Qu)(x,y)=u(x,y)-A_1u(x)-A_2u(y)+A_0u.
\end{equation}
Now let $v\in\mc H_{W_{2,n}}$ be one of the raw witnesses from Step 1.
The same projection formula can be used to define an operator $(Qv)(x,y)=v(x,y)-A_1v(x)-A_2v(y)+A_0v$ for $v \in L^2(\mr^2)$.
However, $Qv\notin L^2(\mr^2)$ in general, let alone $\mc H_{W_{2,n}}$, because its marginal terms are constant in the unused coordinate.
To fix this, we use the compactly supported smooth function $\chi$ above and define the modified projection
\begin{equation*}
(\widetilde Qv)(x,y)=v(x,y)-A_1v(x)\chi(y)-A_2v(y)\chi(x)+(A_0v)\chi(x)\chi(y).
\end{equation*}
Because $\chi=1$ on $[0,1]$, the restrictions of $\widetilde Qv$ and $Qv$ to the unit square coincide.
We first show that $\widetilde Q$ is a bounded linear operator on $H^t(\mr^2)$, uniformly over $t\in[1,2]$.
Fourier factorization and $(1+\omega_1^2+\omega_2^2)^t\leq(1+\omega_1^2)^t(1+\omega_2^2)^t$ give $\norm{A_1v(x)\chi(y)}_{H^t(\mr^2)}\leq\norm{A_1v}_{H^t(\mr)}\norm{\chi}_{H^t(\mr)}$, and similarly for the other product terms above.
It therefore suffices to control the marginal Sobolev norms.
For $v\in H^t(\mr^2)$, we claim that $\norm{A_1v}_{H^t(\mr)}\leq\norm{v}_{H^t(\mr^2)}$.
We first verify the claim for $v\in C_c^\infty(\mr^2)$.
Writing $\mathcal F_x$ for the Fourier transform in the first coordinate, Fubini's theorem gives $\widehat{A_1v}(\omega_1)=\int_0^1\mathcal F_xv(\omega_1,z)dz$.
Hence by Jensen's inequality
\begin{equation*}
\norm{A_1v}_{H^t(\mr)}^2=\int_{\mr}(1+\omega_1^2)^t\bigg|\int_0^1\mathcal F_xv(\omega_1,z)dz\bigg|^2d\omega_1\leq\int_{\mr}\int_0^1(1+\omega_1^2)^t|\mathcal F_xv(\omega_1,z)|^2dzd\omega_1.
\end{equation*}
For each $\omega_1$, Plancherel's identity in the second coordinate gives $\int_{\mr}|\mathcal F_xv(\omega_1,z)|^2dz=\int_{\mr}|\mathcal F_z\mathcal F_xv(\omega_1,\omega_2)|^2d\omega_2=\int_{\mr}|\widehat v(\omega_1,\omega_2)|^2d\omega_2$.
Using this identity and extending the $z$-integral in the preceding display to $\mr$, we have
\begin{align*}
\int_{\mr}\int_0^1(1+\omega_1^2)^t &|\mathcal F_xv(\omega_1,z)|^2dzd\omega_1 \leq\int_{\mr^2}(1+\omega_1^2)^t|\widehat v(\omega_1,\omega_2)|^2d\omega_1d\omega_2\\
&\leq\int_{\mr^2}(1+\omega_1^2+\omega_2^2)^t|\widehat v(\omega_1,\omega_2)|^2d\omega_1d\omega_2=\norm{v}_{H^t(\mr^2)}^2.
\end{align*}
For general $v\in H^t(\mr^2)$, choose $v_m\in C_c^\infty(\mr^2)$ with $v_m\to v$ in $H^t(\mr^2)$.
The preceding bound makes $A_1v_m$ Cauchy in $H^t(\mr)$, while Cauchy--Schwarz gives $\norm{A_1(v_m-v)}_{L^2(\mr)}\leq\norm{v_m-v}_{L^2(\mr^2)}\to0$.
Its $H^t(\mr)$ limit is therefore $A_1v$, proving the claim by density.
The same bound holds for $A_2v$.
Moreover, $A_0v$ is the coefficient of the $L^2([0,1]^2)$ projection onto the constants, so $|A_0v|\leq\norm{v}_{L^2([0,1]^2)}\leq\norm{v}_{L^2(\mr^2)}\leq\norm{v}_{H^t(\mr^2)}$.
The final inequality follows from Plancherel's identity and $(1+\norm{\omega}_2^2)^t\geq1$.

More explicitly, the same factorization gives $\norm{A_2v(y)\chi(x)}_{H^t(\mr^2)}\leq\norm{A_2v}_{H^t(\mr)}\norm{\chi}_{H^t(\mr)}$ and $\norm{A_0v\chi(x)\chi(y)}_{H^t(\mr^2)}\leq|A_0v|\norm{\chi}_{H^t(\mr)}^2$.
Since $\norm{\chi}_{H^t(\mr)}\leq\norm{\chi}_{H^2(\mr)}$ for $t\in[1,2]$, the preceding bounds imply $\norm{\widetilde Qv}_{H^t(\mr^2)}\leq C\norm{v}_{H^t(\mr^2)}$ uniformly over $t\in[1,2]$.

Set $r_n=1+c/\log n$, the Sobolev order of $\mc H_{W_{2,n}}$.
For large enough $n$, $r_n\in[1,2]$.
Writing $\Lambda_n=\Lambda_{r_n,2}$, the Mat\'ern RKHS norm formula in the secondary online appendix and the preceding bound at $t=r_n$ give for $v\in\mc H_{W_{2,n}}$
\begin{align}
\norm{\widetilde Qv}_{\mc H_{W_{2,n}}}^2&=\Lambda_n\norm{\widetilde Qv}_{H^{r_n}(\mr^2)}^2 \leq C\Lambda_n\norm{v}_{H^{r_n}(\mr^2)}^2=C\norm{v}_{\mc H_{W_{2,n}}}^2<\infty.\label{equation:proof-whole-space-anova-bound}
\end{align}
Thus $\widetilde Qv$ is a valid witness in $\mc H_{W_{2,n}}$, and its RKHS penalty is controlled by that of $v$.
Because $Qg_{j\ell}=g_{j\ell}$, orthogonal projection gives the contraction property
\begin{equation*}
\norm{g_{j\ell}-\widetilde Qv_{j\ell}}_{L^2([0,1]^2)}=\norm{Q(g_{j\ell}-v_{j\ell})}_{L^2([0,1]^2)}\leq\norm{g_{j\ell}-v_{j\ell}}_{L^2([0,1]^2)}.
\end{equation*}
Combining the $L^2$ and Sobolev bounds above, $f_{j\ell}=\widetilde Qv_{j\ell}$ satisfies the canonical restrictions and has profile $\mathcal R_{(d-d_0)^2\lambda_n,W_{2,n}}(g_{j\ell},f_{j\ell})\leq C\mathcal R_{(d-d_0)^2\lambda_n,W_{2,n}}(g_{j\ell},v_{j\ell})$.

\medskip

\emph{Lifted main effects.}
Recall that Step 1 assigned $g_j$ to the pair $\pi(j)$, lifted it to $g_j^\chi$, and produced a raw bivariate witness $v_j\in\mc H_{W_{2,n}}$ for the kernel acting on that pair.
Suppose first that $j$ is the first coordinate of $\pi(j)$.
For $(x,y)\in[0,1]^2$, $g_j^\chi(x,y)=g_j(x)$ and $\int_0^1g_j(x)dx=0$, so the target is a centered function of the first coordinate alone.
The orthogonal projection onto this subspace and its smooth modification on $\mr^2$ are
\begin{equation*}
(Q_1v_j)(x,y)=A_1v_j(x)-A_0v_j, \qquad (\widetilde Q_1v_j)(x,y)=A_1v_j(x)\chi(y)-A_0v_j\chi(x)\chi(y).
\end{equation*}
Because $\chi=1$ on $[0,1]$, the restrictions of $\widetilde Q_1v_j$ and $Q_1v_j$ to the unit square coincide.
The calculations above show that $\widetilde Q_1v_j\in\mc H_{W_{2,n}}$ and $\norm{\widetilde Q_1v_j}_{\mc H_{W_{2,n}}}\leq C\norm{v_j}_{\mc H_{W_{2,n}}}$.
Moreover, $Q_1g_j^\chi=g_j^\chi$ on the unit square, so
\begin{equation*}
\norm{g_j^\chi-\widetilde Q_1v_j}_{L^2([0,1]^2)}=\norm{Q_1(g_j^\chi-v_j)}_{L^2([0,1]^2)}\leq\norm{g_j^\chi-v_j}_{L^2([0,1]^2)}.
\end{equation*}
If $j$ is the second coordinate of $\pi(j)$, use instead $(Q_2v_j)(x,y)=A_2v_j(y)-A_0v_j$ and $(\widetilde Q_2v_j)(x,y)=A_2v_j(y)\chi(x)-A_0v_j\chi(x)\chi(y)$.
Set $f_j=\widetilde Q_1v_j$ in the first case and $f_j=\widetilde Q_2v_j$ in the second.
Then $f_j$ satisfies the required centering restriction and $\mathcal R_{(d-d_0)^2\lambda_n,W_{2,n}}(g_j^\chi,f_j)\leq C\mathcal R_{(d-d_0)^2\lambda_n,W_{2,n}}(g_j^\chi,v_j)$.

\medskip

\emph{Step 3: Assembly.}
Recall that we must bound $E[(g(\psi)-f(\psi))^2]+\lambda_n\norm{f}_{\mc H_{K_{n,k}}}^2$ for the approximation $f=a+f_J+\sum_jf_j+\sum_{j<\ell}f_{j\ell}$ constructed above.
We first consider $k=2$ and bound its RKHS norm.
For each pair $(j,\ell)$, the interaction witness $f_{j\ell}$ belongs to $\mc H_{W_{2,n}}$ and approximates the component $g_{j\ell}$.
Whenever $\pi(s)=(j,\ell)$, the main-effect witness $f_s$ was constructed in the same bivariate component RKHS $\mc H_{W_{2,n}}$ and approximates $g_s$ through the lifted target $g_s^\chi$.
By the sums-of-kernels theorem \citep[Theorem~5.4]{paulsen2016rkhs}, $\norm{f}_{\mc H_{K_{n,2}}}^2$ is the minimum of the weighted component norms over all representations $f=a'+h_J+\sum_{j<\ell}h_{j\ell}$ with one $h_{j\ell}\in\mc H_{W_{2,n}}$ for each pair.

To obtain a representation that we can insert into this minimum, we collect all witnesses constructed in the same pair-indexed bivariate RKHS.
To that end, for every $j<\ell$, define the witness assigned to the pair $(j,\ell)$ by $F_{j\ell}=f_{j\ell}+\sum_{s\notin J}\one(\pi(s)=(j,\ell))f_s$.
Because every main effect is assigned once, $\sum_{j<\ell}F_{j\ell}=\sum_{j<\ell}f_{j\ell}+\sum_jf_j$.
The triangle inequality followed by Jensen's inequality gives $\norm{h_1+\cdots+h_r}^2\leq r^2(r^{-1}\sum_{b=1}^r\norm{h_b})^2\leq r\sum_{b=1}^r\norm{h_b}^2$.
Since each $F_{j\ell}$ contains at most three terms, $\sum_{j<\ell}\norm{F_{j\ell}}_{\mc H_{W_{2,n}}}^2\leq3(\sum_{j<\ell}\norm{f_{j\ell}}_{\mc H_{W_{2,n}}}^2+\sum_j\norm{f_j}_{\mc H_{W_{2,n}}}^2)$.

Writing $q=d-d_0$, in the defining sum for $K_{n,2}$ each bivariate kernel has coefficient $1/(2\binom q2)$ when $J$ is nonempty and $1/\binom q2$ when $J$ is empty.
Thus its reciprocal coefficient is at most $2\binom q2$.
Evaluating the minimum defining the sum-kernel RKHS norm at the approximation witnesses $a'=a$, $h_J=f_J$, and $h_{j\ell}=F_{j\ell}$ gives
\begin{equation*}
\begin{aligned}
\norm{f}_{\mc H_{K_{n,2}}}^2&\leq\min_{f=a'+h_J+\sum_{j<\ell}h_{j\ell}}\bigg\{(a')^2+2\norm{h_J}_{\mc H_{W_{d_0,n}}}^2+2\binom q2\sum_{j<\ell}\norm{h_{j\ell}}_{\mc H_{W_{2,n}}}^2\bigg\}\\
&\leq a^2+2\norm{f_J}_{\mc H_{W_{d_0,n}}}^2+6\binom q2\bigg(\sum_j\norm{f_j}_{\mc H_{W_{2,n}}}^2+\sum_{j<\ell}\norm{f_{j\ell}}_{\mc H_{W_{2,n}}}^2\bigg).
\end{aligned}
\end{equation*}
For $k=1$, the same theorem applied to the representation $f=a+f_J+\sum_jf_j$ gives $\norm{f}_{\mc H_{K_{n,1}}}^2\leq a^2+2\norm{f_J}_{\mc H_{W_{d_0,n}}}^2+2q\sum_j\norm{f_j}_{\mc H_{W_{1,n}}}^2$.
Since $6\binom q2\leq3q^2$ and $2q\leq3q$, in either case
\begin{equation}\label{equation:proof-low-order-quotient-norm}
\begin{aligned}
\norm{f}_{\mc H_{K_{n,k}}}^2&\leq C\bigg(a^2+\norm{f_J}_{\mc H_{W_{d_0,n}}}^2\\
&\quad+q^k\bigg(\sum_j\norm{f_j}_{\mc H_{W_{k,n}}}^2+\sum_{j<\ell}\norm{f_{j\ell}}_{\mc H_{W_{2,n}}}^2\bigg)\bigg).
\end{aligned}
\end{equation}
As throughout the proof, the priority term is omitted when $J$ is empty and the bivariate sum is omitted when $k=1$.
The density upper bound in Assumption~\ref{assumption:sampling} gives $E[(g(\psi)-f(\psi))^2]\leq\overline p\norm{g-f}_{L^2([0,1]^d)}^2$.
For $k=2$, canonical orthogonality and the identity of $g_j^\chi$ and $g_j$ on the assigned unit square give $\norm{g-f}_{L^2([0,1]^d)}^2=\norm{g_J-f_J}_2^2+\sum_j\norm{g_j^\chi-f_j}_2^2+\sum_{j<\ell}\norm{g_{j\ell}-f_{j\ell}}_2^2$.
Combining these facts with Equation~\eqref{equation:proof-low-order-quotient-norm} collects the component errors and RKHS penalties into the profiles constructed above:
\begin{align*}
&E[(g(\psi)-f(\psi))^2]+\lambda_n\norm{f}_{\mc H_{K_{n,2}}}^2\leq C\lambda_n a^2+C\big(\norm{g_J-f_J}_2^2+\lambda_n\norm{f_J}_{\mc H_{W_{d_0,n}}}^2\big)\\
&+C\sum_j\big(\norm{g_j^\chi-f_j}_2^2+q^2\lambda_n\norm{f_j}_{\mc H_{W_{2,n}}}^2\big)+C\sum_{j<\ell}\big(\norm{g_{j\ell}-f_{j\ell}}_2^2+q^2\lambda_n\norm{f_{j\ell}}_{\mc H_{W_{2,n}}}^2\big).
\end{align*}
The parentheticals above are, respectively, $\mathcal R_{\lambda_n,W_{d_0,n}}(g_J,f_J)$, $\mathcal R_{q^2\lambda_n,W_{2,n}}(g_j^\chi,f_j)$, and $\mathcal R_{q^2\lambda_n,W_{2,n}}(g_{j\ell},f_{j\ell})$.
Continuing, by step 1 this is bounded above by
\begin{equation*}
C\lambda_n a^2+C(\lambda_n\log n)^{\alpha_J}\norm{g_J}_{H_{\mathrm{av}}^{s_0}(\mr^{d_0})}^2+C(q^2\lambda_n\log n)^{\alpha_2}\bigg(\sum_j\norm{g_j}_{H_{\mathrm{av}}^s(\mr)}^2+\sum_{j<\ell}\norm{g_{j\ell}}_{H_{\mathrm{av}}^s(\mr^2)}^2\bigg).
\end{equation*}
For $k=1$, the same calculation omits the interaction sum and replaces $(g_j^\chi,W_{2,n},q^2,\alpha_2)$ by $(g_j,W_{1,n},q,\alpha_1)$.
The pooled Sobolev budget bounds the right hand side by $C\lambda_n a^2+C((\lambda_n\log n)^{\alpha_J}+((d-d_0)^k\lambda_n\log n)^{\alpha_k})$.

It remains to control the intercept.
Because the canonical components are orthogonal, the density lower bound gives $a^2\leq\norm{g}_{L^2([0,1]^d)}^2\leq\underline p^{-1}E[g(\psi)^2]\leq\underline p^{-1}$.
Set $x=(d-d_0)^k\lambda_n\log n$.
After decreasing $c_0$ if necessary, $0<x\leq c_0<1$, and $(d-d_0)^k\geq1$ and $\alpha_k\leq1$ give $\lambda_n\leq x\leq x^{\alpha_k}$.
Hence the nonpriority profile absorbs $\lambda_n a^2$.
The componentwise bound and intercept calculation establish the required witness bound.
Taking the infimum over $f\in\mc H_{K_{n,k}}$ proves the claim.
\end{proof}

\begin{proof}[Proof of Theorem~\ref{theorem:proof-general-low-order-gsw}]
Define $a_n=\log n/n$ and $b_n=(d-d_0)^k\log n/n$.
We use the convention that the priority exponent and every priority term are omitted when $J$ is empty.
For part (i), put $\delta_n=1-\varphi_n=(1/2)\wedge b_n^{1/2}$.
Let $c_0$ and $n_0$ be as in Lemma~\ref{lemma:proof-low-order-matern-approximation}, decrease $c_0$ so that $c_0\leq1/2$, and set $b_0=(1/4)\wedge c_0^2$.
If $n<n_0$ or $b_n>b_0$, Proposition~\ref{proposition:kernel-oracle} evaluated at zero gives $\mc V_n(\sigma_n,P)\leq4E[m(\psi)^2]/\varphi_n\leq8$.
In this case, $b_n$ is bounded below by a positive constant depending only on the lemma constants: this follows from $b_n>b_0$ in the second case and from $b_n\geq\log n/n$ over the finite set $2\leq n<n_0$ in the first.
Enlarging $C$ therefore gives the claimed bound.
It remains to consider $n\geq n_0$ and $b_n\leq b_0$, for which $\delta_n=b_n^{1/2}\leq1/2$.
Because $0\in\mc G_k$ and $m_{\mc G_k}$ is a projection,
\begin{equation}
\Delta_{\mc G_k}(P)=4E[(m(\psi)-m_{\mc G_k}(\psi))^2]\leq4E[m(\psi)^2]\leq4.
\end{equation}
Projection contraction also gives $E[m_{\mc G_k}(\psi)^2]\leq E[m(\psi)^2]\leq1$, so the approximation lemma applies with $g=m_{\mc G_k}$.
Since $1-\varphi_n=\delta_n$ and $\varphi_n=1-\delta_n\geq1/2$, the part of $\Delta_{\mc G_k}(P)/\varphi_n$ exceeding $\Delta_{\mc G_k}(P)$ is at most $4\delta_n/(1-\delta_n)\leq8\delta_n$.
Proposition~\ref{proposition:proof-restricted-kernel-comparison} and Lemma~\ref{lemma:proof-low-order-matern-approximation} with $\lambda_n=(n\delta_n)^{-1}$ therefore yield
\begin{align}
\mc V_n(\sigma_n,P)&\leq\Delta_{\mc G_k}(P)+8\delta_n+8\inf_{f\in\mc H_{K_{n,k}}}\left\{E[(m_{\mc G_k}(\psi)-f(\psi))^2]+\frac{\norm{f}_{\mc H_{K_{n,k}}}^2}{n\delta_n}\right\}\notag\\
&\leq\Delta_{\mc G_k}(P)+C\left\{b_n^{1/2}+(a_n/\delta_n)^{\alpha_J}+(b_n/\delta_n)^{\alpha_k}\right\}.\label{equation:proof-low-order-approximation-bound}
\end{align}
The lemma's penalty conditions hold because $\delta_n\leq1$ gives $(n\delta_n)^{-1}\geq n^{-1}$, while $(d-d_0)^k\log n(n\delta_n)^{-1}=b_n/\delta_n=b_n^{1/2}\leq c_0$.
Moreover, $b_n\geq a_n$, so $(a_n/\delta_n)^{\alpha_J}\leq a_n^{\alpha_J/2}$.
Because $0<b_n\leq1$ and $\alpha_k\leq1$, we have $b_n^{1/2}=b_n^{\alpha_k/2}b_n^{(1-\alpha_k)/2}\leq b_n^{\alpha_k/2}$, while $(b_n/\delta_n)^{\alpha_k}=b_n^{\alpha_k/2}$.
Thus the remainder after $\Delta_{\mc G_k}(P)$ in Equation~\eqref{equation:proof-low-order-approximation-bound} is at most $C(a_n^{\alpha_J/2}+b_n^{\alpha_k/2})$.
By Cauchy--Schwarz, $a_n^{\alpha_J/2}+b_n^{\alpha_k/2}\leq\sqrt{2}(a_n^{\alpha_J}+b_n^{\alpha_k})^{1/2}$, and absorbing $\sqrt{2}$ into $C$ gives the claimed bound.
The constant is uniform over $P\in\mc P_{d,k}^{S}$, which proves part (i).

For part (ii), fix $\varphi\in(0,1)$ and suppose $m=m_{\mc G_k}$.
Proposition~\ref{proposition:kernel-oracle} gives
\begin{equation}\label{equation:proof-low-order-exact-oracle}
\mc V_n(\sigma_n,P)\leq\frac4\varphi\inf_{f\in\mc H_{K_{n,k}}}\left\{E[(m(\psi)-f(\psi))^2]+\frac{2\varphi}{n(1-\varphi)}\norm{f}_{\mc H_{K_{n,k}}}^2\right\}.
\end{equation}
Let $\lambda_n$ be the maximum of $2\varphi/\{n(1-\varphi)\}$ and $1/n$.
The approximation functional is nondecreasing in its penalty, so the right hand side of Equation~\eqref{equation:proof-low-order-exact-oracle} is at most the same expression evaluated at $\lambda_n$.
Write $L_\varphi=1\vee(2\varphi/(1-\varphi))$, so that $\lambda_n=L_\varphi/n$.
If $n\geq n_0$ and $L_\varphi b_n\leq c_0$, Lemma~\ref{lemma:proof-low-order-matern-approximation} applies and gives the claimed bound because $E[m(\psi)^2]\leq1$.
Otherwise, evaluating Equation~\eqref{equation:proof-low-order-exact-oracle} at zero gives $\mc V_n(\sigma_n,P)\leq4/\varphi$.
As above, $b_n$ is then bounded below by a positive constant depending only on $(n_0,c_0,\varphi)$, so enlarging the constant gives the same bound.
The constant may additionally depend on the fixed $\varphi$, which proves part (ii).
\end{proof}

\begin{proof}[Proof of Theorem~\ref{theorem:low-order-gsw}]
Apply Theorem~\ref{theorem:proof-general-low-order-gsw} with $J=\emptyset$ and $k=2$.
Then $d_0=0$, $K_{n,2}=K_n$, and $\alpha_2=s\wedge1$, so Equation~\eqref{equation:proof-general-low-order-gsw} reduces to Equation~\eqref{equation:low-order-gsw-approximation}.
\end{proof}

We next turn to the paired design.

\begin{proof}[Proof of Theorem~\ref{thm:matched-gsw}]
Let $U_M$ denote any exogenous randomness, independent of the data, used to break ties when constructing the matching, and set $\mc A_M=\sigma(\psi_{1:n},U_M)$.
Order the pairs and the two indices within each pair $\mc A_M$-measurably, and let $D_M$ be the $(n/2)\times n$ matrix whose $p$th row is $e_{i_p}'-e_{j_p}'$.
Write $A_n=[K(\psi_i,\psi_j)]_{i,j=1}^n$. Since $A_n$ is positive semidefinite, $G_n^M=D_MA_nD_M'$ is also positive semidefinite.
If $\widehat\kappa_n>0$, then $\Gamma_n^M\succeq\varphi I_{n/2}$ and $(\Gamma_n^M)_{pp}=\varphi+(1-\varphi)(G_n^M)_{pp}/\widehat\kappa_n\leq1$ for every $p$.
If $\widehat\kappa_n=0$, then $\Gamma_n^M=\varphi I_{n/2}$ has the same properties.
Let $B_M=(\Gamma_n^M)^{1/2}$.
Its $p$th column has squared norm $(\Gamma_n^M)_{pp}\leq1$, so Lemma~6.1 of \citet{harshaw2024gsw} gives $E[T|\mc A_M]=0$, while their Theorem~6.3 gives $\cov(B_MT|\mc A_M)\preceq I_{n/2}$.
Since $B_M$ is invertible,
\begin{equation}\label{equation:proof-matched-gsw-covariance}
E[T|\mc A_M]=0, \qquad E[TT'|\mc A_M]=\cov(T|\mc A_M)\preceq(\Gamma_n^M)^{-1}.
\end{equation}
Since $D_M$ is $\mc A_M$-measurable and $Z=D_M'T$, the tower law gives
\begin{equation*}
E[Z|\psi_{1:n}]=E\big[E[D_M'T|\mc A_M]\big|\psi_{1:n}\big]=E\big[D_M'E[T|\mc A_M]\big|\psi_{1:n}\big]=0.
\end{equation*}
By definition, the auxiliary GSW randomization is independent of the data and $U_M$.
Thus $U_M$ and the GSW randomization are jointly independent of the data.
Since $Z$ is a measurable function of $\psi_{1:n}$ and these random variables, $Z\indep(Y_i(0),Y_i(1))_{i=1}^n|\psi_{1:n}$.
Therefore, $\mathrm{GSW}_M(K,\varphi)$ is admissible.

Let $F_n:\mc H_K\to\mr^n$ be the evaluation operator $(F_nf)_i=f(\psi_i)$ and define $F_n^M=D_MF_n$, so $(F_n^Mf)_p=f(\psi_{i_p})-f(\psi_{j_p})$.
As in the proof of Proposition~\ref{proposition:kernel-oracle}, we first identify the adjoint $(F_n^M)^*:\mr^{n/2}\to\mc H_K$.
Fix $u\in\mr^{n/2}$ and $g\in\mc H_K$.
The reproducing property gives $\langle g,(F_n^M)^*u\rangle_{\mc H_K}=\langle F_n^Mg,u\rangle_2=\sum_{p=1}^{n/2}u_p(g(\psi_{i_p})-g(\psi_{j_p}))=\langle g,\sum_{p=1}^{n/2}u_p(K(\psi_{i_p},\cdot)-K(\psi_{j_p},\cdot))\rangle_{\mc H_K}$.
Since this holds for every $g\in\mc H_K$, $(F_n^M)^*u=\sum_{p=1}^{n/2}u_p(K(\psi_{i_p},\cdot)-K(\psi_{j_p},\cdot))$.
Taking $u=e_q$ gives $(F_n^M(F_n^M)^*)_{pq}=K(\psi_{i_p},\psi_{i_q})-K(\psi_{i_p},\psi_{j_q})-K(\psi_{j_p},\psi_{i_q})+K(\psi_{j_p},\psi_{j_q})=(G_n^M)_{pq}$.
Hence $F_n^M(F_n^M)^*=G_n^M$.

When $\widehat\kappa_n>0$, Equation~\eqref{equation:proof-matched-gsw-covariance} and the definition of $\Gamma_n^M$ give
\begin{equation*}
\cov(T|\mc A_M)\preceq\left(\varphi I_{n/2}+\frac{1-\varphi}{\widehat\kappa_n}F_n^M(F_n^M)^*\right)^{-1}.
\end{equation*}
Apply Lemma~\ref{lemma:proof-gsw-ridge} with $\mathcal H=\mc H_K$, $\mathcal A=\mc A_M$, $F=F_n^M$, $S=T$, $a=\varphi$, and $b=(1-\varphi)/\widehat\kappa_n$.
Evaluating the resulting minimum at any $f\in\mc H_K$ gives, for every $\mc A_M$-measurable $y$,
\begin{equation}\label{equation:proof-matched-ridge-bound}
E[(T'y)^2|\mc A_M]\leq\frac1\varphi\norm{y-F_n^Mf}_2^2+\frac{\widehat\kappa_n}{1-\varphi}\norm{f}_{\mc H_K}^2.
\end{equation}
When $\widehat\kappa_n=0$, the positive semidefinite matrix $G_n^M$ has zero diagonal and hence equals zero.
Therefore, $F_n^M(F_n^M)^*=0$.
For every $u\in\mr^{n/2}$, this gives $\norm{(F_n^M)^*u}_{\mc H_K}^2=u'F_n^M(F_n^M)^*u=0$, so $(F_n^M)^*=0$ and $F_n^M=0$.
Since $\Gamma_n^M=\varphi I_{n/2}$, Equation~\eqref{equation:proof-matched-gsw-covariance} gives
\begin{equation*}
E[(T'y)^2|\mc A_M]\leq\frac1\varphi\norm{y}_2^2=\frac1\varphi\norm{y-F_n^Mf}_2^2+\frac{\widehat\kappa_n}{1-\varphi}\norm{f}_{\mc H_K}^2.
\end{equation*}
Thus, Equation~\eqref{equation:proof-matched-ridge-bound} also holds in this case.
Let $m_{1:n}=(m(\psi_1),\ldots,m(\psi_n))'$.
Proposition~\ref{proposition:variance-decomposition}, the identity $Z=D_M'T$, and the tower law give
\begin{equation}
\mc V_n(\sigma,P)=\frac4nE[(T'D_Mm_{1:n})^2].
\end{equation}
For $f\in\mc H_K$, write $f_{1:n}=(f(\psi_1),\ldots,f(\psi_n))'$ and apply Equation~\eqref{equation:proof-matched-ridge-bound} with $y=D_Mm_{1:n}$.
Taking expectations and then the infimum over deterministic $f$ gives
\begin{align}
\mc V_n(\sigma,P)&\leq\inf_{f\in\mc H_K}E\left[\frac4{\varphi n}\norm{D_M(m_{1:n}-f_{1:n})}_2^2+\frac{4\widehat\kappa_n}{n(1-\varphi)}\norm{f}_{\mc H_K}^2\right]\notag\\
&=\inf_{f\in\mc H_K}\left\{\frac4\varphi Q_n^M(m-f)+\frac{4\kappa_n}{n(1-\varphi)}\norm{f}_{\mc H_K}^2\right\}.\label{equation:proof-matched-population-oracle}
\end{align}
The last equality uses $\kappa_n=E[\widehat\kappa_n]$ and the definition of $Q_n^M$.
This proves the theorem.
\end{proof}

\begin{proof}[Proof of Theorem~\ref{theorem:matched-low-order-gsw}]
Write $g=m_{\mc G_1}$ and $r=m-g$.
For any square-integrable $h$, set $\mu_h=E[h(\psi)]$.
For every pair $(i,j)\in M$, the inequality $(a-b)^2\leq2a^2+2b^2$ gives $(h(\psi_i)-h(\psi_j))^2\leq2(h(\psi_i)-\mu_h)^2+2(h(\psi_j)-\mu_h)^2$.
Since every unit belongs to exactly one pair, summing this inequality and taking expectations gives
\begin{align}
Q_n^M(h)&=\frac1nE\bigg[\sum_{(i,j)\in M}(h(\psi_i)-h(\psi_j))^2\bigg]\notag \leq\frac2nE\bigg[\sum_{i=1}^n(h(\psi_i)-\mu_h)^2\bigg]=2\var(h(\psi)).
\end{align}
Consequently, for every $f\in\mc H_{K_{n,1}}$,
\begin{equation}\label{equation:proof-pair-energy-split}
Q_n^M(m-f)\leq2Q_n^M(r)+2Q_n^M(g-f)\leq2Q_n^M(r)+4\var(g(\psi)-f(\psi)).
\end{equation}
The first inequality applies the same two-term bound to $m-f=r+(g-f)$ inside each pair difference.
The second applies the preceding bound with $h=g-f$.

Apply Theorem~\ref{thm:matched-gsw} with $K=K_{n,1}$ and use Equation~\eqref{equation:proof-pair-energy-split}.
Set $\lambda_n=\varphi\kappa_n/(4n(1-\varphi))$ and $\overline\lambda_n=\max\{\lambda_n,n^{-1}\}$.
We obtain
\begin{equation}\label{equation:proof-matched-low-order-oracle}
\mc V_n(\sigma_n,P)\leq\frac8\varphi Q_n^M(r)+\frac{16}\varphi\inf_{f\in\mc H_{K_{n,1}}}\left\{\var(g(\psi)-f(\psi))+\lambda_n\norm{f}_{\mc H_{K_{n,1}}}^2\right\}.
\end{equation}
Since $0\leq\kappa_n\leq4$ and $\varphi$ is fixed, $n^{-1}\leq\overline\lambda_n\leq C_\varphi n^{-1}$ for a constant $C_\varphi$ depending only on $\varphi$.
Let $c_0$ and $n_0$ be the constants in Lemma~\ref{lemma:proof-low-order-matern-approximation}, and set $b_n=d\log n/n$.
If $n<n_0$ or $C_\varphi b_n>c_0$, Equation~\eqref{equation:proof-matched-low-order-oracle} evaluated at zero and projection contraction give $\mc V_n(\sigma_n,P)\leq8Q_n^M(r)/\varphi+16/\varphi$.
In this case, $b_n$ is bounded below by a positive constant depending only on $(n_0,c_0,\varphi)$, so enlarging $C$ proves Equation~\eqref{equation:matched-low-order-gsw}.
It remains to consider $n\geq n_0$ and $C_\varphi b_n\leq c_0$.
We then have $d\log n\leq c_0n$ and $d\overline\lambda_n\log n\leq c_0$.
Together with $\overline\lambda_n\geq n^{-1}$, these are exactly the lemma's penalty conditions for $k=1$ and $J$ empty.
Projection contraction gives $E[g(\psi)^2]\leq E[m(\psi)^2]\leq1$.
Because $\var(g(\psi)-f(\psi))\leq E[(g(\psi)-f(\psi))^2]$ and the approximation functional is nondecreasing in its penalty, Equation~\eqref{equation:proof-matched-low-order-oracle} and Lemma~\ref{lemma:proof-low-order-matern-approximation}, applied with $k=1$, $J$ empty, $g=m_{\mc G_1}$, and penalty $\overline\lambda_n$, give
\begin{align*}
\mc V_n(\sigma_n,P)&\leq\frac8\varphi Q_n^M(r)+\frac{16}\varphi\inf_{f\in\mc H_{K_{n,1}}}\left\{E[(g(\psi)-f(\psi))^2]+\overline\lambda_n\norm{f}_{\mc H_{K_{n,1}}}^2\right\}\\
&\leq C\bigg(Q_n^M(r)+(d\overline\lambda_n\log n)^{(2s)\wedge1}\bigg).
\end{align*}
Finally, $\overline\lambda_n\leq C_\varphi/n$ bounds the last line by the right hand side of Equation~\eqref{equation:matched-low-order-gsw}.
This proves the theorem.

For the H\"older specialization, additionally suppose $d\geq3$, $M$ is 2-swap-stable for a cost satisfying Assumption~\ref{assumption:matching-cost}, and $r$ satisfies $|r(x)-r(y)|\leq L_r\norm{x-y}_2^\beta$.
Lemma~\ref{lemma:proof-stable-pair-distance} gives
\begin{equation}\label{equation:proof-holder-residual-energy}
Q_n^M(r)\leq\frac{L_r^2}{n}E\left[\sum_{p=1}^{n/2}\norm{\psi_{i_p}-\psi_{j_p}}_2^{2\beta}\right]\leq6L_r^2(44d\Lambda^2)^{2\beta}n^{-2\beta/d}.
\end{equation}
Substitution into Equation~\eqref{equation:matched-low-order-gsw} gives the displayed bound following the theorem.
\end{proof}

\subsection{Proofs for Appendix~\ref{appendix:matched-gsw-inference}}\label{appendix:proofs-matched-gsw-inference}

\begin{proof}[Proof of Theorem~\ref{theorem:matched-gsw-inference}]
Sign symmetry gives $\E[T_p|\mathcal W_{n,J}]=0$.
Moreover, $T_pT_q$ is $\{-1,1\}$-valued with conditional mean $R_{pq}$, so $P(T_pT_q=s|\mathcal W_{n,J})=(1+sR_{pq})/2$ for $s\in\{-1,1\}$.
For $p\neq q$, let $E_s=\{T_pT_q=s\}$.
Both the conditional assignment law and $E_s$ are invariant under the global sign reversal $T\mapsto-T$.
Hence $\E[T_p|E_s,\mathcal W_{n,J}]=\E[-T_p|E_s,\mathcal W_{n,J}]$, so $\E[T_p|E_s,\mathcal W_{n,J}]=0$, and the same argument applies to $T_q$.
Since $C_pC_q=A_pA_q+A_pT_qB_q+A_qT_pB_p+T_pT_qB_pB_q$, it follows that $\E[C_pC_q|E_s,\mathcal W_{n,J}]=A_pA_q+sB_pB_q$.
For fixed $p\neq q$, abbreviate $S=T_pT_q$, $r=R_{pq}$, $\ell=L_{pq}$, $a=A_pA_q$, and $b=B_pB_q$, and let $H_{pq}=(Sr+\ell)C_pC_q/(1+Sr)$ denote the corresponding summand in Equation~\eqref{equation:matched-inference-completion}.
Iterated expectations give
\begin{equation*}
\begin{aligned}
\E[H_{pq}|\mathcal W_{n,J}]&=\sum_{s=\pm1}\frac{1+sr}{2}\frac{sr+\ell}{1+sr}(a+sb)=\frac12\sum_{s=\pm1}(sr+\ell)(a+sb)\\
&=\frac12\bigl((r+\ell)(a+b)+(\ell-r)(a-b)\bigr)=\ell a+rb.
\end{aligned}
\end{equation*}
Thus, the factor $1+sr$ in the conditional probability cancels the denominator in the correction, leaving $\E[H_{pq}|\mathcal W_{n,J}]=L_{pq}A_pA_q+R_{pq}B_pB_q$.
On the diagonal, $\E[C_p^2|\mathcal W_{n,J}]=A_p^2+B_p^2=L_{pp}A_p^2+R_{pp}B_p^2$ because $\E[T_p|\mathcal W_{n,J}]=0$ and $L_{pp}=R_{pp}=1$.
Summing over $p$ and $q$ gives $2(A'LA+B'RB)/N$.
Equation~\eqref{equation:matched-inference-bank} and positive semidefiniteness of $L$ prove the conditional identity in the theorem.
Since $\E[\widehat\theta|\mathcal W_{n,J}]=\SATE_n$ is $\mathcal W_n$-measurable, the conditional law of total variance gives $\var(\widehat\theta|\mathcal W_n)=\E[\var(\widehat\theta|\mathcal W_{n,J})|\mathcal W_n]$.
Taking conditional expectations of that identity given $\mathcal W_n$ and using positive semidefiniteness of $L$ proves Equation~\eqref{equation:matched-inference-expectation}.
\end{proof}

\begin{proof}[Proof of Proposition~\ref{theorem:matched-gsw-ate-calibration}]
Write $c(\psi)=\E[\tau|\psi]$ and $v(\psi)=\var(\tau|\psi)$.
Decompose
\begin{equation}\label{equation:matched-inference-population-slack-decomposition}
\frac2NA'L_{\mathrm{pop}}A=\frac1N\sum_{p=1}^NA_p^2-\frac2N\sum_{\{p,q\}\in\Pi}A_pA_q+\frac1{N-1}\sum_{p=1}^N(A_p-\bar A)^2.
\end{equation}
Let $\mathcal P_{2,n}=\{\{i_p,j_p\}:p\in[N]\}$ be the collection of original pairs and let $\mathcal P_{4,n}=\{\{i_p,j_p,i_q,j_q\}:\{p,q\}\in\Pi\}$ be the collection of four-unit groups induced by the outer matching.
Writing each inner sum below over ordered pairs, the first term in Equation~\eqref{equation:matched-inference-population-slack-decomposition} equals $(2n)^{-1}(\sum_i\tau_i^2+\sum_{G\in\mathcal P_{2,n}}\sum_{k\neq\ell\in G}\tau_k\tau_\ell)$.
The cross-product magnitude $2N^{-1}\sum_{\{p,q\}\in\Pi}A_pA_q$ equals $(2n)^{-1}(\sum_{G\in\mathcal P_{4,n}}\sum_{k\neq\ell\in G}\tau_k\tau_\ell-\sum_{G\in\mathcal P_{2,n}}\sum_{k\neq\ell\in G}\tau_k\tau_\ell)$.

We analyze these sums using the bilinear-form lemma of \citet[Lemma C.11]{cytrynbaum2022local}.
For $\mathcal P_{2,n}$, the lemma's hypothesis is equivalent to the tight-matching condition in Equation~\eqref{equation:tight-matching}.
For $\mathcal P_{4,n}$, the matching-for-unions result of \citet[Lemma D.23]{cytrynbaum2022local} establishes the same hypothesis.
Apply Lemma C.11 first with $K=2$ and $\mathcal P_n=\mathcal P_{2,n}$, then with $K=4$ and $\mathcal P_n=\mathcal P_{4,n}$.
In both cases, set the lemma's generic variables equal to $\tau_i$ and sampling propensity equal to one.
It gives $n^{-1}\sum_{G\in\mathcal P_{2,n}}\sum_{k\neq\ell\in G}\tau_k\tau_\ell\convp\E[c(\psi)^2]$ and $n^{-1}\sum_{G\in\mathcal P_{4,n}}\sum_{k\neq\ell\in G}\tau_k\tau_\ell\convp 3\E[c(\psi)^2]$.

The law of large numbers also implies $n^{-1}\sum_i\tau_i^2\convp\E[\tau^2]$ and $\bar A=n^{-1}\sum_i\tau_i\convp\E[\tau]$.
Consequently, the first term in Equation~\eqref{equation:matched-inference-population-slack-decomposition} and the cross-product magnitude converge to $(\E[\tau^2]+\E[c(\psi)^2])/2$ and $\E[c(\psi)^2]$, respectively.
Their difference is $N^{-1}A'L_\Pi A\convp\E[v(\psi)]/2$.
For the third term, $(N-1)^{-1}\sum_p(A_p-\bar A)^2=N(N-1)^{-1}(N^{-1}\sum_pA_p^2-\bar A^2)\convp\var(c(\psi))+\E[v(\psi)]/2$.
Since $L_{\mathrm{pop}}=(L_\Pi+L_0)/2$, adding the limits gives $2A'L_{\mathrm{pop}}A/N\convp\var(\tau)$ by the law of total variance.
\end{proof}

\subsection{H\"older-to-Sobolev Transfer}\label{appendix:holder-sobolev}

For $0<\beta\leq1$, recall $[f]_\beta=\sup_{x\neq y}|f(x)-f(y)|/\norm{x-y}_2^\beta$ and let $\mc P_{d,\beta}^{\mathrm{Hol}}$ contain the laws for which $\norm{m}_\infty+[m]_\beta\leq1$.
We use the following embedding lemma to transfer Theorem~\ref{theorem:sobolev-gsw} to these classes, which shows that a $\beta$-H\"older outcome model can be treated as a Sobolev model of any order $s<\beta$.

\begin{lem}[H\"older-to-Sobolev Embedding]\label{lemma:holder-sobolev}
Fix $d\geq1$ and $0<\beta<1$.
There is a constant $C_{d,\beta}<\infty$ such that every bounded $f:[0,1]^d\to\mathbb R$ with $[f]_\beta<\infty$ and every $s\in[\beta/2,\beta)$ admit $F_s\in H^s(\mr^d)$ satisfying $F_s=f$ on $[0,1]^d$ and $\norm{F_s}_{H^s(\mr^d)}^2\leq C_{d,\beta}(\beta-s)^{-1}(\norm{f}_\infty+[f]_\beta)^2$.
Every bounded Lipschitz function $f:[0,1]^d\to\mathbb R$ admits $F_1\in H^1(\mr^d)$ satisfying $F_1=f$ on $[0,1]^d$ and $\norm{F_1}_{H^1(\mr^d)}\leq C_d(\norm{f}_\infty+[f]_1)$.
\end{lem}

Choosing $s$ just below $\beta$ trades the norm inflation $(\beta-s)^{-1}$ against the rate in Theorem~\ref{theorem:sobolev-gsw}.
Taking $s=\beta-1/\log n$ limits this inflation to a logarithmic factor while preserving the exponent $2\beta/d$ up to constants.
When $\beta=1$, the endpoint embedding applies directly without this additional logarithmic loss.
The following corollary records the resulting guarantee.

% PROOF-ID: holder-gsw-upper
\begin{cor}[H\"older Adaptation]\label{corollary:holder-gsw}
Fix $d\geq3$, $0<\beta\leq1$, $c>0$, and $\varphi\in(0,1)$.
Let $\sigma_n$ be the design in Theorem~\ref{theorem:sobolev-gsw}.
There is a constant $C$ depending only on $(d,\beta,c,\varphi,\overline p)$ such that, for all sufficiently large $n$,
\begin{equation}\label{equation:holder-gsw}
\sup_{P\in\mc P_{d,\beta}^{\mathrm{Hol}}} \left\{ n\var_{P,\sigma_n}(\widehat\theta)-V^*(P) \right\} \leq C\begin{cases}
\log n\left(\dfrac n{\log n}\right)^{-2\beta/d},&0<\beta<1,\\[6pt]
\left(\dfrac n{\log n}\right)^{-2/d},&\beta=1.
\end{cases}
\end{equation}
\end{cor}

%% ========================= Secondary Online Appendix ========================
%% Included in the working paper for review. Comment out the input below when
%% compiling the journal-submission version.
\clearpage

\hypersetup{pageanchor=false}
\pagenumbering{arabic}\renewcommand{\thepage}{\arabic{page}}

\begin{center}
{\Large Secondary Online Appendix to ``The Limits of Experimental Design: Covariate Balance Beyond Low Dimension''}
\vskip 24pt
{\large Max Cytrynbaum}
\end{center}

This secondary online appendix is not intended for publication.

\setcounter{thm}{0}
\setcounter{equation}{0}

\section{Additional Results and Proofs}

\subsection{Additional Results for Section~\ref*{section:structure}}

\emph{Mat\'ern RKHS.}
Consider the Mat\'ern kernel.
For $d\geq1$ and $\nu>0$, put $r=d/2+\nu$ and define the spectral density $\rho_{d,\nu}(\omega)=\pi^{-d/2}\Gamma(r)\Gamma(\nu)^{-1}(1+\norm{\omega}_2^2)^{-r}$.
Then
\begin{equation}\label{equation:proof-matern-spectral-convention}
W_d^{\mathrm{Mat},\nu}(x,x')=\int_{\mathbb R^d}e^{i\omega'(x-x')}\rho_{d,\nu}(\omega)d\omega.
\end{equation}
This representation follows, for example, by setting $\ell=\sqrt{2\nu}$ in the Mat\'ern spectral density of \citet[Equation~(4.15)]{rasmussen2006gaussian}, substituting it into their Equation~(4.6), and changing variables.
A calculation shows that $W_d^{\mathrm{Mat},\nu}(x,x)=\int_{\mathbb R^d}\rho_{d,\nu}(\omega)d\omega=1$.
Because $\rho_{d,\nu}$ is even, Equation~\eqref{equation:proof-matern-spectral-convention} also shows that the kernel is real-valued and symmetric.
We next verify that this kernel is positive definite. For distinct $x_1,\ldots,x_N\in\mathbb R^d$ and nonzero $c\in\mathbb R^N$, Equation~\eqref{equation:proof-matern-spectral-convention} gives
\begin{equation}\label{equation:proof-matern-positive-definite}
\sum_{i,j=1}^Nc_ic_jW_d^{\mathrm{Mat},\nu}(x_i,x_j)=\int_{\mathbb R^d}\left|\sum_{i=1}^Nc_ie^{i\omega'x_i}\right|^2\rho_{d,\nu}(\omega)d\omega>0.
\end{equation}
The strict inequality follows because $\rho_{d,\nu}(\omega)>0$ and the distinct complex exponentials are linearly independent. Thus, $W_d^{\mathrm{Mat},\nu}$ is positive definite.

Let $\mc H_W(\mathbb R^d)$ denote the RKHS generated by $W$ on $\mathbb R^d$.
Write $\mathcal F$ and $\mathcal F^{-1}$ for the Fourier transform and inverse Fourier transform under the convention $(\mathcal F h)(\omega)=\widehat h(\omega)=(2\pi)^{-d/2}\int_{\mathbb R^d}e^{-i\omega'x}h(x)dx$.
We invoke Theorem~10.12 of \citet{wendland2005scattered}, which characterizes the RKHS of any real-valued, symmetric, positive definite kernel $W(x,x')=\kappa(x-x')$ with $\kappa\in C(\mathbb R^d)\cap L^1(\mathbb R^d)$ through $\widehat\kappa$. In particular, when $\widehat\kappa>0$, \citet{wendland2005scattered} shows that $\norm{h}_{\mc H_W(\mathbb R^d)}^2=(2\pi)^{-d/2}\int_{\mathbb R^d}|\widehat h(\omega)|^2/\widehat\kappa(\omega)d\omega$ and
\begin{equation}\label{equation:proof-wendland-characterization}
\mc H_W(\mathbb R^d)=\bigg\{h\in C(\mathbb R^d)\cap L^2(\mathbb R^d):\norm{h}_{\mc H_W(\mathbb R^d)}<\infty\bigg\}.
\end{equation}
The next lemma specializes this characterization to the Mat\'ern kernel.

\begin{lem}[Mat\'ern RKHS Norm]\label{lemma:proof-matern-native-norm}
Let $W=W_d^{\mathrm{Mat},\nu}$ and $\Lambda_{r,d}=(4\pi)^{-d/2}\Gamma(\nu)/\Gamma(r)$ for $r=d/2+\nu$.
Define $\Theta_{r,d}=d^r\Lambda_{r,d}$.
Then $\mc H_W(\mathbb R^d)=H^r(\mathbb R^d)$.
For $h\in\mc H_W(\mathbb R^d)$, the RKHS norm is
\begin{equation}\label{equation:proof-matern-native-norm}
\norm{h}_{\mc H_W(\mathbb R^d)}^2=\Lambda_{r,d}\int_{\mathbb R^d}|\widehat h(\omega)|^2(1+\norm{\omega}_2^2)^r d\omega.
\end{equation}
Moreover, for every $\bar\nu>0$, there is $C_{\bar\nu}<\infty$ such that, for all $d\geq1$ and $0<\nu\leq\bar\nu$, $\Theta_{d/2+\nu,d}\leq C_{\bar\nu}/\nu$.
\end{lem}

\begin{proof}
We have $\kappa(u)=W(u,0)$.
Positive definiteness was established in Equation~\eqref{equation:proof-matern-positive-definite}. It remains to verify that $\kappa\in C(\mathbb R^d)\cap L^1(\mathbb R^d)$ and $\widehat\kappa>0$. Since $r>d/2$, $\rho_{d,\nu}\in L^1(\mathbb R^d)$, and Equation~\eqref{equation:proof-matern-spectral-convention} gives $\kappa=\mathcal F^{-1}[(2\pi)^{d/2}\rho_{d,\nu}]$. Thus, $\kappa\in C(\mathbb R^d)$ as the inverse Fourier transform of an $L^1$ function. Equation~(4.14) of \citet{rasmussen2006gaussian} with length scale $\sqrt{2\nu}$ gives $\kappa(u)=2^{1-\nu}\Gamma(\nu)^{-1}\norm{u}_2^\nu K_\nu(\norm{u}_2)$, where $K_\nu$ is the modified Bessel function of the second kind. This expression is bounded near zero and decays exponentially as $\norm{u}_2\to\infty$, so $\kappa\in L^1(\mathbb R^d)$. Since $\rho_{d,\nu}$ is continuous, Fourier inversion gives $\widehat\kappa=\mathcal F\kappa=\mathcal F\left(\mathcal F^{-1}[(2\pi)^{d/2}\rho_{d,\nu}]\right)=(2\pi)^{d/2}\rho_{d,\nu}>0$. Hence Equation~\eqref{equation:proof-wendland-characterization} applies. Substituting $\widehat\kappa$ into the formula above gives for $h\in\mc H_W(\mathbb R^d)$,
\begin{align*}
\norm{h}_{\mc H_W(\mathbb R^d)}^2&=(2\pi)^{-d/2}\int_{\mathbb R^d}\frac{|\widehat h(\omega)|^2}{\widehat\kappa(\omega)}d\omega =(4\pi)^{-d/2}\frac{\Gamma(\nu)}{\Gamma(r)}\int_{\mathbb R^d}|\widehat h(\omega)|^2(1+\norm{\omega}_2^2)^r d\omega.
\end{align*}
This proves Equation~\eqref{equation:proof-matern-native-norm}. It remains to identify $\mc H_W(\mathbb R^d)=H^r(\mathbb R^d)$. By Equation~\eqref{equation:sobolev-spectral-norm}, $H^r(\mathbb R^d)=\{h\in L^2(\mathbb R^d):\int_{\mathbb R^d}|\widehat h(\omega)|^2(1+\norm{\omega}_2^2)^r d\omega<\infty\}$. Moreover, since $r>d/2$, the Sobolev embedding theorem gives $H^r(\mathbb R^d)\subseteq C(\mathbb R^d)$ \citep{taylor2011partial}. Comparing these facts with Equation~\eqref{equation:proof-wendland-characterization} proves the identification and, together with Equation~\eqref{equation:proof-matern-native-norm}, Proposition~\ref{proposition:matern-sobolev-equivalence}.

For the final statement, continuity of $\Gamma$ on $[1,1+\bar\nu]$ gives $M_{\bar\nu}\equiv\sup_{0\leq t\leq\bar\nu}\Gamma(1+t)<\infty$, so $\Gamma(\nu)=\Gamma(1+\nu)/\nu\leq M_{\bar\nu}/\nu$. For $d\geq3$, monotonicity of $\Gamma$ on $[3/2,\infty)$ and Stirling's lower bound give
\begin{equation*}
\Theta_{d/2+\nu,d}\leq\frac{M_{\bar\nu}}\nu d^{\bar\nu}\frac{(d/(4\pi))^{d/2}}{\Gamma(d/2)}\leq\frac C\nu d^{\bar\nu+1/2}\left(\frac e{2\pi}\right)^{d/2}\leq\frac{C_{\bar\nu}}\nu.
\end{equation*}
For $d=1,2$, $\nu\Theta_{d/2+\nu,d}$ extends continuously to $0\leq\nu\leq\bar\nu$ and is therefore bounded. This proves the final statement.
\end{proof}

\begin{proof}[Proof of Proposition~\ref{proposition:kernel-oracle}]
In the notation of Lemma~\ref{lemma:proof-gsw-ridge}, let $\mathcal H=\mc H_K$, $x_i=K(\psi_i,\cdot)$, and $F=F_n$, where $F_n:\mathcal H\to\mathbb R^n$ is the evaluation operator $(F_ng)_i=g(\psi_i)$.
First, we identify the adjoint $F_n^*:\mathbb R^n\to\mathcal H$.
Fix $u\in\mathbb R^n$ and any $g\in\mathcal H$.
By the reproducing property, $\langle g,F_n^*u\rangle_{\mathcal H}=\langle F_ng,u\rangle_2=\sum_{i=1}^nu_ig(\psi_i)=\langle g,\sum_{i=1}^nu_iK(\psi_i,\cdot)\rangle_{\mathcal H}$.
This holds for every $g\in\mathcal H$, so $F_n^*u=\sum_{i=1}^nu_iK(\psi_i,\cdot)$.
Next, we claim that $F_nF_n^*=K_n$.
Taking $u=e_j$ gives $(F_nF_n^*)_{ij}=(F_n^*e_j)(\psi_i)=K(\psi_j,\psi_i)=K(\psi_i,\psi_j)$, proving the claim.

Let $m_{1:n}=(m(\psi_1),\ldots,m(\psi_n))'$.
Together with Equation~\eqref{equation:proof-section4-gsw-covariance}, the identity $F_nF_n^*=K_n$ gives $E[Z|\mathcal F_n]=0$ and $\cov(Z|\mathcal F_n)\preceq(\varphi I_n+(1-\varphi)F_nF_n^*/2)^{-1}$. Applying Lemma~\ref{lemma:proof-gsw-ridge} with $\mathcal A=\mathcal F_n$, $S=Z$, $y=m_{1:n}$, $a=\varphi$, and $b=(1-\varphi)/2$ therefore gives, for every fixed $g\in\mc H_K$,
\begin{align}
E\bigg[\bigg(\sum_{i=1}^nZ_im(\psi_i)\bigg)^2\big|\mathcal F_n\bigg]&\leq\min_{h\in\mc H_K}\bigg\{\frac1\varphi\norm{m_{1:n}-F_nh}_2^2+\frac2{1-\varphi}\norm{h}_{\mc H_K}^2\bigg\}\nonumber\\
&\leq\frac1\varphi\sum_{i=1}^n(m(\psi_i)-g(\psi_i))^2+\frac2{1-\varphi}\norm{g}_{\mc H_K}^2.\label{equation:proof-kernel-conditional-oracle}
\end{align}
The second inequality evaluates the minimum at $h=g$.
Take expectations in Equation~\eqref{equation:proof-kernel-conditional-oracle}, multiply by $4/n$, and apply Proposition~\ref{proposition:variance-decomposition}. The resulting inequality holds for every fixed $g\in\mc H_K$. Taking the infimum over $g$ proves Equation~\eqref{equation:kernel-oracle}.
\end{proof}

\begin{lem}[Sobolev Variance Calibration]\label{lemma:proof-sobolev-variance-calibration}
Fix $s>0$ and let $U_d$ be uniform on $[0,1]^d$.
There is a constant $c_s>0$ depending only on $s$ such that, for every $d\geq1$,
\begin{equation*}
c_s\leq\sup_{\norm{m}_{H_{\mathrm{av}}^s(\mathbb R^d)}\leq1}\var(m(U_d))\leq1.
\end{equation*}
\end{lem}

\begin{proof}
For the upper bound, every $m$ in the unit ball satisfies
\begin{equation*}
\var(m(U_d))\leq\int_{[0,1]^d}m(x)^2dx\leq\norm{m}_{L^2(\mathbb R^d)}^2\leq\norm{m}_{H_{\mathrm{av}}^s(\mathbb R^d)}^2\leq1.
\end{equation*}
For the lower bound, use the function $\eta$ from the proof of Theorem~\ref{theorem:sobolev-lower} and set $g_d(x)=\Pi_{\ell=1}^d\eta(x_\ell)$.
Because $\int_0^1\eta(u)du=0$ and $\int_0^1\eta(u)^2du=1$, we have $E[g_d(U_d)]=0$ and $\var(g_d(U_d))=1$.
The calculation in Equations~\eqref{equation:proof-sobolev-bump-scaling}--\eqref{equation:proof-sobolev-bump-moment}, specialized to $h=1$, gives $\norm{g_d}_{H_{\mathrm{av}}^s(\mathbb R^d)}^2\leq C_s$ uniformly over $d\geq1$.
Therefore, $m_d=C_s^{-1/2}g_d$ belongs to the unit ball and satisfies $\var(m_d(U_d))=C_s^{-1}$.
The result follows with $c_s=C_s^{-1}$.
\end{proof}

\subsection{Proofs for Appendix~\ref*{appendix:holder-sobolev}}

\begin{proof}[Proof of Lemma~\ref{lemma:holder-sobolev}]
Write $Q=[0,1]^d$, $M=\norm{f}_\infty$, and $L=[f]_\beta$.
Because $0<\beta\leq1$, $\rho_\beta(x,y)=\norm{x-y}_2^\beta$ is a metric on $\mr^d$.
Define the McShane extension
\begin{equation}\label{equation:holder-extension}
\tilde f(x)=\inf_{z\in Q}\bigl(f(z)+L\norm{x-z}_2^\beta\bigr).
\end{equation}
For $x\in Q$, choosing $z=x$ gives $\tilde f(x)\leq f(x)$, while H\"older continuity gives $f(z)+L\norm{x-z}_2^\beta\geq f(x)$ for every $z\in Q$.
Thus $\tilde f=f$ on $Q$.
The triangle inequality for $\rho_\beta$ also gives $[\tilde f]_\beta\leq L$.
Set $\bar f(x)=(-M)\vee(\tilde f(x)\wedge M)$.
Because truncation is nonexpansive, $\bar f=f$ on $Q$, $\norm{\bar f}_\infty\leq M$, and $[\bar f]_\beta\leq L$.

Fix a smooth function $\chi$ with compact support such that $0\leq\chi\leq1$ and $\chi=1$ on $Q$, and define $F=\chi\bar f$.
Then $F=f$ on $Q$ and $F$ has support in a fixed bounded set $D$.
If $\norm{x-y}_2\leq1$, then $|F(x)-F(y)|\leq|\bar f(x)-\bar f(y)|+M|\chi(x)-\chi(y)|\leq(L+M\norm{\nabla\chi}_\infty)\norm{x-y}_2^\beta$.
If $\norm{x-y}_2>1$, then $|F(x)-F(y)|\leq2M\leq2M\norm{x-y}_2^\beta$.
Consequently, $\norm{F}_\infty+[F]_\beta\leq C_d(M+L)$.
Let $A=M+L$ and fix $s\in[\beta/2,\beta)$.

The whole-space increment characterization of $H^s(\mr^d)$ \citep{adams2003sobolev} applies uniformly over this range of $s$.
If $F(x)-F(y)\neq0$, at least one of $x$ and $y$ belongs to $D$.
Since the integrand below is symmetric in $x$ and $y$, we may therefore restrict one integral to $D$ at the cost of a factor of two:
\begin{align*}
\norm{F}_{H^s(\mr^d)}^2&\leq C_{d,\beta}\left(\norm{F}_{L^2(\mr^d)}^2+\iint_{\mr^d\times\mr^d}\frac{|F(x)-F(y)|^2}{\norm{x-y}_2^{d+2s}}dxdy\right)\\
&\leq C_{d,\beta}\left(\norm{F}_{L^2(\mr^d)}^2+2\int_D\int_{\mr^d}\frac{|F(x)-F(y)|^2}{\norm{x-y}_2^{d+2s}}dydx\right).
\end{align*}
A calculation using the increment bounds above and polar coordinates gives
\begin{align*}
\norm{F}_{H^s(\mr^d)}^2&\leq C_{d,\beta}A^2\left(1+\int_0^1r^{2(\beta-s)-1}dr+\int_1^\infty r^{-2s-1}dr\right)\\
&=C_{d,\beta}A^2\left(1+\frac1{2(\beta-s)}+\frac1{2s}\right)\leq\frac{C_{d,\beta}}{\beta-s}A^2.
\end{align*}
Setting $F_s=F$ proves the stated bound.

For a bounded Lipschitz $f$, repeat the same construction with $\beta=1$ to obtain a compactly supported Lipschitz extension $F_1$ satisfying $\norm{F_1}_\infty+[F_1]_1\leq C_d(\norm{f}_\infty+[f]_1)$.
Rademacher's theorem bounds $\norm{\nabla F_1}_\infty$ by the Lipschitz coefficient of $F_1$.
Because $F_1$ has support in the fixed bounded set $D$, Equation~\eqref{equation:sobolev-whole-space-norm} gives $\norm{F_1}_{H^1(\mr^d)}\leq C_d(\norm{f}_\infty+[f]_1)$, proving the final claim.
\end{proof}

\begin{proof}[Proof of Corollary~\ref{corollary:holder-gsw}]
Suppose first that $0<\beta<1$ and set $s_n=\beta-1/\log n$.
For all sufficiently large $n$, $s_n\in[\beta/2,\beta)$.
Because $d\geq3$ and $\beta<1$, we have $s_n<1<d/2$, so only the subcritical calculation in Lemma~\ref{lemma:proof-penalized-matern-approximation} is required.
Its constants are uniform over $s_n\in[\beta/2,\beta)$ because $s_n$ is bounded away from zero and the $s_n$-dependent exponential factor in its proof is bounded by $e^{4\beta c}$.
Set $\lambda_n=\max\{2\varphi/((1-\varphi)n),n^{-1}\}$.
For fixed $\varphi$, $n^{-1}\leq\lambda_n\leq C_\varphi/n$, so $\lambda_n\leq c_0/\log n$ for all sufficiently large $n$ and the penalty condition in Lemma~\ref{lemma:proof-penalized-matern-approximation} holds.
Fix $P\in\mc P_{d,\beta}^{\mathrm{Hol}}$ and let $F_{s_n}$ be the extension of $m$ supplied by Lemma~\ref{lemma:holder-sobolev}.
Since $F_{s_n}=m$ on the unit cube and the dimension-normalized spectral weight is no larger than the ordinary Sobolev weight, the oracle calculation in the proof of Theorem~\ref{theorem:sobolev-gsw} and Lemma~\ref{lemma:proof-penalized-matern-approximation} give
\begin{align*}
\norm{F_{s_n}}_{H_{\mathrm{av}}^{s_n}(\mr^d)}^2&\leq\norm{F_{s_n}}_{H^{s_n}(\mr^d)}^2\leq\frac{C_{d,\beta}}{\beta-s_n}(\norm{m}_\infty+[m]_\beta)^2\leq C_{d,\beta}\log n,\\
\mc V_n(P)&\leq C\norm{F_{s_n}}_{H_{\mathrm{av}}^{s_n}(\mr^d)}^2(\lambda_n\log n)^{2s_n/d}\leq C\log n\left(\frac{\log n}{n}\right)^{2s_n/d}.
\end{align*}
Moreover,
\begin{equation}\label{equation:moving-holder-exponent}
\left(\frac n{\log n}\right)^{-2s_n/d} \leq e^{2/d}\left(\frac n{\log n}\right)^{-2\beta/d}.
\end{equation}
This proves the first branch uniformly over $P\in\mc P_{d,\beta}^{\mathrm{Hol}}$.
When $\beta=1$, the Lipschitz conclusion of Lemma~\ref{lemma:holder-sobolev} supplies an extension $F_1$ whose $H^1(\mr^d)$ norm is bounded by a constant depending only on $d$.
The spectral comparison above also bounds its $H_{\mathrm{av}}^1(\mr^d)$ norm.
Since $1<d/2$, the same oracle and approximation calculation at $s=1$ gives $\mc V_n(P)\leq C(\log n/n)^{2/d}$, proving the second branch.
\end{proof}

\subsection{Simulation Designs}\label{appendix:simulation-designs}

This subsection records the models and Monte Carlo protocols for Figures~\ref{figure:matching-reversal-simulation}--\ref{figure:structured-balance-simulation}.
The calibrated empirical populations in Figure~\ref{figure:empirical-application} are described in Section~\ref{section:empirical-application}.
At every parameter configuration, we average over 100 independent covariate samples $\psi_{1:n}^{(r)}$, $r\in[100]$, drawn from the uniform distribution on the relevant unit cube.
Within each sample, design expectations that are not available in closed form are approximated using 100 assignment vectors.
The matched designs use minimum-total squared-Euclidean matching constructed separately within each covariate sample.

\paragraph{Matching reversal.}
For Figure~\ref{figure:matching-reversal-simulation}, let $D=20$, define $v(x)=\sqrt{12}(x-1/2)$ and $w(x)=\sqrt{2}\cos(2\pi x)$, and set $a_j=0.8^{j-1}/(\sum_{\ell=1}^{D}0.8^{2(\ell-1)})^{1/2}$.
The linear and nonlinear outcome models are $m_{\mathrm{lin}}(\psi)=\sum_{j=1}^{D}a_jv(\psi_j)$ and $m_{\mathrm{add}}(\psi)=\sum_{j=1}^{D}a_jw(\psi_j)$, respectively.
Both have unit variance and the nonlinear main effect is orthogonal to a linear function of its coordinate.
Each design observes only the first $d$ coordinates, although its variance is evaluated under the corresponding full outcome model.
The dotted line is the efficiency bound using the observed coordinates, $V_d=1+4\sum_{j>d}a_j^2$.

Let $b_d(Z)=n^{-1}\sum_{i=1}^nZ_i\psi_{i,1:d}$ and let $\widehat\Sigma_d$ be the sample covariance matrix of the observed covariates.
Linear GSW applies the walk to their Mahalanobis-standardized values and uses $\varphi=0.1$.
Best-of-$M$ rerandomization selects, from $M=10{,}000$ independent randomizations, the allocation minimizing $b_d(Z)'\widehat\Sigma_d^{-1}b_d(Z)$.
Nonparametric GSW uses the additive kernel $K_1$ from Example~\ref{example:structured-priority} with $\nu=1$ and $\varphi=0.03$.
All five designs are evaluated under both outcome models.

\paragraph{Smooth outcomes.}
For Figure~\ref{figure:smoothness-simulation}, let $W_\nu$ be the unit-length-scale Mat\'ern kernel with parameter $\nu$ and let $\psi$ and $\psi'$ be independent and uniform on $[0,1]^d$.
Define $c_\nu=1-\E[W_\nu(\psi,\psi')]$ and let $m_\nu\sim\mathrm{GP}(0,W_\nu/c_\nu)$.
This normalization gives $\E[\var(m_\nu(\psi)|m_\nu)]=1$.
The figure reports prior-averaged excess variance relative to matched pairs.
Equivalently, the plotted value for a design $\sigma$ is $\E_\sigma[Z'W_\nu Z]/\E_{\mathrm{MP}}[Z'W_\nu Z]$.
In the left panel, $n=240$, $d=8$, and $\nu$ ranges over $\{0.1,0.2,0.5,1,2,4,8,16,32\}$.
Fixed Mat\'ern GSW uses $W_1$, while Oracle Mat\'ern GSW and kernel rerandomization use $W_\nu$.
In the right panel, $d=8$, fixed Mat\'ern GSW uses $W_1$, and $n$ ranges over $\{80,120,180,240,360,480,720\}$ for $\nu\in\{0.5,1,4,16\}$.
All GSW designs use $\varphi=0.03$.

\paragraph{Structured balance.}
For Figure~\ref{figure:structured-balance-simulation}, $u_j(x)=\sqrt{2}\cos(\pi x)$, $g_d(\psi)=d^{-1/2}\sum_{j=1}^du_j(\psi_j)$, and $r(\psi)=u_1(\psi_1)u_2(\psi_2)$.
Under the uniform covariate law, $g_d$ and $r$ have unit variance and are orthogonal.
For $\rho\geq0$, the outcome model is
\begin{equation}\label{equation:simulation-structured-model}
m_\rho(\psi)=\frac{g_d(\psi)+\rho r(\psi)}{\sqrt{1+\rho^2}}.
\end{equation}
The left panel uses $\rho=0$, so $m=m_{\mc G_1}$.
The right panel uses $\rho=0.5$, so the omitted interaction contributes one fifth of the predictable variance and $V^*(P)+\Delta_{\mc G_1}(P)=V^*(P)+0.8$.
Full-dimensional GSW uses the Mat\'ern kernel on $[0,1]^d$, main-effects GSW uses the additive kernel $K_1$, and matched main-effects GSW applies the same kernel to pair orientations.
All kernels use $\nu=1$, and all GSW designs use $\varphi=0.03$.

\subsection{The Gram-Schmidt Walk}\label{appendix:gsw}

This subsection gives the formal description of the assignment algorithm used in Section~\ref{subsection:unit-gsw}.
The construction is the Gram-Schmidt walk of \citet{bansal2019gram}, with the zero initialization and randomized pivot order of \citet{harshaw2024gsw} and the Gram-matrix formulation in \citet[Section 2.7]{harshaw2021dissertation}.

Let $\Gamma$ be an $n\times n$ positive semidefinite matrix with diagonal entries at most one.
The coordinates of the walk index units.
Starting from $z^{(0)}=0$, define the active set
\begin{equation}\label{equation:gsw-active-set}
\mathcal A_t = \left\{ i\in[n]:|z_i^{(t)}|<1 \right\}.
\end{equation}
Its elements are the units whose assignments remain fractional.
The algorithm selects a pivot $p_t\in\mathcal A_t$ uniformly at random and retains that pivot until it reaches $-1$ or $1$.
When the current pivot leaves $\mathcal A_t$, the algorithm selects a new pivot uniformly from the remaining active units.
Given $\mathcal A_t$ and $p_t$, choose
\begin{equation}\label{equation:gsw-direction}
u^{(t)} \in \argmin\bigl\{u'\Gamma u:u\in\mathbb R^n, u_{p_t}=1, u_i=0\text{ for every }i\notin\mathcal A_t\bigr\}.
\end{equation}
Let $\delta_t^+$ and $\delta_t^-$ be the largest nonnegative values such that $z^{(t)}+\delta_t^+u^{(t)}$ and $z^{(t)}-\delta_t^-u^{(t)}$ remain in $[-1,1]^n$.
The next assignment is
\begin{equation}\label{equation:gsw-update}
z^{(t+1)} = \begin{cases}
z^{(t)}+\delta_t^+u^{(t)}, & \text{with probability } \dfrac{\delta_t^-}{\delta_t^++\delta_t^-},\\
z^{(t)}-\delta_t^-u^{(t)}, & \text{with probability } \dfrac{\delta_t^+}{\delta_t^++\delta_t^-}.
\end{cases}
\end{equation}
The probabilities in Equation~\eqref{equation:gsw-update} make $E[z^{(t+1)}|z^{(t)},\mathcal A_t,p_t]=z^{(t)}$.
At least one active coordinate reaches the boundary after every update.
The algorithm therefore terminates at an assignment $Z\in\{-1,1\}^n$ after at most $n$ updates.

To connect this formulation to the ordinary vector algorithm, take any factor $C$ satisfying $C'C=\Gamma$ and let $c_i$ be its $i$th column.
Then $\norm{\sum_{i=1}^nu_ic_i}_2^2=u'\Gamma u$, so Equation~\eqref{equation:gsw-direction} is exactly the direction chosen by ordinary GSW applied to the columns of $C$.
The direction can also be computed directly from $\Gamma$, so an explicit factorization and an explicit feature map are not part of the definition.

For the kernelized design, Equation~\eqref{equation:kernel-gsw-gram} supplies the input matrix.
The covariance guarantee $\cov(Z|\psi_{1:n})\preceq\Gamma_n^{-1}\preceq\varphi^{-1}I_n$ prevents balance of the selected kernel features from creating arbitrarily large variance in directions not described by the kernel.
See \citet{harshaw2024gsw} for a detailed design-based interpretation.

\end{document}